\documentclass[authorcolumns,numberwithinsect]{no-lipics-v2022}
\usepackage{nicefrac}
\usepackage{bm}
\usepackage{accents}
\usepackage{trimclip}

\title{Counting Small Induced Subgraphs with~Edge-monotone~Properties}

\author{Simon Döring}{Max Planck Institute for Informatics and\\Saarbrücken Graduate School
of Computer Science, SIC, Saarbrücken, Germany}{sdoering@mpi-inf.mpg.de}{}{}

\author{Dániel Marx}{CISPA Helmholtz Center for Information Security, Saarbrücken, Germany}
{marx@cispa.de}{}{}
\author{Philip Wellnitz}{Max Planck Institute for Informatics, SIC,\\
Saarbrücken, Germany}{wellnitz@mpi-inf.mpg.de}{https://orcid.org/0000-0002-6482-8478}{}

\authorrunning{S. Döring, D. Marx, and P. Wellnitz}

\acknowledgements{\footnotesize We thank Jacob Focke and Marc Roth for insightful discussions. We also thank Rebekka Burkholz for providing us with an Overleaf account.}

\nolinenumbers

\SetKwComment{Comment}{$\triangleright$\ \rm}{}
\SetKwRepeat{Do}{do}{while}

\usetikzlibrary{arrows,arrows.meta,shapes.misc, decorations.pathreplacing,decorations.pathmorphing,calligraphy, patterns.meta}
\tikzset{dotmark/.style={circle,fill,inner sep=1.5pt}}
\tikzset{emptymark/.style={circle,draw,fill=white,inner sep=1.5pt}}
\tikzset{crossmark/.style={thick,inner sep=1.5pt}}

\newcommand{\newhat}{\scalebox{1.5}[.75]{\trimbox{0pt 1.1ex}{\textasciicircum}}}

\newcommand{\stretchedhat}[1]{\accentset{\newhat}{#1}}

\def\fragmentco#1#2{{[}#1\,{.\,.}\,#2{)}}

\def\fragment#1#2{{[}#1\,{.\,.}\,#2{]}}
\def\position#1{{[}#1{]}}
\def\setn#1{{[}#1{]}}

\makeatletter
\renewenvironment{cases}{%
  \matrix@check\cases\env@cases
}{%
  \endarray\right.%
}
\def\env@cases{%
  \let\@ifnextchar\new@ifnextchar
  \left\lbrace
  \def\arraystretch{1.1}%
  \array{@{\;}c@{\quad}l@{}}%
}
\makeatother

\def\mid{\ensuremath :}

\def\emptyset{\varnothing}

\newcommand{\nat}{\mathbb{N}}
\newcommand{\Z}{\mathbb{Z}}

\newcommand{\field}[1]{\smash{\mathbb{F}}_{\smash{#1}\vphantom{q}}}

\def\NUM{\text{\#}}
\def\W#1{\ensuremath {\sf W{\bm{[}\,#1\,\bm{]}}}}
\def\W#1{\ensuremath {\sf W{\bm{[}#1\bm{]}}}}
\def\w{\NUM\W1}

\newcommand{\sylow}{\textup{Syl}}

\newcommand{\push}{\overrightarrow{\varphi}}

\newcommand{\wrGraph}{\circ}

\newcommand{\prodGroup}{\times}
\newcommand{\BigProdGroup}{\bigtimes}

\makeatletter
\def\bigtimes{%
  \DOTSB\mathop{\mathpalette\mattos@bigtimes\relax}\slimits@
}
\newcommand\mattos@bigtimes[2]{%
  \vcenter{\hbox{%
    \sbox\z@{$#1\sum$}%
    \resizebox{!}{0.9\dimexpr\ht\z@+\dp\z@}{\raisebox{\depth}{$\m@th#1\times$}}%
  }}%
  \vphantom{\sum}%
}
\def\bigbowtie{%
  \DOTSB\mathop{\mathpalette\mattos@bigbowtie\relax}\slimits@
}
\newcommand\mattos@bigbowtie[2]{%
  \vcenter{\hbox{%
    \sbox\z@{$#1\sum$}%
    \resizebox{!}{0.9\dimexpr\ht\z@+\dp\z@}{\raisebox{\depth}{$\m@th#1\triangledown$}}%
  }}%
  \vphantom{\sum}%
}

\makeatother

\newcommand{\JoinGraph}{\triangledown}
\DeclareMathOperator*{\BigJoinGraph}{\bigbowtie}

\newcommand{\UnionGraph}{\uplus}

\newcommand{\unionSet}{\uplus}

\DeclareMathOperator{\hw}{hw}
\DeclareMathOperator{\diam}{diam}

\DeclareMathOperator{\IS}{IS}
\DeclareMathOperator{\tw}{tw}

\DeclareMathOperator{\id}{id}
\DeclareMathOperator{\diag}{diag}
\DeclareMathOperator{\aut}{Aut}
\DeclareMathOperator{\fp}{FP}

\DeclareMathOperator{\edgesub}{\mathcal{E}} 
\DeclareMathOperator{\sylelm}{\overline{\varphi}}

\newcommand{\homs}[2]{\mbox{\ensuremath{\mathrm{Hom}(#1 \to #2)}}}

\newcommand{\indsubs}[2]{\mbox{\ensuremath{\mathrm{IndSub}(#1 \to #2)}}}
\newcommand{\auts}[1]{\ensuremath{\mathrm{Aut}(#1)}}

\newcommand{\cphoms}[2]{\ensuremath{\mathrm{cp}\text{-}\mathrm{Hom}}(#1 \to #2)}

\newcommand{\cpindsubs}[2]{\ensuremath{\mathrm{cp}\text{-}\mathrm{IndSub}}(#1 \to #2)}

\newcommand{\clique}{\ensuremath{\textsc{Clique}}}
\newcommand{\homsprob}{\ensuremath{\textsc{Hom}}}
\newcommand{\cphomsprob}{\ensuremath{\textsc{cp-Hom}}}
\newcommand{\indsubsprob}{\ensuremath{\textsc{IndSub}}}
\newcommand{\cpindsubsprob}{\ensuremath{\textsc{cp-IndSub}}}

\DeclareMathOperator{\wrLevel}{\varepsilon}
\DeclareMathOperator{\Hasselevel}{\ell}

\def\fps#1{\ensuremath\fp(\Gamma, #1)}
\def\fpb#1#2{\ensuremath\fp(#1, #2)}

\def\hl#1{\ensuremath \Hasselevel(#1)}

\def\nt#1{\ensuremath M_{#1}}
\def\scat#1{\ensuremath \text{Sc}_{#1}}
\def\con#1{\ensuremath \text{Co}_{#1}}
\def\aename#1{\ensuremath \stretchedhat{#1}}
\def\ae#1#2{\ensuremath \aename{#1}(#2)}
\def\ess#1#2{\ensuremath #1{\{}#2{\}}}

\newcommand{\fpt}{\leq_{\mathrm{T}}^{\mathrm{fpt}}}

\def\od#1{\ensuremath\mathbb{O}(#1)}

\def\phisum#1{\ensuremath w_{#1}}
\def\aesum#1{\ensuremath \widehat{w}_{#1}}

\def\phisumvec{\ensuremath \mathbf{w}}
\def\aesumvec{\ensuremath \widehat{\mathbf{w}}}

\def\phisumvecr#1{\ensuremath \mathbf{w}^{#1}}
\def\aesumvecr#1{\ensuremath \widehat{\mathbf{w}}^{#1}}

\def\sunderset#1#2{\substack{#2\\#1}}

\def\sym#1{\ensuremath \mathfrak{S}_{#1}}

\def\rotgr#1{\ensuremath \smash{\bm{\circlearrowright}}{}_{#1}}
\def\cgr#1#2{\ensuremath \smash{C_{#1}^{#2}}\vphantom{{}_q^d}}

\def\scalesim{\ensuremath \sim}

\newcommand{\graphs}[1]{\mathcal{G}_{#1}} 
\newcommand{\metaGraph}[2]{#1\big\langle#2\big\rangle} 
\newcommand{\oneEdge}[2]{{\{#1, #2\}}}

\declaretheorem[numbered=no,name=Exponential Time Hypothesis (ETH),style=plainqstyle]{ethy}

\begin{document}
\maketitle

\begin{abstract}
    We study the parameterized complexity of \NUM$\indsubsprob(\Phi)$, where given a graph
    $G$ and an integer $k$, the task is to count the number of induced subgraphs on $k$
    vertices that satisfy the graph property $\Phi$.
    Focke and Roth [STOC~2022]
    completely characterized the complexity for each $\Phi$ that is a \emph{hereditary
    property} (that is, closed under vertex deletions): \NUM$\indsubsprob(\Phi)$ is \w-hard
    except in the degenerate cases when every graph satisfies $\Phi$ or only finitely many
    graphs satisfy $\Phi$.
    We complement this result with a classification for each $\Phi$ that is \emph{edge
    monotone} (that is, closed under edge deletions): \NUM$\indsubsprob(\Phi)$ is \w-hard
    except in the degenerate case when there are only finitely many integers $k$ such that
    $\Phi$ is nontrivial on $k$-vertex graphs.
    Our result generalizes earlier results for specific properties $\Phi$ that are related
    to the connectivity or density of the graph.

    Further, we extend the \w-hardness result by a lower bound which shows that
    \NUM$\indsubsprob(\Phi)$ cannot be solved in time $f(k) \cdot \smash{|V(G)|^{o(\sqrt{\log k} /
    \log\log k)}}$ for any function $f$, unless the Exponential-Time Hypothesis (ETH)
    fails.
    For many natural properties, we obtain even a tight bound $f(k) \cdot |V(G)|^{o(k)}$;
    for example, this is the case for every property $\Phi$ that is nontrivial on
    $k$-vertex graphs for each $k$ greater than some $k_0$.
\end{abstract}

\clearpage

\thispagestyle{plain}
\tableofcontents

\clearpage
\addtocounter{page}{-1}

\section{Introduction}

Searching and counting patterns is one of the oldest algorithmic tasks in computer
science and has many applications in other scientific fields. The theoretical notion of
graphs give a useful and widely used way to model various types of data.
In this work, we focus on the parameterized complexity of counting small-size patterns in graphs, which is
motivated by related applications in the study of database systems
\cite{10.1145/380752.380867}, neural and social networks
\cite{doi:10.1126/science.298.5594.824, 10.1007/978-3-319-21233-3_5}, biology
\cite{10.1007/11599128_7}, and many other fields.

When counting patterns (such as cycles or connected graphs) of some small size $k$ in an
$n$-vertex graph, one can use an obvious brute-force approach that enumerates all $O(n^k)$
subsets of $k$ vertices in the graph. While this is polynomial for fixed $k$, it would be
desirable to have a running time less sensitive to the size $k$ of the pattern. Flum and
Grohe~\cite{path_cycle} initiated the study of the parameterized complexity of counting
problems, with the goal of determining which counting problems are {\em fixed-parameter
tractable (FPT)}, that is, can be solved in time $f(k)n^{O(1)}$ for some computable
function $f$. It is widely believed that many basic counting problems, such as
\textsc{\NUM{}Clique} are {\em not} FPT. The notion of \w-hardness can be used to give
evidence that a counting problem is unlikely to be FPT by showing that it is at least as
hard as \textsc{\NUM{}Clique}. For example, it is known that counting paths of length $k$,
cycles of length $k$, matchings of size $k$, and many other types of subgraphs are \w-hard
\cite{path_cycle,10.1007/978-3-642-39206-1_30,6978997,hom:basis}.

While the parameterized complexity of counting different types of subgraphs is well
understood, the complexity of more general properties of patterns is far from clear.
Formally, a \emph{graph property} is a computable function $\Phi$ from the set of graphs
to $\{0, 1\}$ that is invariant under relabeling. Common examples are \emph{is bipartite},
\emph{is clique}, \emph{is independent set}, \emph{is connected}, or \emph{is planar}, to
name but a few.
Given a property $\Phi$, we would like to count the number of subsets of vertices of a
specified size that induce a graph with this property.
That is, for a fixed graph property $\Phi$, Jerrum and Meeks \cite{connected,
hard_families} introduced the $\NUM{}\indsubsprob(\Phi)$ problem, where given a graph $G$ and
a nonnegative integer $k$, the task is to compute the number of induced subgraphs of $G$
of size $k$ that satisfy $\Phi$. We denote this number by $\NUM{}\indsubs{(\Phi, k)}{G}$ and
write $\NUM{}\indsubs{(\Phi, k)}{\star}$ for the function that maps $G$ to $\NUM{}\indsubs{(\Phi,
k)}{G}$.
This problem for specific types of properties $\Phi$ has been in the focus of a large
amount of research in recent years \cite{even_odd, hom:basis, topo, hamming, alge,
hereditary}.

It is obvious that $\NUM{}\indsubsprob(\Phi)$ is \w-hard if $\Phi$ is the graph property
\emph{is a complete graph} since $\NUM{}\indsubs{(\Phi, k)}{G}$ is the number of $k$-cliques.
However, it turns out that $\NUM{}\indsubsprob(\Phi)$ is also \w-hard for many other graph
properties. Early works in that area focused on showing \w-hardness for specific graph
properties \cite{connected, even_odd}, and it looks like $\NUM{}\indsubsprob(\Phi)$ is \w-hard
for all possible graph properties $\Phi$ except trivial ones. In this setting, we say that
{\em $\Phi$ is trivial on $k$} if it is constant on $k$-vertex graphs, hence
$\NUM{}\indsubs{(\Phi, k)}{G}$ is either 0 or $\binom{|V(G)|}{k}$. We say that $\Phi$ is {\em
trivial} if there is an $N$ such that $\Phi$ is trivial on all $k \geq N$. It is easy to
verify that $\NUM{}\indsubsprob(\Phi)$ is FPT whenever $\Phi$ is trivial.

\begin{conjecture}[\cite{hereditary, hamming}]\label{con:w1:hard}
    For all nontrivial, computable graph properties $\Phi$ the $\NUM{}\indsubsprob(\Phi)$ problem
    is \w-hard. Otherwise, $\NUM{}\indsubsprob(\Phi)$ is FPT.
\end{conjecture}

Curticapean, Dell, and Marx showed in \cite{hom:basis} that $\NUM{}\indsubsprob(\Phi)$ is
always FPT or \w-hard. However, the proof does not give an easy way to determine which
case holds for a given property $\Phi$. As we review it in \cref{chap:prior},
\cref{con:w1:hard} is known to hold for many classes of graph properties. In particular,
Focke and Roth~\cite{hereditary} showed that \cref{con:w1:hard} holds if $\Phi$ is a {\em
hereditary} property, that is, closed under deletion of vertices.  In this paper, we show
a complementary result for graph properties that are \emph{edge monotone} (closed under
deletion of edges). We also provide quantitative lower bounds on the exponent of $n$ under
the {\em Exponential-Time Hypothesis (ETH)} \cite{IMPAGLIAZZO2001512}.

\begin{restatable*}{mtheorem}{thmedgemon}\label{theo:edge_mono}\label{theo:general:lower:bound}
    Let \(\Phi\) denote a nontrivial edge-monotone graph property.
    \begin{itemize}
        \item The problem \(\NUM{}\indsubsprob(\Phi)\) is \(\w\)-hard.
        \item Further, assuming ETH, there is a universal constant $\gamma > 0$ (independent of \(\Phi\)) such that
            for any integer $k \geq 3$ on which $\Phi$ is nontrivial,
            no algorithm (that reads the whole input)
            computes for every graph \(G\) the number $\NUM{}\indsubs{(\Phi, k)}{G}$ in time
            $\smash{O(|V(G)|^{\gamma \sqrt{\log k} / \log\log k})}$.
            \qedhere
    \end{itemize}
\end{restatable*}
\medskip

As an example, the following nontrivial graph properties are all edge-monotone and not
covered by any method in \cref{chap:prior}, but \cref{theo:edge_mono} shows their
\w-hardness.

\begin{example}\label{example:monotone:properties}
    \begin{itemize}
        \item $\Phi_1^c(G) = 1$ if and only if $G$ is disconnected or $\diam(G) \geq c
            |V(G)|$ for a fixed constant $c \in (0, 1)$.
        \item $\Phi_2^c(G) = 1$ if and only if $G$ is bipartite or contains an independent
            set of size at least $c$ for a fixed constant $c \geq 3$.
        \item $\Phi_3^c(G) = 1$ if and only if the maximum degree of $G$ is at most
            $c|V(G)|$ for a fixed constant $c \in (\nicefrac{1}{2}, 1)$.
            \qedhere
    \end{itemize}
\end{example}

A notable difference between hereditary and edge-monotone properties is that an
edge-monotone property might be nontrivial only on a sparse set of integers: for example,
one can define $\Phi(G)$ to be 1 if and only if $G$ is a clique with $|V(G)|$ being a
power of two (or a prime, or the product of the first $i$ primes, or $\ldots$). On the
other hand, if $\Phi$ is a nontrivial hereditary property, then it is an easy exercise to show that
it has to be nontrivial for every $k$ larger than some $N$.
Our algebraic proof techniques for \cref{theo:edge_mono} are very sensitive to the
number-theoretic properties of the size $k$, hence it is a significant challenge to make
it work when $\Phi$ can be nontrivial only on certain integers.

In \cite{topo, hamming, alge, hereditary}, it was proven that for specific classes of
nontrivial graph properties~$\Phi$, there is a function $g$ such that the problem
$\NUM{}\indsubsprob(\Phi)$ cannot be solved in time $f(k) \cdot |V(G)|^{o(g(k))}$ for any
computable function $f$, unless ETH fails.
The second part of \cref{theo:edge_mono} shows a similar lower bound on the exponent of
the running time for edge-monotone properties.
However, readers familiar with the way ETH-based lower bounds are stated in the
parameterized complexity literature should notice that the second part of
\cref{theo:edge_mono} uses a very different formulation (a similar formulation was given
by Cohen-Addad et al.~\cite{cohenaddad2021tight} for various problems). We use a stronger
approach by showing that, assuming ETH, there is a constant $\gamma > 0$ such that for
all edge-monotone graph properties $\Phi$ and {\em any} fixed nontrivial $k \geq 3$ the
function $\NUM{}\indsubs{(\Phi, k)}{G}$ cannot be computed for all graphs \(G\) in time $O(|V(G)|^{\gamma
g(k)})$ unless ETH fails. That is, our stronger statement applies not only to algorithms
solving the problem in general, but also gives a meaningful statement for algorithms
solving the problem for a fixed $k$. It is easy to observe that this approach implies that
no algorithm solves $\NUM{}\indsubsprob(\Phi)$ in time $f(k) \cdot |V(G)|^{o(g(k))}$ for any
computable function $f$, unless ETH fails.

\begin{corollary}\label{cor:logsquare:nontight}
    Let \(\Phi\) denote a nontrivial edge-monotone graph property.
    \begin{itemize}
        \item The problem \(\NUM{}\indsubsprob(\Phi)\) is \(\w\)-hard.
        \item Further, assuming ETH, no algorithm
            computes for every graph \(G\) and every positive integer \(k\) the number
            $\NUM{}\indsubs{(\Phi,k)}{G}$ in time $f(k) \cdot |V(G)|^{o(\sqrt{\log k} /
            \log \log k)}$ for any computable function~$f$.
    \end{itemize}
\end{corollary}
\begin{proof}
    Suppose that $\NUM{}\indsubsprob(\Phi)$ can be solved in $O(f(k) |V(G)|^{g(k)})$ for $g(k)
    \in o(\sqrt{\log k} / \log \log k)$.
    Then, there is an $N$ with $g(k') \leq \gamma
    \sqrt{\log k'} / \log \log k' $ for all $k' \geq N \geq 3$. Since $\Phi$ is
    nontrivial, there is a $k \geq N$ such that $\Phi$ is nontrivial on $k$.
    \Cref{theo:edge_mono} shows that no algorithm solves $\NUM{}\indsubsprob(\Phi)$ in
    $O(f(k) |V(G)|^{g(k)})$.
\end{proof}
Thus our formulation of the second part of \cref{theo:edge_mono} implies the usual
formulation. Further, even though the statement is stronger, it turns out that it is
somewhat easier to work with this formulation: technicalities involving the little-$o$
notation and the function $f(k)$ disappear, which results in more streamlined proofs.

Observe that quantitative lower bounds for $ \NUM{}\indsubsprob(\Phi)$ in
\cref{theo:edge_mono} are fairly weak: we say that the exponent cannot be much better than
$\sqrt{\log k}/\log\log k$. Naturally, we would like to show lower bounds for
$\NUM{}\indsubsprob{(\Phi)}$ of the form $O(|V(G)|^{\gamma k})$, that is, we would like to
show bounds that are tight in the sense that we can indeed solve $\NUM{}\indsubsprob(\Phi)$ in time $O(|V(G)|^k)$
using a brute force approach. While we cannot prove such tight lower bounds in
full generality, we are able to obtain tight lower bounds for specific cases that cover most
properties of interest.

\begin{restatable*}{mtheorem}{thetightlowerbound}\label{theo:tight:lower:bound}
    For each prime $p$, there is a constant $\gamma_p > 0$ such that for
    each integer \(m\) with $p^m \geq 3$ and each edge-monotone graph property $\Phi$ that is nontrivial
    on $p^m$, no algorithm (that reads the whole input) computes for every graph \(G\) the
    number $\NUM{}\indsubs{(\Phi, k)}{G}$ in time $O(|V(G)|^{\gamma_p p^m})$, unless ETH
    fails.
\end{restatable*}
As before, we also restate \cref{theo:tight:lower:bound} in terms of ruling out $o(k)$ in the exponent.
\begin{corollary}\label{cor:tight:lower:bound}
    For all edge-monotone $\Phi$ that for fixed prime number $p$ are nontrivial on
    infinitely many numbers of the form $p^m$, no algorithm
    computes for every graph \(G\) end every positive integer \(k\) the number
    $\NUM{}\indsubs{(\Phi, k)}{G}$ in time $f(k) \cdot |V(G)|^{o(k)}$ for any computable
    function $f$, unless ETH fails.
\end{corollary}
\begin{proof}
    Suppose that $\NUM{}\indsubsprob(\Phi)$ can be solved in $O(f(k) |V(G)|^{g(k)})$ for $g(k)
    \in o(k)$. Then, there is an $N$ with $g(k') \leq \gamma k$ for all $k' \geq N \geq 3$.
    Since $\Phi$ is nontrivial, there is a $k \geq N$ with $\Phi$ is nontrivial on $k$.
    \cref{theo:tight:lower:bound} shows that no algorithm solves $\NUM{}\indsubsprob(\Phi)$ in
    $O(f(k) |V(G)|^{g(k)})$.
\end{proof}
Observe that \cref{cor:tight:lower:bound} holds whenever there is a constant $N$ such that $\Phi$ is
nontrivial on all $k \geq N$. It is easy to check that this condition holds for all $\Phi_i$ in
\cref{example:monotone:properties}. Thus, we obtain tight lower bounds for all $\NUM{}\indsubsprob(\Phi_i)$
using \cref{cor:tight:lower:bound}.

\subsection{Prior Work}\label{chap:prior}
In the following, we summarize recent results for the parameterized complexity of the
$\NUM{}\indsubsprob(\Phi)$ problem. The results are ordered by their date of publication.

\begin{enumerate}[(a)]
    \item $\NUM{}\indsubsprob(\Phi)$ is \w-hard for
    \begin{enumerate}[(i)]
        \item $\Phi(G) = 1$ if and only if $G$ is connected \cite{connected}
        \item $\Phi(G) = 1$ if and only if $|E(G)|$ is even,  $\Phi(G) = 1$ if and only if
            $|E(G)|$ is odd (both in \cite{even_odd})
    \end{enumerate}
    \medskip

    \item In \cite{hard_families}, Jerrum and Meeks proved \w-hardness if the number of
        distinct edge densities of graphs that satisfy $\Phi$ is low. This also shows that
        $\NUM{}\indsubsprob(\Phi)$ is \w-hard whenever $\Phi$ is minor-closed.
    \medskip

\item In \cite{MEEKS2016170}, Meeks proved that $\NUM{}\indsubsprob(\Phi)$ is \w-hard if
        $\Phi$ is closed under the addition of edges and the edge-minimal graphs of $\Phi$
        have unbounded treewidth. A Graph is edge-minimal if $\Phi(G) = 1$ and $\Phi(G') =
        0$ for all proper edge-subgraphs of $G$.
    \medskip

\item In \cite{topo}, Roth and Schmitt proved that $\NUM{}\indsubsprob(\Phi)$ is \w-hard
        if $\Phi$ is nontrivial, edge-monotone and satisfies at least one of the following
        conditions.
    \begin{enumerate}[(i)]
        \item $\Phi$ is false for odd cycles. A cycle has
            vertex set $\{0, \dots, n-1\}$ and $\{a, b\}$ is an edge if and only if $a - b \equiv_n 1$.
        \item $\Phi$ is true for odd anti-holes. An anti-hole is the complement graph of a cycle.

        \item There is a $c \in \nat$ such that $\Phi(H) =
            1$ if and only if $H$ is not c-edge-connected

        \item There is a graph $F$ such that $\Phi(H) = 1$ if and only if there is no homomorphism from $F$ to $H$.
    \end{enumerate}
    \medskip

    \item In \cite{alge},  Dörfler, Roth, Schmitt, and Wellnitz proved that
        $\NUM{}\indsubsprob(\Phi)$ is \w-hard if there are infinitely many prime powers $t$
        such that $\Phi(K_{t, t}) \neq \Phi(\IS_{2t})$.
    \medskip

    \item In \cite{hamming}, Roth, Schmitt, and Wellnitz proved the following criteria for
        checking \w-hardness. Let $f^{\Phi, k}_i \coloneqq \NUM{} \{A  \subseteq E(K_k) :
            \NUM{}A = i
        \land \Phi(K_k[A]) = 1\}$ denote a vector and $\hw(f^{\Phi, k}) \coloneqq \NUM{}\{i : f^{\Phi,
        k}_i \neq 0 \}$ the hamming weight of $f^{\Phi, k}$. We define the function $\beta
        \colon \mathcal{K}(\Phi) \to \Z_{\geq 0}; k \mapsto \binom{k}{2} - \hw(f^{\Phi,
        k})$, where $\mathcal{K}(\Phi)$ is the set of $n \in \nat$ with $\Phi$ is
        nontrivial on $n$. The problem $\NUM{}\indsubsprob(\Phi)$ is \w-hard if $\beta(k) \in \omega(k)$.

        They also proved that $\NUM{}\indsubsprob(\Phi)$ is \w-hard if $\Phi$ is
        \emph{monotone}, meaning closed under taking subgraphs.
    \medskip

\item In \cite{hereditary}, Focke and Roth proved that $\NUM{}\indsubsprob(\Phi)$ is
        \w-hard if $\Phi$ is nontrivial and hereditary. A~property is called
        \emph{hereditary} if it is closed under vertex-deletion, meaning that if $G$
        satisfies $\Phi$, then each induced subgraph of $G$ also satisfies $\Phi$.
    \medskip

\item Lastly, we observe that the counting problem $\NUM{}\indsubsprob(\Phi)$ is \w-hard
    if and only if $\NUM{}\indsubsprob(\neg \Phi)$ is \w-hard; and
    $\NUM{}\indsubsprob(\Phi)$ is \w-hard if and only if
    $\NUM{}\indsubsprob(\overline{\Phi})$ is \w-hard. Here, $\neg \Phi(G) \coloneqq 1 - \Phi(G)$;
        and $\overline{\Phi}(G) \coloneqq \Phi(\overline{G})$, where $\overline{G}$ is the
        complement of $G$ (see \cite[see Fact 2.3]{hamming}). This means that we can prove
        \w-hardness of $\NUM{}\indsubsprob(\Phi)$ by analyzing $\NUM{}\indsubsprob(\neg \Phi)$
        or $\NUM{}\indsubsprob(\overline{\Phi})$.
\end{enumerate}

\begin{figure}
    \centering
    \begin{tikzpicture}[yscale=-1]
        \draw[rounded corners=5pt,draw=black!30,ultra thick, fill=black!2]
            (-7,-3) rectangle (7,3);
            \node[anchor=north west, outer sep=2pt,black!80] at (-7,-3) {{\bf nontrivial}};

        \draw[rounded corners=5pt,draw=purple!50!red!30,ultra thick, fill=purple!2]
            (-6.75,-2.25) rectangle (2,2.5);
        \node[anchor=north west, outer sep=2pt,purple!50!red!80,align=left] at (-6.75,-2.25)
        {{\bf edge-monotone}\\{\footnotesize (closed under}\\{\footnotesize removal of edges)}\\
        \w-hard\\{\bf [this work]}};

        \draw[rounded corners=5pt,draw=blue!30,ultra thick, fill=blue!2]
            (-2,-2) rectangle (6.75,2.75);
        \node[anchor=north east, outer sep=2pt,blue!80,align=right] at (6.75,-2)
        {{\bf hereditary}\\{\footnotesize (closed under}\\{\footnotesize removal of
            vertices)}\\
        \w-hard\\{\cite{hereditary}}};

        \draw[rounded corners=5pt,draw=purple!50!red!50!blue!30,ultra thick, fill=purple!50!red!50!blue!2]
            (-2,2.5) -- (-2,-2) -- (2,-2)
            (-2,2.5) -- (2,2.5) -- (2,-2);

        \node[anchor=north, outer sep=2pt,purple!50!red!50!blue!80,align=center] at (0,-2)
        {{\bf monotone}\\{\footnotesize (closed under removal}\\{\footnotesize of
            vertices and edges)}\\
        \w-hard\\{\cite{hamming}}};

        \draw[rounded corners=5pt,draw=black!20]
            (-6.5,.5) rectangle (-1,2);
        \draw[rounded corners=5pt,draw=black!20]
            (-1.75,.75) rectangle (1.75,2.25);
        \node[anchor=north west, outer sep=2pt,black!50,align=left] at (-6.5,.55) {{other results}\\
        \cite{MEEKS2016170,connected,topo}};
        \node[anchor=north east, outer sep=2pt,black!50,align=right] at (1.75,.8)
        {{minor-closed}\\
        \cite{hard_families}};

    \end{tikzpicture}

    \caption{Hierarchy of classes of graph properties \(\Phi\), together  with the results that
        show \w-hardness for the corresponding problem \(\NUM{}\indsubsprob(\Phi)\) (if such
        results exist).}
    \label{fig:venn:diagram}
\end{figure}
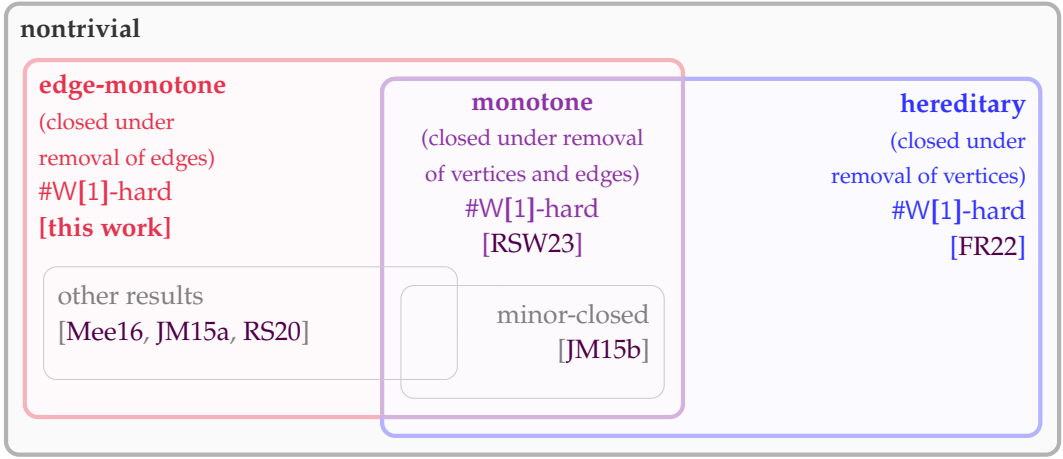

\subsection{High-level Ideas}

Our results build on concepts introduced in earlier work, but we need to develop
substantial new technical ideas to be able to deploy them in our setting. We briefly
review these concepts here and highlight our main new technical contributions; a more
detailed technical overview follows in \cref{sec:techoverview}.

Given a graph property $\Phi$, we define the \emph{alternating
    enumerator} \(\aename{\Phi}\) as
    \[
        \ae{\Phi}{H} \coloneqq \sum_{S \subseteq E(H)} \Phi(\ess{H}{S}) (-1)^{\NUM{}S}.
    \]
Here, $\ess{H}{S}$ denotes the subgraph of $H$ that contains the same set of vertices but only the edge set $S$.
Dörfler~et~al.~\cite{alge} proved that $\NUM{}\indsubsprob(\Phi)$ is \w-hard if there are
graphs with arbitrary large treewidth and nonzero alternating enumerator.
Thus, \w-hardness can be established by showing that such graphs exist.
However, the complicated definition of
$\ae{\Phi}{H}$ does not make this task easy, even for a specific property $\Phi$.
Dörfler~et~al.~\cite{alge} made the following observation that can help in arguing that
$\ae{\Phi}{H}$ is nonzero. Write $\Gamma$ for a group that consists in automorphisms of $V(H)$
and suppose that the order of $\Gamma$ is a power of $p$. We say that a subset $S\subseteq
E(H)$ is a fixed point with respect to $\Gamma$ if every automorphism in $\Gamma$ moves
every edge in $S$ to an edge in $S$. Clearly, we can show that $\ae{\Phi}{H}$ is nonzero
by showing that it is nonzero modulo $p$. The main observation is that if we want to
compute the sum in $\ae{\Phi}{H}$ modulo $p$, then we do not need to sum over all sets
$S\subseteq E(H)$: it is sufficient to sum over the fixed points with respect to $\Gamma$,
as the other subgraphs somehow cancel out. Thus, \w-hardness can be established by finding
graphs~$H$ with large treewidth and appropriate groups $\Gamma$ to show that
$\ae{\Phi}{H}$ is nonzero modulo some prime $p$.

It is not difficult to show that every fixed point with respect to $\Gamma$ is
the disjoint union of orbits of edges. This means that the fixed points have a
natural lattice structure: we can imagine the fixed points that are the disjoint union of
$\ell$ orbits as the $\ell$-th level of the lattice.
In broad terms, our approach is to define some group $\Gamma$, consider the fixed points
of the complete graph on $k$ vertices with respect to $\Gamma$, and then try to find a
fixed point of sufficiently high level whose alternating enumerator is nonzero. Fixed
points on higher levels have more edges and hence larger treewidth. Now, we invoke (adaptations of)
the earlier results of Dörfler et al.~\cite{alge} to show hardness with
such graphs. In more detail, our proofs are based on the following four main technical
ideas.

\begin{enumerate}
    \item \textbf{Duality of the highest nonzero level.} We use linear algebra arguments to
        show that if $\Phi$ is 0 on every fixed point on the topmost $c$ levels, then there is
        a fixed point on level $\ge c$ whose alternating enumerator is nonzero, and the
        level implies that treewidth is at least $c$.  Thus, we may
        assume that there is a fixed point of fairly high level that is nonzero in $\Phi$.
        Note that this statement is true for every property $\Phi$, even if $\Phi$ is not
        edge-monotone.
        \medskip

    \item \textbf{Avalanche effect for difference graphs on $\field{p^m}$}. Write $A$ for a
        subset of $\field{p^m}$. Then we can define the difference graph on the vertex set
        $\field{p^m}$, where there is an edge between $x$ and $y$ if and only if $x-y\in A$.
        Write $\Gamma$ for the additive group of $\field{p^m}$. One can observe that the fixed
        points with respect to $\Gamma$ are exactly the difference graphs. If $\Phi$ is
        edge-monotone and $\Phi$ is nonzero on a fixed point $S$, then this implies that $\Phi$
        should be nonzero on all the other fixed points that are subsets of $S'$. In
        particular, a single nonzero fixed point on one of the top $c$ levels starts an
        ``avalanche'' that forces every fixed point on a level of at most roughly $p^m/c$
        to be nonzero in $\Phi$.
        The proof is based on the fact that $\field{p^m}$ multiplication is an isomorphism of the
        difference graph. Thus, if we pick the fixed point with the lowest level that does not
        satisfy $\Phi$, then it has a level of at least $p^m/c$ and hence fairly large treewidth. A
        simple calculation using the binomial theorem shows that the any fixed point of
        the lowest level that do not satisfy $\Phi$ always has a nonzero alternating enumerator.
\end{enumerate}

We combine the previous two ideas to obtain \w-hardness if $\Phi$ is nontrivial on infinitely
many prime powers. If $k=p^m$, then we use the following win/win approach: if $\Phi$
is 0 on the topmost $c=\sqrt{p^m}$ levels, then (1) gives a fixed point of treewidth at
least $\sqrt{p^m}$ with nonzero alternating enumerator; otherwise, (2) gives such a fixed
point. To obtain the \w-hardness result in \cref{theo:edge_mono}, we
extend our proof to the case when $k=d\cdot p^m$ with the following idea.
We also prove the ETH-based quantitative part of \cref{theo:edge_mono}
by observing that the largest prime
power divisor of $k$ is always at least logarithmic in $k$.

\begin{enumerate}
  \setcounter{enumi}{2}
    \item \textbf{Product construction and reduction.} We extend the lower bound for
        properties that are nontrivial on some prime powers to the general case in the following way.
        Write $\Gamma$ for a group of permutations on $p^m$ vertices. There is a natural way
        to raise \(\Gamma\) to a group $\Gamma^d$ of permutations on $d\cdot p^m$ vertices. We
        observe that the fixed points of $\Gamma^d$ can be described as the disjoint
        unions of the fixed points of $\Gamma$, plus additionally fully connecting some
        pairs of these fixed points. Let us take a look at the fixed point \(F\) with the lowest level
        that is zero in $\Phi$ (which we know to have nonzero alternating enumerator). If
        \(F\) contains one of the aforementioned full connections, then \(F\) has high treewidth, which is
        what we wanted. Otherwise, \(F\) is the disjoint union of fixed points of $\Gamma$;
        write $F'$ for one of them and write $H$ for the union of the remaining $d-1$
        fixed points. Then by
        ``pinning'' $H$ we can define a nontrivial property $\Phi'$ on $p^m$: set
        $\Phi'(G) \coloneqq \Phi(H\unionSet G)$, which is zero on the $p^m$-vertex graph
        $F'$. Now a
        standard reduction based on the Inclusion-Exclusion principle shows how to reduce
        $\NUM{}\indsubsprob(\Phi')$ to $\NUM{}\indsubsprob(\Phi)$, hence the lower bounds for
        prime powers can be used.
\end{enumerate}

Unfortunately, a square root loss between $k$ and the treewidth of the identified fixed
point is inevitable when using (1) and (2), hence they cannot lead to tight bounds. The
last idea allows us to obtain tight bounds in at least some cases.

\begin{enumerate}
  \setcounter{enumi}{3}
    \item\textbf{Avalanche effect for lexicographic product of graphs.}  For a prime $p$
        and integer $m\ge 2$, we show another way of defining fixed points on $p^m$
        vertices that have better avalanche properties. We define a group $\Gamma$ (which
        is in fact the Sylow $p$-group of the automorphism group of $K_{p^m}$) such that
        the fixed points with respect to $\Gamma$
        are exactly the so-called $m$-dimensional lexicographic products of difference
        graphs on $\field{p}$. Let us consider the lowest level $\ell$ that contains a
        fixed point where $\Phi$ is 0. We observe that every level $\ell$ has fixed points
        that contain the complete bipartite graph $K_{p^{m-1},p^{m-1}}$. If $\Phi$ is 0 on
        one such fixed point \(F\) on level $\ell$, then it is easy to show that the alternating
        enumerator of \(F\) is nonzero. Now, earlier work shows how to reduce the counting of
        $p^{m-1}$-cliques to $\NUM{}\indsubsprob(\Phi)$ with $k=p^m$. Otherwise, if $\Phi$ is
        nonzero on every such fixed point on level $\ell$, then the avalanche effect shows
        that $\Phi$ is nonzero also on every other fixed point on the same level $\ell$, a
        contradiction.
\end{enumerate}

Specifically, if $k=p^m$ for some constant prime $p$, then we obtain the tight bound that
$\NUM{}\indsubsprob(\Phi)$ for $k$ is at least as hard as counting $p^{m-1}$-cliques, which proves
\cref{theo:tight:lower:bound}.

\section{Technical Overview}\label{sec:techoverview}

In this section, we present an overview of the most important techniques and ideas that
we use to show \w-hardness of $\NUM{}\indsubsprob(\Phi)$ for each nontrivial edge-monotone
graph property $\Phi$, as well as the ETH-based quantitative lower bounds. We start with a
review of the techniques that we use from previous work and then elaborate on our novel
technical ideas.

\paragraph*{Alternating Enumerator, \w-hardness, and Lower Bounds}

A problem instance of $\NUM{}\homsprob(\mathcal{H})$ is a pair of a graph $H \in \mathcal{H}$ and a graph
$G \in \mathcal{G}$ and the output is the number of homomorphisms from $H$ to $G$ (that
is, $\NUM{}\homs{H}{G}$); we parameterize by $\kappa(H, G) \coloneqq |V(H)|$.
Dalmau and Jonsson~\cite{DALMAU2004315} proved that
$\NUM{}\homsprob(\mathcal{H})$ is \w-hard if and only if the treewidth of
the set $\mathcal{H}$ is unbounded (that is, there is no constant $c$ such that the
treewidth of all elements in $\mathcal{H}$ is below $c$).

Further results follow from the fact that a certain colored version of
$\homsprob(\mathcal{H})$ can be reduced to $\NUM{}\homsprob(\mathcal{H})$
\cite{cohenaddad2021tight,DBLP:journals/toc/Marx10}. We prove our \w-hardness results by
using a parameterized Turing reduction from
$\NUM{}\homsprob(\mathcal{H})$ to $\NUM{}\indsubsprob(\Phi)$, which was first developed by
D\"orfler et al.~\cite{alge}.

For this reduction to work, the methods of \cite{alge}
require that the \emph{alternating enumerator} \(\ae{\Phi}{H}\) is nonvanishing for each graph $H \in
\mathcal{H}$.\footnote{In \cite{topo, alge}, the authors use $\hat{\chi}(\Phi, H)$ to
denote the alternating enumerator.}

\begin{restatable}{definition}{defae}\label{def:alternating_enumerator}
    For a graph property \(\Phi\) and a graph \(H\), we define the \emph{alternating
    enumerator} \(\ae{\Phi}{H}\)~as
    \[
        \ae{\Phi}{H} \coloneqq \sum_{S \subseteq E(H)} \Phi(\ess{H}{S}) (-1)^{\NUM{}S}.
    \]
    We say that a graph \(H\) is \emph{nonvanishing for a graph property \(\Phi\)} if the
    alternating enumerator of \(\Phi\) and \(H\) is nonzero.
    We say that a sequence of graphs \(H_k\) is nonvanishing for a graph property
    \(\Phi\) if every \(H_k\) is nonvanishing for \(\Phi\),
    that is, if for all \(k\), we have \(\ae{\Phi}{H_k} \ne 0\).
\end{restatable}

This means that we can show \w-hardness by finding a nonvanishing sequence $H_k$ that
has unbounded treewidth.
Further, assuming ETH, the reduction of \cite{alge} yields a
lower bound for $\NUM{}\indsubs{(\Phi, k)}{G}$ for a fixed $k$. To be more precise, we show that
if \(H\) is nonvanishing for \(\Phi\), we can use the reduction of \cite{alge}
to solve $\homsprob(\{H\})$ using an oracle for $\NUM{}\indsubs{(\Phi,
|V(H)|)}{\star}$.

From the work of Cohen-Addad et al.~\cite{cohenaddad2021tight}, we obtain that $\NUM{}\homsprob(\{H\})$ cannot be solved in
time $O(n^{\alpha_{\homsprob} \cdot \tw(H) / \log \tw(H)})$, where $n$ is the number of vertices of the
input graph. Hence, if we could compute $\NUM{}\indsubs{(\Phi, |V(H)|)}{\star}$ fast enough, then our reduction
shows that we can solve $\homsprob(\{H\})$ in time $O(n^{\alpha_{\homsprob} \cdot \tw(H) / \log
\tw(H)})$, which contradicts ETH.
Summarizing the previous discussion, we obtain the following lemma, which essentially
follows from previous work.
For completeness, we include a proof in \cref{sec:appendix:A}.

\begin{restatable*}[\cite{alge, cohenaddad2021tight}]{lemma}{lemchitw}
    \label{lem:alpha:treewidth}
    Let \(\Phi\) denote a nontrivial graph property.
    \begin{itemize}
        \item If there is a sequence of graphs with unbounded treewidth where each graph
        has an alternating enumerator that is nonvanishing for \(\Phi\), then
        \(\NUM{}\indsubsprob(\Phi)\) is $\w$-hard.

        \item Assuming ETH, there is a universal constant
        $\alpha_{\indsubsprob} > 0$ (that is independent of \(\Phi\))
        such that for any positive integer $k$ for which
        there is a graph $H_k$ with $k$ vertices, $\ae{\Phi}{H_k} \neq 0$, and $\tw(H_k)
        \geq 2$, no algorithm (that reads the whole input) computes for every graph \(G\)
        the number $\NUM{}\indsubs{(\Phi, k)}{G}$ in time
        $O(|V(G)|^{\alpha_{\text{\indsubsprob}} \tw(H_k) / \log \tw(H_k)})$.
        \qedhere
    \end{itemize}
\end{restatable*}
\medskip

For edge-monotone graph properties $\Phi$, the alternating enumerator of the $k$-clique
$\ae{\Phi}{K_k}$ is equal to the reduced Euler characteristic of the simplicial graph
complex $\hat{\chi}(\Delta(\Phi_k))$, where $\Phi_k$ is $\Phi$ restricted on $k$-vertex
graphs (see \cite[Lemma 14]{alge}). The reduced Euler characteristic is in turn closely
related to Karp’s famous evasiveness conjecture (see \cite{evasiveness}), which
conjectures that each nontrivial edge-monotone $\Phi_k$ is evasive. This conjecture holds if the
reduced Euler characteristic $\hat{\chi}(\Delta(\Phi_k))$ is nonvanishing (see
\cite[Theorem 4]{alge} and \cite{evasiveness}).

However, this means that the computation of \(\ae{\Phi}{H}\) is highly nontrivial, which
makes it hard to apply \cref{lem:alpha:treewidth}.
Fortunately for us, it suffices to show that $\ae{\Phi}{H} \not \equiv_p 0$ for a prime
number $p$---which turns out to be easier. In particular,
as observed by D\"orfler et al.~\cite{alge}, for a prime \(p\),
we can compute $\ae{\Phi}{H} \bmod p$
in an elegant way using the fixed points of a \(p\)-subgroup
\(\Gamma\) of \(\aut(H)\) when acting on edge-subgraphs of \(H\).
Thus, we heavily rely on the following lemma (whose proof is presented in
\cref{sec:appendix:A} for completeness).

\begin{restatable*}[\cite{alge}]{lemma}{lemchicomp}\label{lem:chi:comp}\label{equ:lem:chi:comp}
    Let $H$ denote a graph and let $\Gamma \subseteq \aut(H)$ denote a $p$-group, then
    \begin{align*}
        \ae{\Phi}{H} \equiv_p \sum_{ A \in \fp(\Gamma, H) } \Phi(A) (-1)^{\NUM{}E(A)}.
        \tag*{\qedhere}
    \end{align*}
\end{restatable*}

Combining \cref{lem:alpha:treewidth,lem:chi:comp}, we immediately obtain the following tool to show hardness.
\begin{restatable}{corollaryq}{corseqhard}\label{cor:chi:hard}
    Let \(\Phi\) denote a graph property and
    let $(H_k)$ denote a sequence of graphs such that
    \begin{itemize}
        \item \((H_{k})\) has unbounded treewidth and
        \item for each graph \(H_{k}\), there is a prime \(p_k\) and a \(p_k\)-group
            \(\Gamma_k \subseteq \auts{H_{k}}\) such that
            \begin{align*}
                \ae{\Phi}{H_k} \equiv_{p_k}
                \sum_{A \in \fpb{\Gamma_k}{H_k}} \Phi(A) (-1)^{\NUM{}E(A)}
                 \quad\text{is nonzero modulo \(p_k\).}
            \end{align*}
    \end{itemize}
    Then, $\NUM{}\indsubsprob(\Phi)$ is \w-hard.
\end{restatable}

A fixed point $A$ of $\Gamma$ in $H$ is an edge-subgraph of $H$ such that $gA = A$ for all
$g \in \Gamma$. We use $\fp(\Gamma, H)$ to denote the set of all fixed points. The
advantage of this approach is
that the set of fixed points $\fp(\Gamma, H)$ is usually much smaller than the set of all
edge-subgraphs
of $H$. As it turns out, the set $\fp(\Gamma, H)$ itself has a natural lattice
structure, which we exploit to find nonvanishing graphs with large treewidth.

\paragraph*{Fixed Points as a Union of Orbits}

Our goal is to find, for a given graph property $\Phi$ and value $k$, a nonvanishing
graph $H$ with $k$ vertices that has \emph{large} treewidth. To find these graphs, we
analyze the fixed point structure of a certain graph $H$ under a certain $p$-group $\Gamma
\subseteq \aut(H)$.

In \cref{chap:prop:chi}, we introduce a systematic way to analyze and describe the fixed
points of a group~\(\Gamma\) and a graph \(H\).
Write $E(H) / \Gamma \coloneqq \{O_1, \dots, O_s\}$ for the orbits of the group
action $\cdot \colon \Gamma \times E(H)$ that maps $g \in \Gamma$ and $\{u, v\} \in E(H)$ to $\{g(u), g(v)\}$.

Our first observation is that every fixed point \(F \in \fp(\Gamma, H)\) decomposes into a set of orbits
of \(E(H)/ \Gamma\); that is, we have
$V(F) = V(H)$ and $E(F) = \cup_{i \in A} O_i$ for some $A \subseteq \setn{s}$.
This means that the orbits $E(H) / \Gamma$ are the {basic building blocks} of the fixed
points $\fp(\Gamma, H)$.

\begin{restatable*}{lemma}{remarkfixunion}\label{remark:fixed:point:union}
    Let \(H\) denote a graph and let \(\Gamma \subseteq \aut(H)\) denote a group.
    Further, let \(\Gamma\) act on \(E(H)\) and write \(E(H)/\Gamma\) for the set of all
    resulting orbits.
    Finally, let \(\Gamma\) act on edge-subgraphs of \(H\) and write \(\fp(\Gamma, H)\) for
    the set of all resulting fixed points of \(\Gamma\) in \(H\).

    Then, the edge set of each fixed point $F \in \fp(\Gamma, H)$ is the (possible empty)
    disjoint union of orbits $O_1, \dots, O_s$, where $O_i \in E(H)/\Gamma$
    and each disjoint union of orbits yields a fixed point.
    The partition into orbits is unique.
\end{restatable*}

The \emph{level} of a fixed point \(F\), denoted by \(\hl{F}\), is the number of orbits
that \(F\) is made up of.
The level allows us to classify the fixed points into \emph{low} fixed points
(consisting of a few orbits) and \emph{high} fixed points (consisting of many
orbits). Usually, fixed points with a high level also have a high treewidth since
they contain more edges.
Additionally, if $\Gamma$ is a $p$-group, we can use the Orbit-stabilizer Theorem (which implies
that the size of the orbit divides the size of the group) to show that
$(-1)^{\NUM{}E(F)} \equiv_p (-1)^{\Hasselevel(F)}$ for all $F \in \fp(\Gamma, H)$.

Our second observation is that fixed points of fixed points lie within each other,
that is, for all $A \in \fp(\Gamma, H)$, we have
\[\fp(\Gamma, A) = \{B \in \fp(\Gamma, H) \mid E(B) \subseteq E(A)\}.\]
We use $B \subseteq A$ to denote that $B$ is a fixed point that lies in
$A$, and say that $A$ is a sub-point $B$.

Combining our two observations with \cref{lem:chi:comp}, we obtain
\begin{align}\label{eq:intro:ae}
    \ae{\Phi}{A} \equiv_p \sum_{\substack{ B \in \fp(\Gamma, H) \\ B \subseteq A}} \Phi(B) (-1)^{\Hasselevel(A)}.
\end{align}
Observe that we are summing only over fixed points $B \in \fp(\Gamma, H)$ that lie in $A$.

Hence, for computing the alternating enumerator of all $A \in \fp(\Gamma, H)$
it suffices to analyze the fixed point structure $\fp(\Gamma, H)$ for a single $H$.
\Cref{eq:intro:ae} helps significantly in understanding the alternating enumerator.
For instance, \cref{eq:intro:ae} plays a central role in proving the following key lemma.

\begin{restatable*}{lemma}{aesubpointstrue}\label{remark:chi:all:children:true}
    Let \(H\) denote a graph, let \(\Gamma \subseteq \auts{H}\) denote a \(p\)-group, and
    let \(\Phi\) denote a graph property.
    Further, let \(A \in \fps{H}\) denote a fixed point without property \(\Phi\),
    such that all of the proper sub-points of \(A\) do have property \(\Phi\);
    that is, we have \(\Phi(A) = 0\) and \(\Phi(B) = 1\) for every \(B \subsetneq A\).
    Then $A$ is nonvanishing.
\end{restatable*}

We employ the following strategy. Instead of directly finding one specific nonvanishing
graph $H$ with $k$ vertices and \emph{large} treewidth, we use \cref{eq:intro:ae} to find
some nonvanishing fixed point $A \in \fp(\Gamma, H)$. The advantage of this approach is
that we can use \cref{eq:intro:ae} to analyze the alternating enumerator of many
different fixed points simultaneously to find a nonvanishing fixed point.
Naturally, we still have to ensure that our fixed point has large treewidth.
To that end, we prove that for
certain graphs $H$ and groups $\Gamma$, the level of $A$ is lower bound for the treewidth.

\paragraph*{The Prime Power Case: Difference Graphs}

For a prime \(p\) and a positive integer \(m\), we write \(\field{p^m}\) for the finite
field with \(p^m\) elements. For $m = 1$, we write $\field{p} = \fragmentco{0}{p}$. The
elements that are invertible are denoted with
$\field{p^m}^\ast = \field{p^m} \setminus \{0\}$.

For a prime \(p\), we set $\field{p}^+ \coloneqq \{1, \dots, \lceil (p-1)/2 \rceil\}$,
which is a subset of \(\field{p}^*\) that contains exactly one of $x$ and $-x$ for every $x\in \field{p}^\ast$.
We generalize this notion to finite fields with a prime power number of elements and write
\(\field{p^m}^+\) for a set of elements that we obtain by including into it exactly one of
$x$ and $-x$ for every $x\in \field{p^m}^*$ (observe that if $p=2$, then $x=-x$ and hence
\(\field{p^m}^+=\field{p^m}^*\)).  We use \(\field{p^m}^+\) only in
situations where the specific choice of elements does not matter.
It is instructive to make explicit the following easy observation.

\begin{lemma}\label{ft:10.25-1}
    Let $p$ denote a prime and let $m > 0$ denote an integer.
    Then, \(|\field{p^m}^{\ast}| = p^m-1\) and  \(|\field{p^m}^{+}| \ge (p^m - 1)
    / 2 \).
\end{lemma}
\begin{proof}
    For the bound \(|\field{p^m}^{+}| \ge (p^m - 1) / 2\), we observe that in the special
    case \(p = 2\), we have \(x = -x\) and thus  \(|\field{2^m}^{+}| = |\field{2^m}^{\ast}|
    = (2^m - 1)\).
\end{proof}

Suppose that $\Phi$ is nontrivial on $k = p^m$. Let us consider the graph $K_{p^m}$ whose vertex
set is the finite field $\field{p^m}$ with $p^m$ elements. The {\em rotation subgroup}
$\rotgr{p^m}\subseteq \aut(K_{p^m})$ contains those permutations of $K_{p^m}$ that are
described by addition in $F_{p^m}$, that is,
\[
    \rotgr{p^m} \coloneqq \{ \varphi_c \in \auts{K_{p^m}} \mid c \in \field{p^m} \text{
    and }  \varphi_c( v ) = v + c \text{ for all \(v \in K_{p^m}\)}  \}.
\]

We observe that fixed points of the group $\rotgr{p^m}$ acting on $K_{p^m}$ are exactly
the difference graphs, as defined below.
\begin{restatable}{definition}{defcircgraph}
    For a prime \(p\), an integer $m > 0$, and a set \(A \subseteq \field{p^m}^+\),
    we define the \emph{difference graph} \(\cgr{p^m}{A}\) via
    \begin{align*}
        V( \cgr{p^m}{A} ) &\coloneqq \field{p^m} \quad\text{and}\quad
        E( \cgr{p^m}{A} ) \coloneqq \{ \{u, v\} \mid u, v \in \field{p^m},  (u - v) \in A \cup (-A) \},
    \end{align*}
    where $-A = \{-x \mid x \in A\}$. Observe that $\cgr{p^m}{A} = K_{p^m}$ whenever \(A = \field{p^m}^+\).

    The \emph{level} of a difference graph $\cgr{p^m}{A}$ is the cardinality of $A$; we write
    \(\hl{\cgr{p^m}{A}}\) for the level of \(\cgr{p^m}{A}\).
\end{restatable}

As $\cgr{p^m}{A}$ is $2|A|$-regular if $p \neq 2$ and $|A|$-regular if $p = 2$, it has
treewidth at least $|A|$. This means that it is sufficient to find a nonvanishing
difference graph (fixed point) with a high level.
To that end, we consider two different cases.

First, let us consider the case when $\Phi(\cgr{p^m}{A}) = 0$ for all fixed points
$\cgr{p^m}{A}$ with $\Hasselevel(A) \geq p^m - \sqrt{p^m}$.
We~introduce two vectors: the
vector showing the total value of $\Phi$ on the fixed points on each level, and the
analogous vector for $\aename{\Phi}$. The two vectors are related by an invertible linear
transform. A careful inspection of the matrix of the transformation shows that if $\Phi$ is
0 on fixed points of a level of at least $p^m - \sqrt{p^m}$, then fixed points of a level
of at least $\sqrt{p^m}$ cannot all have zero alternating enumerator.

\begin{restatable*}{lemma}{lemchibasic}\label{lem:chi:basic}
    Let \(H\) denote a graph, let \(\Gamma \subseteq \auts{H}\) denote a \(p\)-group, and
    let \(\Phi\) denote a graph property. Let \(c < \hl{H}\) denote a nonnegative integer.
    Suppose that we have \(\Phi(\emptyset) = 1\) and
    \(\Phi(F) = 0\) for every
    \(F \in \fps{H}\) with level \(\hl{F} > \hl{H} - c\).
    Then, there is a fixed point \(S \in \fps{H}\) with
    \(\ae{\Phi}{S} \not\equiv_p 0\) and \(\hl{S} \ge c\).
\end{restatable*}

Second, let us consider the case when $\Phi(\cgr{p^m}{A}) \neq 0$ for some fixed point
$\cgr{p^m}{A}$ with $\Hasselevel(A) \geq p^m - \sqrt{p^m}$.
Since $\cgr{p^m}{A}$ is true in $\Phi$, all edge-subgraphs of $\cgr{p^m}{A}$ are also true (here is the point
where we use that $\Phi$ is edge-monotonicity).
Further, each fixed point $\cgr{p^m}{B}$
that is isomorphic to an edge-subgraph of $\cgr{p^m}{A}$ (that is, $\cgr{p^m}{B}$ is
isomorphic to a graph $\cgr{p^m}{A^\ast}$ with $A^\ast \subseteq A$) satisfies $\Phi$ as
well. Since $\cgr{p^m}{A}$ has a high level, this ``starts an avalanche'' and at some
point, all fixed points below a certain level satisfy $\Phi$. We show that this
happens at level roughly $\sqrt{p^m}$. Now let us look at the  fixed point
$\cgr{p^m}{A}$ with $\Phi(\cgr{p^m}{A}) = 0$ such that all proper
edge-subgraphs satisfy $\Phi$.
We have that \(\hl{\cgr{p^m}{A}}\) is at least roughly $\sqrt{p^m}$ and
\cref{remark:chi:all:children:true} yields that $\cgr{p^m}{A}$ is nonvanishing.

\begin{restatable*}{lemma}{lemlevelhighprime}\label{lem:level:high:prime_case}
    Let $p$ denote a prime, let $m > 0$ denote an integer,
    and let \(\Phi\) denote an edge-monotone graph property that is nontrivial on \(p^m\).
    Further, write \(c\) and \(d\) for positive integers with \(c d \le |\field{p^m}^+|\).
    Suppose that there is a fixed point \(\cgr{p^m}{A} \in \fpb{\rotgr{p^m}}{K_{p^m}}\)
    with \(\hl{\cgr{p^m}{A}}
    \ge |\field{p^m}^+| - d\) and \(\Phi(\cgr{p^m}{A}) =
    1\).

    Then, there is a fixed point \(\cgr{p^m}{B} \subseteq \cgr{p^m}{A}\) with
    \(\hl{\cgr{p^m}{B}}
    \ge c\) and
    \(\ae{\Phi}{\cgr{p^m}{B}} \not\equiv_{p} 0\).
\end{restatable*}

Both cases together allow us to find a nonvanishing fixed point $\cgr{p^m}{A}$ with a
level of roughly $\sqrt{p^m}$, thus $\cgr{p^m}{A}$ has treewidth roughly $\sqrt{p^m}$.
Thus, if $\Phi$ is nontrivial on infinitely many prime powers, then we can use this insight to
find a sequence of nonvanishing graphs with unbounded treewidth.

\begin{restatable*}{theorem}{thmedgemonprimecase}\label{theo:edge_mono:prime}\label{theo:sqrt:lower:bound}
    Let $\Phi$ denote an edge-monotone graph property.
    \begin{itemize}
        \item  If \(\Phi\) is
             nontrivial on infinitely
             many prime powers, then $\NUM{}\indsubsprob(\Phi)$ is \w-hard.
        \item Further, assuming ETH, there is a universal constant $\alpha > 0$ (independent of \(\Phi\)) such that
            for any prime power $k \geq 3$ on which $\Phi$ is nontrivial,
            no algorithm (that reads the whole input)
            computes for every graph \(G\) the number
            $\NUM{}\indsubs{(\Phi, k)}{G}$ in time
            $O(|V(G)|^{\alpha \sqrt{k} / \log k})$.
            \qedhere
    \end{itemize}
\end{restatable*}

\paragraph*{General Case: Reduction to Prime Powers}

After showing \w-hardness for edge-monotone graph properties $\Phi$ that are  nontrivial
on infinitely many prime powers, our next goal is
to find a reduction from the general case (that is $\Phi$ is nontrivial on infinitely many
numbers) to this prime power case.
If $\Phi$ is nontrivial on a $k$, then we
analyze the largest prime power factor $q(k)$ of $k$. Set
$d \coloneqq k / q(k)$ and $p^m \coloneqq q(k)$. We \emph{join} $d$ copies of the complete
graph $K_{q(k)}$ into a graph $K_{q(k)}^d \coloneqq K_{q(k)} \JoinGraph \cdots \JoinGraph
K_{q(k)}$ which is isomorphic to $K_{d \cdot q(k)}$.
Further, we take the group-theoretical product of $d$ copies of $\rotgr{q(k)}$ to obtain the $p$-group
$\rotgr{q(k)}^d$ that acts on the vertices of~$K_{q(k)}^d$.

The advantage of this construction is that we can understand
the fixed points $\fpb{\rotgr{q(k)}^d}{K_{q(k)}^d}$ in terms of the fixed points
$\fpb{\rotgr{q(k)}}{K_{q(k)}}$.
Specifically, each fixed point $F \in \fpb{\rotgr{q(k)}^d}{K_{q(k)}^d}$
is made up of
\begin{itemize}
    \item a graph \(C\) with \(E(C) \subseteq \setn{d}\) and
        $d$ fixed points $\cgr{q(k)}{A_1}, \dots, \cgr{q(k)}{A_d} \in
        \fpb{\rotgr{q(k)}}{K_{q(k)}}$,
    \item where each pair of fixed points \(\cgr{q(k)}{A_i}\) and  \(\cgr{q(k)}{A_j}\) is
        either fully connected or not connected at all, depending on whether the edge
        \(\{ i, j \} \) is present in \(E(C)\).
\end{itemize}
Consult \cref{fig:example:fixed:points:prod} for visualizations of examples.

We use
$\metaGraph{C}{\cgr{q(k)}{A_1}, \dots, \cgr{q(k)}{A_d}}$ to denote said fixed points; we
have
\begin{align*}
    V(\metaGraph{C}{\cgr{q(k)}{A_1}, \dots, \cgr{q(k)}{A_d}}) &\coloneqq \setn{d} \times \field{q(k)} \\
    E(\metaGraph{C}{\cgr{q(k)}{A_1}, \dots, \cgr{q(k)}{A_d}}) &\coloneqq \{\{(i, v_i), (j, u_j)\} \mid \{i,
    j\} \in E(C) \text{ or } (i = j \text{ and } v_i - u_i \in A_i) \}.
\end{align*}
We show that $\fp(\rotgr{q(k)}^d, K^d_{q(k)}) = \{\metaGraph{C}{\cgr{q(k)}{A^1}, \dots,
\cgr{q(k)}{A^d}} \mid \text{$C$ is a $d$-vertex graph, }  A^i \subseteq \field{q(k)}^+ \}$ (see
\cref{theo:prod:fixed:points}). A very important observation is that if $C$ is not the
empty graph, then the fixed point $\metaGraph{C}{\cdots}$ contains $K_{q(k), q(k)}$ as a
subgraph, and has thus a treewidth of at lest $q(k)$.
Our discussion motivations the following notation.

\begin{restatable*}{definition}{defconsca}
    Let $\Phi$ denote an edge-monotone graph property and write \(\nt{\Phi}\) for the set
    of numbers on which \(\Phi\) is nontrivial.
    We say that \(\Phi\) is \emph{concentrated} on an integer \(k \in \nt{\Phi}\)
    if there is a graph $H$ on $k$ vertices with $\ae{\Phi}{H} \neq 0$
    and $H$ contains $K_{q(k), q(k)}$ as a subgraph.

    We say that \(\Phi\) is \emph{scattered} on an integer \(k \in \nt{\Phi}\)
    if it is not concentrated for \(\Phi\).
\end{restatable*}

If $\Phi$ is nontrivial, then we consider the following case distinction.
First, we assume that $\Phi$ is (nontrivial and) concentrated on infinitely many values~$k$.
For each concentrated $k$, by definition, we have a nonvanishing graph with a treewidth of at
least $q(k)$. We observe that $q(k) \geq c \log(k)$ for some constant $c > 0$ (see
\cref{lem:upper:bound:q}). This means that if $\Phi$ is (nontrivial and) concentrated
on infinitely many $k$, then we can use these values to construct a nonvanishing sequence with
unbounded treewidth. Now, \cref{lem:alpha:treewidth} shows that $\NUM{}\indsubsprob(\Phi)$ is
\w-hard.

Otherwise, $\Phi$ is (nontrivial and) scattered on infinitely many values $k$. For any such
$k$, let us consider a fixed point $\metaGraph{C}{\cgr{q(k)}{A_1}, \dots,
\cgr{q(k)}{A_d}}$ of minimum level that is zero in $\Phi$; as we have seen (\cref{remark:chi:all:children:true}), the
alternating enumerator is nonzero for this graph. As $k$ is scattered, we have that $C$ is
the empty graph, hence the fixed point is the disjoint union of $C^{A_1}_{p^m}$, $\dots$,
$C^{A_d}_{p^m}$. Assume without loss of generality that $A_d\neq \emptyset$ and let $H$
denote the disjoint union of $C^{A_1}_{p^m}$, $\dots$, $C^{A_{d-1}}_{p^m}$. Let us define
a graph property whose value is $\Phi(G \UnionGraph H)$ on $G$. Then, this property is
nontrivial on $q(k)$-vertex graphs: it is zero on $C^{A_d}_{p^m}$ and nonzero on
$\IS_{p^m}$. Thus for each $k$ on which $\Phi$ is scattered, we can construct a graph
property that is nontrivial on the prime power $q(k)$.

\begin{restatable*}{lemma}{lemmareduction}
\label{lem:reduction}
    Let $\Phi$ denote an edge-monotone graph property and write \(\nt{\Phi}\) for the set
    of numbers on which \(\Phi\) is nontrivial.
    For any number \(k \in \nt{\Phi}\) on which \(\Phi\) is scattered, there is a graph
    $H$ on $k-q(k)$ vertices such that the
    property \((\Phi - {H}) \coloneqq \{ G  \mid G \UnionGraph H \in \Phi \} \)
    is edge-monotone and nontrivial on $q(k)$.
\end{restatable*}

Now, the idea is to combine the infinitely many graph properties from \cref{lem:reduction}
into a single graph property $\Phi'$ that is nontrivial on infinitely many prime powers.
The problem $\NUM{}\indsubsprob(\Phi')$ is now \w-hard due to \cref{theo:edge_mono:prime}.
Further, we show how to compute $\NUM{}\indsubs{(\Phi', p^m)}{\star}$ with an oracle for
$\NUM{}\indsubs{(\Phi, k)}{\star}$ using the Inclusion-Exclusion principle (see
\cref{lem:inc:exc}).
Thereby we obtain a parameterized Turing reduction from $\NUM{}\indsubsprob(\Phi')$ to
$\NUM{}\indsubsprob(\Phi)$.
Since the problem $\NUM{}\indsubsprob(\Phi')$ is \w-hard, we thus obtain that
$\NUM{}\indsubsprob(\Phi)$ is \w-hard.
Combining both cases leads to \cref{theo:edge_mono}.

\thmedgemon*

\paragraph*{Tight Lower Bounds via Large Bicliques}

Lastly, in \cref{sec:lower:bounds}, we show stronger lower bounds assuming ETH. We write
$n$ for the number of vertices of the input graph. So far, we proved our lower bounds by
showing that we can solve $\homsprob(\{H\})$ using an oracle for $\NUM{}\indsubs{(\Phi,
|V(H)|)}{\star}$ whenever $\ae{\Phi}{H} \neq 0$. Further, we used that $\homsprob(\{H\})$
cannot be solved in time $O(n^{\alpha_{\homsprob} \cdot \tw(H) / \log \tw(H)})$ which
yields an $O(n^{\alpha_{\indsubsprob} \cdot \tw(H) / \log \tw(H)})$ lower bound for
$\NUM{}\indsubs{(\Phi, |V(H)|)}{\star}$. However, we would like to prove that
$\NUM{}\indsubs{(\Phi, k)}{\star}$ cannot be solved in time $O(n^{\gamma k})$ for a global
constant $\gamma > 0$ that does not depend on $k$. This is not possible with our current
method since we cannot get rid of the $\log$ factor in the denominator of the exponent. A
lower bound of $O(n^{\gamma k})$ for $\NUM{}\indsubs{(\Phi, k)}{\star}$ would also prove that
we cannot solve $\NUM{}\indsubsprob(\Phi)$ in time $f(k) n^{o(k)}$ for any computable function
unless ETH fails. This is tight in the sense that there is a brute-force algorithm that
solves $\NUM{}\indsubsprob(\Phi)$ in time $O(f(k) n^k)$, which is achieved by simply iterating
through induced subgraphs of size $k$.

To achieve this goal, we use a reduction from \cite[Theorem 1]{alge} that uses
nonvanishing graphs that contain large bicliques $K_{k,k}$. Moreover, instead of starting
the reduction from $\NUM{}\homsprob(\{H\})$, we start from the $k$-{\clique} problem (which is
the problem of deciding whether an input graph $G$ contains a $k$-clique). It is known
that \clique\ has no algorithm with running time $f(k)n^{o(k)}$ for any computable
function $f$, unless ETH fails \cite{param_algo}. However, we need a stronger form of this
statement saying that there exists a constant $\alpha$ such that $k$-{\clique} cannot be
solved in time $O(n^{\alpha k})$ unless ETH fails; we prove this statement in \cref{sec:tight:bounds:bicliques}.
The idea is to find a reduction from $k$-{\clique} to $\NUM{}\indsubs{(\Phi, k)}{\star}$ such
that an algorithm computing $\NUM{}\indsubs{(\Phi, k)}{\star}$ in time $O(n^{\gamma k})$ could
be used to compute $k$-{\clique} in time $O(n^{\alpha k})$  which is not possible unless
ETH fails. We reprove the reduction of \cite[Theorem 1]{alge} in \cref{app:use:biclique}
in the form that we need.

\begin{restatable*}[Modification of \cite{alge}]{theorem}{thmlowerboundbicliques}\label{theo:lower:bound:bicliques}
    There is a global constant $\beta > 0$ and a positive integer $N$ such that for all
    graph properties $\Phi$, functions $h$, numbers $k$ with
    \begin{itemize}
        \item $h(k) \geq N$
        \item there is a graph $F$ with $k$ vertices and $\ae{\Phi}{F} \neq 0$,
        \item and $F$ contains $K_{h(k), h(k)}$ as a subgraph
    \end{itemize}
    there is no algorithm (that reads the whole input) that for every \(G\) computes
    $\NUM{}\indsubs{(\Phi, k)}{G}$ in time $O(|V(G)|^{\beta h(k)})$ unless ETH fails.
\end{restatable*}

This means that we can prove stronger lower bounds for $\NUM{}\indsubs{(\Phi, k)}{\star}$ by
finding nonvanishing graphs that contain large bicliques. As discussed below, for a prime
$p$ and $m\ge 1$,  we can construct a $p$-group $\sylow_{p^m}$ such that we can find a
fixed point of $\sylow_{p^m}$ in $K_{p^m}$ that is nonvanishing and contains $K_{p^{m-1},
p^{m-1}}$ as a subgraph. If $p$ is small (constant) compared to $p^m$, then the size of
this biclique is approximately the same as the number of vertices
 (however, this means that we cannot use these fixed points if $\Phi$ is nontrivial only on prime numbers).
Therefore, \cref{theo:lower:bound:bicliques} allow us to prove tight lower bounds whenever
$\Phi$ is nontrivial on a prime power $k = p^m$.

\thetightlowerbound*

\paragraph*{Finding Nonvanishing Graphs with Large Bicliques}

To find nonvanishing graphs with large bicliques, we use a new $p$-group $\sylow_{p^m}$,
defined as follows. We still use the complete graph $K_{p^m}$, but with vertex set $\{0, 1,
\dots, p-1\}^m$.  For all $m$-tuples $\sylelm = (\varphi_0, \dots,
\varphi_{m-1})$ of functions with $\varphi_j \colon \fragmentco{0}{p}^{j} \to \fragmentco{0}{p} $, we
define the following function on  $V(K_{p^m})$.\footnote{We write $\varphi_0$ for
$\varphi_0(\emptyset)$ since $\varphi_0$ is a function that is defined on a single element}
\[\sylelm(x_1, \dots, x_m) = (x_1 + \varphi_0, x_2 + \varphi_1(x_1), x_3 + \varphi_2(x_1,
x_2), \dots, x_m + \varphi_{m-1}(x_1, \dots, x_{m-1})) \] where all computations are done
modulo $p$. It is easy to see that each $\sylelm$ is a bijection on the vertex set and
therefore in $\aut(K_{p^m}) \cong \sym{p^m}$, where $\sym{p^m}$ is the symmetric group on
$p^m$ elements. We denote by
$\sylow_{p^m}$ the set of all $m$-tuples of functions $\sylelm = (\varphi_0, \dots,
\varphi_{m-1})$ with $\varphi_j \colon \fragmentco{0}{p}^{j} \to \fragmentco{0}{p}$.

It is easy to check that $|\sylow_{p^m}|$ is a $p$-power. To describe the fixed points of
$\sylow_{p^m}$ in $K_{p^m}$ we need a concept that is known as the \emph{lexicographic
product} of the graphs. We use the following standard definition~\cite[Page 22]{graph_theory_harary}.

\begin{restatable*}{definition}{deflexgraph}
    For graphs $G_1, \dots, G_m$,
    we define their \emph{lexicographic product} $G_1 \wrGraph \cdots \wrGraph G_m$
    via
    \begin{align*}
        V(G_1 \wrGraph \cdots \wrGraph G_m) &\coloneqq
        V(G_1) \times \cdots \times V(G_m) \quad\text{and}\\
        E(G_1 \wrGraph \cdots \wrGraph G_m) &\coloneqq
        \{ \{(u_1, \dots, u_m), (v_1, \dots, v_m) \}\\&\hskip3em
        \mid \text{there is an $i \in \setn{m}$ with  $u_j = v_j$ for all $j < i$ and $\{u_i,
        v_i\} \in E(G_i)$}\}.
        \tag*{\qedhere}
    \end{align*}
\end{restatable*}

\begin{restatable*}{lemma}{lemfixedpointslexprod}\label{lemma:fixed:points:p^m}
    For any prime \(p\) and any positive integer \(m\), we have
    \begin{align*}
        \fpb{\sylow_{p^m}}{K_{p^m}}
        = \{\cgr{p}{A_1} \wrGraph \cdots \wrGraph \cgr{p}{A_m} \mid A_i \subseteq \field{p}^+\}.
        \tag*{\qedhere}
    \end{align*}
\end{restatable*}

One important observation is that a fixed point  $\cgr{p}{A_1} \wrGraph \dots \wrGraph
\cgr{p}{A_m}$ has a large biclique if there is a small number $i$ with $A_i
\neq \emptyset$.
We capture this observation by introducing the \emph{empty-prefix} of a graph.
The {empty-prefix} of $H = \cgr{p}{A_1} \wrGraph \dots \wrGraph \cgr{p}{A_m}$
is the smallest index \(i\) with $A_i \neq \emptyset$, minus one;
that is, $\wrLevel(A_1, \dots, A_m) \coloneqq i-1$, where $i$ is the smallest
index with $A_i \neq \emptyset$. We observe that a graph with a low empty-prefix contains
a large biclique as a subgraph.

\begin{restatable*}{lemma}{lemwreathlevel}\label{lem:treewidth:wreath:level}
    Let $p$ denote a prime number and let $m$ denote a positive integer.
    For each \(i \in \setn{m}\), let $A_i \subseteq \field{p}^+$ denote a subset and set
    $A \coloneqq (A_1, \dots, A_m)$.
    Then, $\cgr{p}{A_1} \wrGraph \cdots \wrGraph \cgr{p}{A_{m}}$ contains
    $K_{p^{m - 1 - \wrLevel(A)}, p^{m - 1 - \wrLevel(A)}}$ as a subgraph.
\end{restatable*}

Suppose that our graph property is nontrivial on $p^m$ for $m \geq 2$. If
we find a fixed point $H = \cgr{p}{A_1} \wrGraph \dots \wrGraph
\cgr{p}{A_m}$ with a minimal empty-prefix of $\wrLevel(A_1, \dots, A_m) = 0$,
then we know that the treewidth of $H$ is at least $p^{m-1}$. Thus,
our goal is to find a fixed point $H$ with $\wrLevel(A_1, \dots, A_m) = 0$
and $\ae{\Phi}{H} \not \equiv_p 0$.

To that end, we prove that a fixed point with a high empty-prefix is
always isomorphic to an edge-subgraph of a fixed point with a low
empty-prefix, which allows us to always consider fixed points with the low
empty-prefix. To be more precise, we show that each fixed point
$\cgr{p}{\emptyset} \wrGraph \dots \cgr{p}{\emptyset} \wrGraph
\cgr{p}{A_1} \wrGraph \dots \wrGraph \cgr{p}{A_{m-j}}$ is isomorphic to an
edge-subgraph of $\cgr{p}{A_1} \wrGraph \dots \wrGraph \cgr{p}{A_{m}}$. Intuitively, we
achieve thus by defining an isomorphism that pushes the edges of each $\cgr{p}{A_i}$ \emph{one level down}.

\begin{restatable*}{lemma}{lemwreathsubiso}\label{lem:wreath:sub:iso}
    Let $p$ denote a prime number and let $m$ denote a positive integer.
    For each \(i \in \setn{m}\), let $A_i \subseteq \field{p}^+$ denote a subset.
    Then, for all $j \in \setn{m}$, the graph
    $\cgr{p}{\emptyset} \wrGraph \cdots\cgr{p}{\emptyset} \wrGraph
    \cgr{p}{A_1} \wrGraph \cdots \wrGraph \cgr{p}{A_{m-j}}$ is isomorphic to an
    edge-subgraph of $\cgr{p}{A_1} \wrGraph \cdots \wrGraph \cgr{p}{A_{m}}$.
\end{restatable*}

Now, we are ready to show that there is always a nonvanishing fixed point of
$\sylow_{p^m}$ in $K_{p^m}$ that contains $K_{p^{m-1}, p^{m-1}}$ as a subgraph. We show
this by finding a fixed point $H$ with $\Phi(H) = 0$, an empty-prefix of $0$, and
$\Phi(\Tilde{H}) = 1$ for all $\Tilde{H} \subseteq H$ (that is, fixed points $\Tilde{H}$ that
lie in $H$). To do this, we consider the first level $i$ such that there is a fixed point
of level $i$ that does not satisfy $\Phi$. Observe that the only fixed point of level $0$
is the independent set and $\Phi(\IS_{p^m}) = 1$ thus $i \geq 1$. Further, if we assume
that all fixed points $\cgr{p}{A_1} \wrGraph \dots \cgr{p}{A_m}$ of level $i$ with
$\wrLevel(A_1, \dots, A_m) = 0$ satisfy $\Phi$, then we use \cref{lem:wreath:sub:iso}
to show that each fixed point of level $i$ satisfy $\Phi$, a contradiction. Thus, there is
a fixed point $H$ with $\Phi(H) = 0$, an empty-prefix of $0$, and $\Phi(\Tilde{H}) = 1$
for all $\Tilde{H} \subseteq H$. Lastly, we use
\cref{remark:chi:all:children:true,lem:treewidth:wreath:level} and to show that
$\ae{\Phi}{H} \neq 0$ and that $H$ contains $K_{p^{m-1}, p^{m-1}}$ as a subgraph.

\begin{restatable*}{theorem}{thmedgemonprimepowerbicliques}\label{theo:edge_mono:prime:power:bicliques}
    Let $\Phi$ denote an edge-monotone graph property that is nontrivial on a prime power
    $p^m$, then there is a nonvanishing fixed point of $\sylow_{p^m}$ in $K_{p^m}$ that
    contains $K_{p^{m-1}, p^{m-1}}$ as a subgraph.
\end{restatable*}

\section{Additional Preliminaries}

\paragraph*{Numbers and Sets}

For a natural number $n$,
we write $\setn{n}$ for the set $\{1, \dots, n\}$.
For natural numbers \(a\) and \(b\), we write $\fragmentco{a}{b}$ for the set
$\{a, \dots, b - 1\}$.

Next, we write $\binom{n}{k}$ for the binomial coefficient
and we set $\binom{n}{k} \coloneqq 0$ whenever $k < 0$ or $n < k$.
If $A$ is a set then $\binom{A}{k}$ is the set of all subsets $B \subseteq A$ with size $k$.

For a subset \(B \subseteq R\) of a field $R$, we set \(-B \coloneqq \{-b \mid b \in B\}\);
for a set \(B \subseteq R\) and an element \(\lambda \in R\) we set
\(\lambda B \coloneqq \{ \lambda b \mid b \in B\}\).

We use $a \equiv_m b$ as a shorthand for $a \equiv b \bmod m$.

\paragraph*{Graphs}

We consider only {simple} graphs, that is undirected graphs that have
neither weights, loops, nor parallel edges.
We write $\mathcal{G}$ for the set of all (simple) graphs and \(\graphs{n}\) for the
set of all (simple) graphs with the vertex set \(\setn{n}\).

For a graph \(G\), we write \(V(G)\) for its vertices and we write \(E(G)\) for its edges.
Given a subset of edges $A \subseteq E(G)$, we write $\ess{G}{A}$ for the graph
with vertex set $V(\ess{G}{A}) \coloneqq V(G)$ and edge set $E(\ess{G}{A}) \coloneqq A$.
We say that \(\ess{G}{A}\) is an \emph{edge-subgraph} of $G$ and
we write $\edgesub(G)$ for the set of all edge-subgraphs of $G$.
For a set $X \subseteq V(G)$, we write $G \setminus X$ for the graph that is obtained from
\(G\) by removing the vertices in \(X\) and all edges with at least one endpoint in \(X\).

For two graphs \(G\) and  \(H\) with a common vertex set \(V\), we write \(G \cup H\)
for the graph on \(V\) with edges \(E(G) \cup E(H)\).

We write $\IS_{k}$ for the independent set with $k$ vertices.
Further, we write \(K_n\) for the complete graph on \(n\) vertices
and we write \(K_{n, m}\) for the complete bipartite graph on \(n + m\) vertices.

We write $\tw(G)$ for the \emph{treewidth} of a graph $G$.
As we use the treewidth of a graph in a black-box manner, we refer an interested reader to
\cite[Chapter 7.2]{param_algo},  for a formal definition.
Intuitively, the {treewidth} of a graph measures how far away a given graph is from being a tree.
For example, the {treewidth} of a tree is $1$; the treewidth of a complete graph on \(n\)
vertices is \(n-1\).

\paragraph*{Group Theory and Morphisms between Graphs}
\label{sec:group:theory}

For a finite set \(X\) of size \(n\), the set of all bijections \(X \to X\) and the
function composition together form a group which we call the symmetric group \(\sym{n}\).
We also write \(\sym{X}\) if we wish to emphasize the set \(X\).

A group $\Gamma$ is a \emph{permutation group} if \(\Gamma\) is a subgroup of the
symmetric group \(\sym{X}\) for some finite set \(X\).
A permutation group \(\Gamma \subseteq \sym{X}\) is
\emph{transitive} if every \(x \in X\) can be mapped to any other element \(y \in X\) via
some \(g \in \Gamma\).
Finally, a \emph{$p$-group} is a finite group \(\Gamma\) that has an order that
is a power of $p$ (that is, the number of elements of \(\Gamma\) is a power of \(p\)).

A graph \emph{homomorphism} $h : V(H) \to V(G)$ from $H$ to $G$ is a function between the
vertex sets of two graphs that preserves adjacencies (but not necessarily
non-adjacencies), that is, \(h\) maps the vertices of every edge $\{u, v\} \in E(H)$
to vertices $\{h(u), h(v)\} \in E(G)$.
We write $\homs{H}{G}$ for the set of all homomorphisms from $H$ to $G$.

We say a {homomorphism} $h : V(H) \to V(G)$ is an
\emph{isomorphism} if $h$ is a bijection on the vertex sets and $h^{-1} : V(H) \to V(G)$
also defines a homomorphism
(that is, $\{u, v\} \in E(H)$ if and only if $\{h(u), h(v)\} \in E(G)$).
We say that two graphs \(G\) and \(H\) are \emph{isomorphic}, denoted by \(G \cong H\)
if there is an {isomorphism} between them.
An \emph{automorphism} of $G$ is an {isomorphism} from $G$ to $G$.
We write \(\auts{G}\) for the set of all \emph{automorphism} of $G$.

For each graph $G$, the set of automorphisms $\auts{G}$ forms a group with composition as
the group operation. Observe that the automorphism group of a clique \(\auts{K_n}\) is just
the symmetric group \(\sym{n}\).

For primes \(p\) and prime powers \(p^m\), the group \(\auts{K_{p^m}}\) contains a useful
subgroup with order \(p^m\):
in particular, \(\auts{K_{p^m}}\) contains automorphisms \(\varphi_c\) of \(G\) that
``rotate'' the vertices of \(G\), that is, automorphisms that send every vertex \(i\)
to vertex \({i + c}\) for
some \(c \in \field{p^m}\) (where we identify the vertices of \(K_{p^m}\) with the elements of
\(\field{p^m}\)); we
say that \(\varphi_c\) is a \emph{rotation (by \(c\))} of \(G\).
In particular, we write $\rotgr{p^m}$ for the subgroup of all rotations of
\(\auts{K_{p^m}}\), that is
\[
    \rotgr{p^m} \coloneqq \{ \varphi_c \in \auts{K_{p^m}} \mid c \in \field{p^m} \text{
    and }  \varphi_c( v ) = v + c \text{ for all \(v \in K_{p^m}\)}  \}.
\]
Observe that \(\rotgr{p^m}\) is a group with \(p^m\) elements.

Let \(G\) denote a graph and consider a subgroup \(\Gamma \subseteq \auts{G}\).
Any \(g \in G\) is a bijection on \(V(G)\), which we may interpret as \(g\) permuting the
vertices of \(G\).
Thus, we may also interpret $g$ as an operation on the edges of \(G\).
However, as we wish to use results from the literature (such as
the Orbit-stabilizer Theorem), we need to describe this operation on the edges of \(G\)
using the language of group actions.

Formally, we can turn the operation on vertices into a {\em
group action} $\cdot : \Gamma\times E(G)\to E(G)$ that tells us how each member of
$\Gamma$ moves the edges of $G$.
Specifically, we define $g\cdot \{u,v\} \coloneqq \{g(u),g(v)\}$.
We interpret said group action \(\cdot\) multiplicatively and typically write just \(g\{u, v\}\).

Extending the previous group action, \(\Gamma\) also acts on
edge-subgraphs of \(G\) via
\begin{align*}
    V(g\cdot A) &\coloneqq V(A),\; \text{and}\\
    E(g\cdot A) &\coloneqq \{\{g(u), g(v)\} \mid \{u, v\} \in E(A)\},
\end{align*}
for each \(g \in \Gamma\) and each \(A \in \edgesub(G)\).
Again, we interpret said group actions \(\cdot\) multiplicatively and typically write just \(gA\).

For each $\{u, v\} \in E(G)$ the set $\Gamma \cdot \{u, v\} \coloneqq \{g \cdot \{u, v\}
\mid g \in \Gamma\}$ is the \emph{orbit} of $\{u, v\}$.
Two different edges either have disjoint orbits or equal orbits.
We write \(E(G)/\Gamma\) to denote the set of all orbits of \(\cdot\).
Recall that \(E(G)/\Gamma\) forms a partition of \(E(G)\).
In a slight abuse of notation, for an orbit \(O\), we also write \(O\) for the
edge-subgraph \(\ess{G}{O}\); similarly, we use \(E(H)/\Gamma\) to denote the set of all
such edge-subgraphs.

We say that an edge-subgraph $A \in \edgesub(G)$ is a \emph{fixed point} of $\Gamma$ in $G$
if $gA = A$ for all $g \in \Gamma$.
We write $\fps{G}$ for the set of all fixed points of $\Gamma$ in $G$.

As it turns out, a graph \(H\) being a fixed point of \(\Gamma\) in \(G\) is a strong property that is
very useful to us. For now and as a first useful observation, we see that \(H\) inherits key
properties from \(G\).

\begin{lemma}\label{remark:subgroup:fixpoint}
    Let \(G\) denote a graph and let $\Gamma \subseteq \auts{G}$ denote a group.
    For any $H \in \fps{G}$ all of the following hold.
    \begin{enumerate}[(1)]
        \item We have $\Gamma \subseteq \auts{H}$.
        \item Any fixed point of \(\Gamma\) in \(H\) is also a fixed point of \(\Gamma\) in
            \(G\) and we have
            $\fps{H} = \{A \in \fps{G} \mid E(A) \subseteq E(H)\}$.
    \end{enumerate}
\end{lemma}
\begin{proof}
    For (1), we first recall that as \(H\) is a fixed point of \(\Gamma\) in \(G\), we
    have \(gH = H\) for every \(g \in \Gamma\).
    In particular, this means that every \(G\)-automorphism \(g \in \Gamma\) is also an
    automorphism of \(H\).
    This in turn yields the claim.

    For (2), we first see that by (1), \(\fps{H}\) is indeed well-defined.
    Now, consider a fixed point \(A\in\fps{H}\). By definition, we have \(E(A)
    \subseteq E(H) \subseteq E(G)\) and \(gA = A\) for every \(g\in\Gamma\), so \(A\) is
    indeed also a fixed point of \(\Gamma\) in \(G\).

    For the other direction, observe that a fixed point \(A \in \fps{G}\) with
    \(E(A) \subseteq E(H)\) is just a fixed point of \(\Gamma\) in \(H\), which completes
    the proof.
\end{proof}

Next, we observe that the notion of fixed points is compatible with set operations (on the
edges of the underlying graph).

\begin{lemma}\label{lem:setopfp}
    Let \(G\) denote a graph and let $\Gamma \subseteq \auts{G}$ denote a group.
    For any $H_1, H_2 \in \fps{G}$ all of the following hold.
    \begin{enumerate}[(1)]
        \item We have \(H_1 \cup H_2 \in \fps{G}\).\footnote{Slightly abusing
            notation, we use set operation for fixed points to mean the same operation on
        the edge sets of the corresponding edge-subgraphs.}
        \item If \(H_2 \subseteq H_1\), then we have \(H_1 \setminus H_2 \in \fps{G}\).
    \end{enumerate}
\end{lemma}
\begin{proof}
    For (1), we wish to show that for any \(g \in \Gamma\), we have \(g(H_1 \cup H_2) = H_1 \cup H_2\).
    To that end, observe that we have \(gH_1 = H_1\) and \(gH_2 = H_2\); that is, \(g\) is
    an automorphism when restricted to either \(H_1\) or \(H_2\).
    As unions of automorphisms are automorphisms, we obtain the claim.

    For (2), as in the proof of (1), observe that any automorphism \(g  \in \Gamma\)
    (which is also an automorphism on~\(H_1\) as \(H_1\) is a fixed point)
    stays an automorphism when restricted to \(H_2\).
    In particular, this means that any such automorphism decomposes
    into an automorphism on \(H_2\) and an automorphism on \(H_1\setminus H_2\).
\end{proof}

\subparagraph*{Counting Problems and Parameterized Complexity}
A \emph{parameterized (counting) problem} consists of a function $P \colon \Sigma^\ast \to \nat$ and
a computable parameterization $\kappa \colon \Sigma^\ast \to \nat$.
A parameterized problem \((P, \kappa)\) is called \emph{fixed-parameter
tractable} (FPT) if there is a computable function $f$ and a deterministic algorithm
\(\mathbb{A}\) such that \(\mathbb{A}\) computes $P(x)$ in time  $f(\kappa(x)) |x|^{O(1)}$
for all $x \in \Sigma^\ast$.

A parameterized Turing reduction from $(P, \kappa)$ to $(P', \kappa')$ is a deterministic
FPT algorithm with oracle access to $P'$ that computes $P(x)$ such that there is a
computable function $g$ with the property that $\kappa'(y) \leq g(\kappa(y))$ holds for
each oracle access to $P'$. We write $A \fpt B$ to denote that there is a parameterized Turing
reduction from $A$ to $B$.

In the counting version $\NUM{}P$ of a decision problem $P$, the task is to compute  the
number of valid solutions for a given input $x$.
Of special importance is the counting problem \NUM{}\clique,
which gets as input a graph $G$ and a natural number $k$.
The output is the number of induced
subgraphs of $G$ of size $k$ which form a $k$-clique.
We parameterize \(\NUM{}\clique\) by $\kappa(G, k) \coloneqq k$.
We are interested in barriers to obtain fast algorithms for \(\NUM{}\clique\) and for
parameterized counting problems in general.

A problem $(P',\kappa')$ is \w-hard if there is parameterized Turing reduction
from $(\NUM{}\clique, \kappa)$ to $(P', \kappa')$.
It is widely believed that \(\NUM{}\clique\) has no FPT algorithm \cite{path_cycle, CHEN2005216}, ; hence
\w-hardness rules out FPT algorithms based on said belief.

For more fine-grained lower bounds, we also rely on the \emph{Exponential Time Hypothesis}
(ETH).

\begin{ethy}[{\cite[Conjecture 14.1]{param_algo} \cite{IMPAGLIAZZO2001512}}]
    There is a positive real value \(\varepsilon > 0\) such that
    the problem 3-SAT cannot be solved in time
    $O^\ast(2^{\varepsilon n})$,  where $n$ is the number of
    variables used in the formula.
\end{ethy}

The Exponential Time Hypothesis is a stronger hypothesis and in fact implies
there are no FPT algorithms for \w-hard problems. For example, it can be used
to show that no algorithm solves ${\clique}$ in time $f(k)n^{o(k)}$,
where $n$ is the number of vertices. Observe, that this implies that there
is no FPT algorithm for the \w-hard ${\clique}$ (see \cref{theo:k-clique:ETH}).

\section{Fixed Points and the Alternating Enumerator}
\label{sec:alt:en:hasse}
\label{chap:prop:chi}

In this section, we introduce the framework of how we can use fixed points of groups that act on graphs to show that the
alternating enumerator is nonvanishing on certain graphs. In particular, the key technical
result of this section is \cref{lem:chi:basic}, which states a duality between the highest
nonzero levels in $\Phi$ and $\aename{\Phi}$. That is, it is not possible that both levels
are low: their sum must be at least $n$.

\subsection{Fixed Points as Unions of Orbits and Sub-points of Fixed Points}

As a first step, we discuss how to construct any fixed point from a small set of basic
building blocks.

\remarkfixunion
\begin{proof}
    We readily confirm that the empty set is indeed a fixed point; we may obtain the empty
    set as the empty union of orbits.

    Next, we turn to single orbits and verify that,
    indeed, they are fixed points as well.

    \begin{claim}\label{cl:9.29-1}
        We have \(E(H)/\Gamma \subseteq \fps{H}\).\footnote{Recall that we use \(O\in
        E(H)/\Gamma\) to denote both the edge set as well as the edge-subgraph.}
    \end{claim}
    \begin{claimproof}
        Consider an orbit \(O \in E(H)/\Gamma\) and an automorphism \(g \in \Gamma\).
        We intend to show that \(gO = O\), which suffices to prove the claim.

        To that end, by definition of an orbit, we have \(ge \in O\) for every edge \(e \in
        O\). As \(g\) is an automorphism of \(G\), no two different edges
        from \(O\) are mapped to the same edge; this yields the claim.
    \end{claimproof}

    Observe that \cref{cl:9.29-1,lem:setopfp} together yield that the union of orbits is
    indeed a fixed point.

    Finally, we show that we can always split off some orbit from a (non-empty) fixed point.

    \begin{claim}\label{cl:9.29-2}
        (The edge set of) any fixed point \(F \in \fps{H}\)
        can be obtained as the (disjoint) union of and orbit \(O \in
        E(H)/\Gamma\) and a fixed point from \(\fps{H}\).
    \end{claim}
    \begin{claimproof}
        Consider an arbitrary edge \(e \in E(F)\) and let \(O\) denote the corresponding
        orbit in \(E(H)/\Gamma\). By definition of \(O\) and \(F\), we have \(O\subseteq
        E(F)\).
        Combining \cref{cl:9.29-1,lem:setopfp} yields the claim.
    \end{claimproof}

    Iterating \cref{cl:9.29-2} and recalling that the orbits \(E(H)/\Gamma\) partition
    \(E(H)\) completes the proof.
\end{proof}

\begin{remark}
    Observe that \cref{remark:fixed:point:union} allows us to strengthen
    \cref{lem:setopfp} to all set operations:
    the orbits of different edges are either equal or
    disjoint, which for any $H_1, H_2 \in \fps{G}$ directly yields all of the following.
    \begin{enumerate}[(1)]
        \item We have \(H_1 \cup H_2 \in \fps{G}\).
        \item We have \(H_1 \cap H_2 \in \fps{G}\).
        \item We have \(H_1 \setminus H_2 \in \fps{G}\).
            \qedhere
    \end{enumerate}
\end{remark}

\cref{remark:fixed:point:union} and in particular \cref{cl:9.29-2} induce an order of the fixed points in
\(\fps{H}\): we say a fixed point \(F\) is a \emph{sub-point} of another fixed point
\(G\) if we can obtain \(G\) from the union of \(F\) and (potentially multiple) orbits in
\(E(H)/\Gamma\).

\begin{definition}\label{def:10.4-1}
    Let \(H\) denote a graph and let \(\Gamma \subseteq \auts{H}\) denote a group.
    Further, let \(\Gamma\) act on \(E(H)\) and write \(E(H)/\Gamma\) for the set of all
    resulting orbits.
    Finally, let \(\Gamma\) act on edge-subgraphs of \(H\) and write \(\fps{H}\) for
    the set of all resulting fixed points of \(\Gamma\) in \(H\).

    For a fixed point \(F\in \fps{H}\), its \emph{orbit factorization} \(\od{F}\) is the
    unique subset of \(E(H)/\Gamma\) whose union is~\(F\):
    \[
    F = \bigcup \od{F}.
    \]
    The \emph{level} of \(F\), denoted by \(\hl{F}\), is the size of the orbit
    factorization of \(F\)
    \[
    \hl{F} \coloneqq |\od{F}|.
    \]
    Finally, for two fixed points \(F_1, F_2  \in \fps{H}\), we say that \(F_2\) is a \emph{sub-point}
    of \(F_1\), denoted by \(F_2 \subseteq F_1\), if the orbit factorization of \(F_2\) is
    a subset of the orbit factorization of \(F_1\). If the inclusion is strict, we say
    that \(F_2\) is a \emph{proper sub-point} of \(F_1\).
\end{definition}

Observe that for two fixed points \(F_1\) and \(F_2\) of some group \(\Gamma\) with \(F_2
\subseteq F_1\), we may equivalently write \(F_2 \in \fps{F_1}\).

The next lemma makes it possible to group the sub-points of a fixed point according to their level.

\begin{lemma}\label{lem:10.4-2}
    Let \(H\) denote a graph and let \(\Gamma \subseteq \auts{H}\) denote a \(p\)-group.

    For any fixed point \(F \in \fps{H}\), we have \[
        (-1)^{\NUM{}E(F)} \equiv_p (-1)^{\hl{F}}.
    \]
\end{lemma}
\begin{proof}
    Write \(\od{F} = \{ O_1, \dots, O_{\hl{F}} \} \) for the orbit factorization
    of \(F\).
    By the Orbit-stabilizer Theorem, the size of each $O_i$ is a divider of the group
    order of $\Gamma$.
    Thus, $|O_i|$ is always odd if $p$ is an odd prime number, which implies the claim.

    For $p = 2$, we have $(-1) \equiv_2 1$, which immediately yields the claim.
\end{proof}

Using \cref{lem:10.4-2}, we can obtain yet another expression for the alternating
enumerator.

\begin{corollary}\label{cor:10.7-1}
    Let \(H\) denote a graph, let \(\Gamma \subseteq \auts{H}\) denote a \(p\)-group, and
    let \(\Phi\) denote a graph property.

    We have \[
        \ae{\Phi}{H} \equiv_{p}
        \sum_{A \in \fps{H}} \Phi(A) (-1)^{\hl{A}}.
    \]

    Further, for any fixed point \(A\in\fps{H}\), we have \[
        \ae{\Phi}{A} \equiv_{p}
        \sum_{B \subseteq A} \Phi(B) (-1)^{\hl{B}}.
        \tag*{\lipicsEnd}
    \]
\end{corollary}

\subsection{Nonvanishing of the Alternating Enumerator via Sub-points}

In the next step, we use sub-points to analyze the alternating enumerator.
As a first result, we use \cref{cor:10.7-1} to obtain that the alternating enumerator is
nonvanishing on fixed
points that are ``minimally vanishing'' for the corresponding graph property.

\aesubpointstrue
\begin{proof}
    By definition, the fixed point \(A\) has exactly \(\binom{\hl{A}}{i}\)
    sub-points that have a level of exactly \(i\).
    Starting from \cref{cor:10.7-1}, we rewrite the alternating enumerator and obtain
    \begin{align*}
        \ae{\Phi}{A}
        &\equiv_p \sum_{ B \subseteq A } \Phi(B) (-1)^{\hl{B}} \\
        &\equiv_p -(-1)^{\hl{A}} + \sum_{ B \subseteq A } (-1)^{\hl{B}}
        & \text{(\(\Phi(A) = 0\) and \(\Phi(B) = 1\) for all \(B \subsetneq A\))}\\
        &\equiv_p  -(-1)^{\hl{A}} + \sum_{i = 0}^{\hl{A}} \binom{{\hl{A}}}{i} (-1)^i
        & \text{(group by level)}\\
        &\equiv_p -(-1)^{\hl{A}} + (-1 + 1)^{\hl{A}}
        & \text{(Binomial Theorem)}.
    \end{align*}
    In total, this yields the claim.
\end{proof}

For our next steps toward the proof \cref{lem:chi:basic}, it is instructive to group sub-points by their level.
In particular, we are interested in how many sub-points have a given graph property and we
are interested in the sum of the alternating enumerators of said sub-points.

\begin{definition}\label{def:10.8-1}
    Let \(H\) denote a graph, let \(\Gamma \subseteq \auts{H}\) denote a \(p\)-group, and
    let \(\Phi\) denote a graph property.

    For every level \(i \in \fragment{0}{\hl{H}}\), we write \(\phisum{i}\) for
    the number (modulo \(p\)) of sub-points of \(H\) with level \(i\) that satisfy
    \(\Phi\); that is,
    \[
        \phisum{i} \coloneqq \sum_{\sunderset{\hl{A} = i}{A\in \fps{H}}} \Phi(A)
        \bmod{p}.
    \]
    We write \(\phisumvec \coloneqq (\phisum{i})\) for the \((\hl{H} + 1)\)-dimensional vector
    that consists of the values \(\phisum{i}\).

    Further, for every level \(i \in \fragment{0}{\hl{H}}\), we write
    \(\aesum{i}\) for the sum (modulo \(p\)) of the alternating enumerators of the sub-points of \(H\)
    with level \(i\); that is,
    \[
        \aesum{i} \coloneqq \sum_{\sunderset{\hl{A} = i}{A\in \fps{H}}} \ae{\Phi}{A}
        \bmod{p}.
    \]
    We write \(\aesumvec \coloneqq (\aesum{i})\) for the \((\hl{H} + 1)\)-dimensional vector
    that consists of the values \(\aesum{i}\).
\end{definition}

Let us take a step back and discuss \cref{def:10.8-1} for a bit.
First, observe that \(\aesum{i} \not\equiv_p 0 \) implies \(\ae{\Phi}{A}
\not\equiv_p 0\) for some sub-point \(A\) with level \(i\).
Hence, for our purposes, it suffices to understand when \(\aesum{i}\) is nonzero.

Second, observe that understanding \(\aesum{i}\) directly is thus not much easier compared
to understanding a single alternating enumerator.
However, this is where \(\phisum{i}\)---which is very easy to understand and compute---turns out to be useful.
As we show next, we can express \(\aesum{i}\)-values as a linear combination of
\(\phisum{i}\)-values and---more importantly---vice versa.
In particular, this then allows us to show that some \(\aesum{i}\)-value has to be nonzero
when enough \(\phisum{j}\)-values are nonzero.

Let us start by defining the transformation matrix that we use in the following.

\begin{definition}\label{lem:10.12-2}
    For a positive integer \(n\),
    we define the \(\aename{\Phi}\)-transformation matrix \(C_n \in \Z^{(n+1) \times (n+1)}\) as
    \[
        (C_n)_{i, j} \coloneqq (-1)^j\; \binom{n - j}{i - j}, \quad\text{where \(i,j \in
        \fragment{0}{n+1}\)}.
    \]
    Further, for each \(0 \le c \le n\), we define the \(c\)-restricted
    \(\Phi\)-\(\aename{\Phi}\)-transformation matrix \(C_{n;c} \in \Z^{(n-c+1) \times (n-c+1)}\) as the
    square submatrix of \(C_n\) that consists in the first \(n-c+1\) columns and the last
    \(n-c+1\) rows of \(C_n \); that is,
    \[
        (C_{n;c})_{i, j} \coloneqq (C_n)_{i+c, j} =  (-1)^j\; \binom{n - j}{i + c - j},
    \]
    where \(i,j \in \fragment{0}{n-c}\).
\end{definition}

Recall that we set \(\smash{\tbinom{a}{b}} = 0\) whenever \(b < 0\) or \(b > a\).
In particular, this means that we have \((C_n)_{i, j} = 0\) for all \(j > i\), that is,
\(C_n\) is a lower triangular matrix.

We proceed to show that \cref{lem:10.12-2} indeed relates the \(\aesumvec\)-vector with
the \(\phisumvec\)-vector.

\begin{lemma}
    \label{lem:10.12-1}
    Let \(H\) denote a graph, let \(\Gamma \subseteq \auts{H}\) denote a \(p\)-group, and
    let \(\Phi\) denote a graph property. Define $n \coloneqq {\hl{H}}$.

    Then, we have \(\aesumvec \equiv_p  C_{} \; \phisumvec\); that is,
    for every level \(i \in \fragment{0}{n}\), we have
    \[
        \aesum{i} \equiv_p \sum_{j = 0}^{n} (-1)^j \binom{n -
        j}{i - j} \; \phisum{j}.
    \]
\end{lemma}
\begin{proof}
    Fix a level \(i \in \fragment{0}{n}\).
    First, we use \cref{cor:10.7-1} to rewrite \(\aesum{i}\). We obtain
    \begin{align*}
        \nonumber
        \aesum{i}
        &\equiv_p \sum_{ \sunderset{ \hl{A} = i }{A \in \fps{H}}} \ae{\Phi}{A}
        \equiv_p \sum_{\sunderset{ \hl{A} = i }{A \in \fps{H}}} \sum_{B \subseteq A}
        \Phi(B) (-1)^{\hl{B}}.
        \intertext{Next, we group the terms of $\sum_{B \subseteq A} \Phi(B)$ according to their
        level. We obtain}
        \nonumber
        \aesum{i}
        &\equiv_p \sum_{{\sunderset{ \hl{A} = i }{A \in \fps{H}}}}
        \sum_{j = 0}^i (-1)^j \sum_{\sunderset{ \hl{B} = j }{B \subseteq A}} \Phi(B)
        \equiv_p  \sum_{j = 0}^i (-1)^j \sum_{{\sunderset{ \hl{A} = i }{A \in \fps{H}}}}
        \sum_{\sunderset{ \hl{B} = j }{B \subseteq A}} \Phi(B).
        \intertext{Next, we reorder the summation and group single fixed points \(B\) to obtain}
        \aesum{i}
        &\equiv_p  \sum_{j = 0}^i (-1)^j
        \sum_{{\sunderset{ \hl{B} = j }{B \in \fps{H}}}}
        \sum_{\substack{ A \in \fps{H} \\ A \supseteq B \\ \hl{A} = i }} \Phi(B).
        \intertext{Finally, we compute the cardinality of the set \(\{ A \in \fps{H} \mid A
        \supseteq B \;\text{and}\; \hl{A} = i \}\) as the number of possible ways to
        choose \(i - j\) elements from \(n - j\) elements when disregarding the order
        of the chosen elements and allowing each element to be selected at most once. We obtain
        the claimed equality}
        \aesum{i}
        &\equiv_p \sum_{j = 0}^i (-1)^j \binom{n - j}{i - j}
        \sum_{\substack{ B \in \fps{H} \\ \hl{B} = j}} \Phi(B)
        \equiv_p \sum_{j = 0}^{n} (-1)^j \binom{n - j}{i - j} \phisum{j},
    \end{align*}
    where for the last step, we use the definition of $\phisum{j}$ and we exploit that
    \(\binom{a}{b} = 0\) whenever \(b < 0\).
\end{proof}

The key insight in the proof of \cref{lem:chi:basic} is the following. Assuming that the
last $c$ levels of \(\phisumvec\)-vector are 0 (that is, $\phisum{i} = 0$ for $i > n -
c$), then the last $n-c+1$ levels of the $\aesumvec$-vector are determined from
\(\phisumvec\) by the lower left $(n- c + 1)\times (n- c + 1)$ submatrix of $C_n$ (consult
\cref{fig:10.12-3}). We show next that this submatrix \(C_{n;c}\) is invertible, hence we
can revert the operation. This allows us to prove that the last $n-c+1$ levels of
$\aesumvec$ do not vanish.

\begin{figure}[t]
    \centering
    \begin{tikzpicture}
        \fill[blue!50!red!10] (0,3) rectangle (.5,0);
        \draw[thick] (0,0) rectangle (.5,5);
        \draw (0,3) -- (.5,3);

        \draw[thick] (1.5,0) rectangle (6.5,5);
        \draw[fill=blue!50!red!10] (1.5,3) rectangle (4.5,0);

        \fill[blue!50!red!10] (7.5,2) rectangle (8,5);
        \draw[thick] (7.5,0) rectangle (8,5);
        \draw (7.5,2) -- (8,2);

        \draw (1.5,5) -- (6.5,0);

        \node[scale=2, transform shape] at (5, 3.5) { 0 };

        \node[blue!50!red] at (3, 1.5) { \(C_{n;c}\)  };
        \node[blue!50!red] at (.25, 1.5) { \(\aesumvecr{c}\)  };
        \node[blue!50!red] at (7.75, 3.5) { \(\phisumvecr{c}\)  };

        \node at (4, -.5) { \(C_{n}\)  };
        \node at (1, -.5) { \(=\)  };
        \node at (7, -.5) { \(\cdot\)  };
        \node at (.25, -.5) { \(\aesumvec\)  };
        \node at (7.75, -.5) { \(\phisumvec\)  };

        \node[scale=2, transform shape] at (1, 2.5) { \(=\) };
        \node[scale=2, transform shape] at (7, 2.5) { \(\cdot\) };

    \end{tikzpicture}
    \caption{\Cref{lem:10.12-1} visualized. Assume that the last $c$ levels of the
        \(\phisumvec\)-vector are 0 (that is, $\phisum{i} = 0$ for $i > n - c$).
        Then the last $n-c+1$ levels of the $\aesumvec$-vector can be obtained by
        multiplying the first $n-c+1$ levels of the \(\phisumvec\)-vector by the lower left
        $(n-c+1)\times (n-c+1)$ submatrix $C_{n;c}$ of $C$. As \(C_{n;c}\) is invertible, this
        transformation is a bijection (\cref{lem:10.13-1}).}
    \label{fig:10.12-3}
\end{figure}

\begin{lemma}\label{lem:10.13-1}
    For a prime \(p\) and integers \(0 \le c \le n\), the matrix
    \(C_{n;c}\) is invertible in \(\field{p}^{(n-c+1)\times (n-c+1)}\).
\end{lemma}
\begin{proof}
    We compute the determinant of \(C_{n;c}\) and in particular show that
    \(\det(C_{n;c}) \not\equiv_p 0\), which suffices to prove the claim.

    To that end, first consider the matrix \(D_{n;c}\) with \[
        (D_{n;c})_{i,j} \coloneqq |(C_{n;c})_{i,j}| = \binom{n - j}{c + i - j},
        \quad\text{for \(i,j \in \fragment{0}{n-c }\)}.
    \]
    As we have \(C_{n;c} = D_{n;c} \cdot \diag(1,-1,1,-1, \dots, (-1)^{n - c + 1})\), we also
    have \[|\!\det(C_{n;c})| = |\!\det(D_{n;c}) \cdot \det(\diag(1,-1, \dots, (-1)^{n-c +
        1}))| =
    |\!\det(D_{n;c})|.\]

    Next, write $E_{n;c}$ for the matrix that is obtained from \(D_{n;c}\) by reversing its
    rows and columns, that is,
    \begin{align*}
        (E_{n;c})_{i, j} &\coloneqq (D_{n;c})_{ (n - c - i),  (n - c - j)} \\
        &= \binom{n - (n - c - j)}{c + (n - c - i) - (n - c - j)}
                         = \binom{c + j}{c + j - i} = \binom{c + j}{i},
        \quad\text{for \(i,j \in \fragment{0}{n-c }\)};
    \end{align*}
    where the last equality uses $\binom{a}{a - b} = \binom{a}{b}$.
    As permuting rows and columns of a matrix does not change (up to the sign) the
    determinant of a matrix, we have \(|\!\det(D_{n;c})| = |\!\det(E_{n;c})|\).

    Now, we use Vandermonde's Identity to obtain
    \begin{align}\label{eq:10.12-a}
        (E_{n;c})_{i, j} = \binom{c + j}{i} = \sum_{a = 0}^i \binom{c}{i- a} \binom{j}{a}.
    \end{align}
    \Cref{eq:10.12-a} in turn shows that we have $(E_{n;c}) = (L_{n;c})(U_{n;c}) $ for the matrices
    \[(L_{n;c})_{i, j} \coloneqq \binom{c}{i-j} \quad\text{and}\quad (U_{n;c})_{i,j} \coloneqq \binom{j}{i},
    \quad\text{for \(i,j \in \fragment{0}{n-c }\)}.\]
    Hence, we have $\det(E_{n;c}) = \det(L_{n;c}) \det(U_{n;c})$.

    Finally, observe that $L_{n;c}$ is a lower-triangular matrix (as
    \(\smash{\tbinom{c}{i-j}} =
    0\) whenever \(i<j\)) and observe that
    $U_{n;c}$ is an upper-triangular matrix (as \(\smash{\tbinom{j}i} = 0\) whenever \(i >
    j\)). Hence, we obtain
    \[\det(L_{n;c}) = \prod_{a = 0}^{n - c} (L_{n;c})_{a, a} = \prod_{a = 0}^{n - c}
    \binom{c}{0} = 1 \quad\text{and}\quad
    \det(U_{n;c}) = \prod_{a = 0}^{n - c} (U_{n;c})_{a, a}
    = \prod_{a = 0}^{n - c} \binom{a}{a} = 1;\]
    which yields $\det(E_{n;c}) = 1$ and thus $\det(C_{n;c}) \in \{-1, 1\}$.
    Hence, we have $\det(C_{n;c}) \not\equiv_p 0$, which implies that $C_{n;c}$ is regular in
    $\field{p}^{(n-c+1)\times(n-c+1)}$; completing the proof.
\end{proof}

We combine \cref{lem:10.12-1,lem:10.13-1} to obtain the main result of the section.

\lemchibasic
\begin{proof}
Let us define $n=\hl{H}$.    Write \(\phisumvecr{c}\) for the vector that consists
    in the \(n-c+1\)  first entries of the vector \(\phisumvec\) and
    write \(\aesumvecr{c}\) for the vector that consists
    in the \(n-c+1\) last entries of the vector \(\aesumvec\). First, we show
    \(\aesumvecr{c} \equiv_p C_{n;c} \; \phisumvecr{c}\).
    Observe that we have \(\phisum{i} = 0\) for all \( i > n - c\).
    Hence, \cref{lem:10.12-1} yields \[
        \aesum{i}
        \equiv_p \sum_{j = 0}^{n} (-1)^j \binom{n - j}{i - j} \; \phisum{j}
        \equiv_p \sum_{j = 0}^{n - c} (-1)^j \binom{n - j}{i - j} \; \phisum{j};
    \]
    This can be rewritten into \(\aesumvecr{c} \equiv_p C_{n;c} \; \phisumvecr{c}\) by shifting the indices.

    From \(\Phi(\emptyset) = 1\), we conclude \(\phisum{0} = 1 \not\equiv_p 0\); hence,
    \(\phisumvecr{c}\) is not the zero vector.
    Now, by \cref{lem:10.13-1}, the matrix \(C_{n,c}\) is regular; hence we obtain that \(\aesumvecr{c} \equiv_p C_{n;c} \;
    \phisumvecr{c}\) cannot be the zero vector, either.
    As \( \aesumvecr{c}\) contains the  \(n - c + 1\) last entries of \(\aesumvec\),
    there is thus a fixed point \(S \in \fps{H}\) with \(\ae{\Phi}{S} \not\equiv_p 0\) and
    \(n \ge c\); completing the proof.
\end{proof}

\section{Prime Powers and Difference Graphs}\label{sec:difference:graphs}

In this section, we obtain a first application of the techniques from
\cref{sec:alt:en:hasse}. In particular, our goal for this section is to prove the
following theorem.

\thmedgemonprimecase*
\medskip

For our proof of \cref{theo:edge_mono:prime}, we consider the complete graph \(K_{p^m}\)
for a prime \(p\) and a nonnegative integer \(m\).
In particular, we identify the vertices of \(K_{p^m}\) with elements of the finite field
\(\field{p^m}\) and consider the rotation subgroup \(\rotgr{p^m}\subseteq
\aut(K_{p^m})\)---recall that the rotation group \(\rotgr{p^m}\) consists in all
automorphisms of \(K_{p^m}\) that map every vertex \(v\) to the vertex \(v + b\) (for some
\(b \in \field{p^m}\), where addition is in \(\field{p^m}\)).

In a first step, we use \cref{remark:fixed:point:union} to understand the fixed points \(\fpb{\rotgr{p^m}}{K_{p^m}}\) by
computing the orbits \(E(K_{p^m})/\rotgr{p^m}\).
In a second step, we then use \cref{lem:chi:basic} to obtain useful fixed points that have
both nonvanishing alternating enumerator and high level. Finally, we show how to construct a
suitable sequence of useful fixed points to obtain the desired hardness via
\cref{lem:alpha:treewidth}.

\subsection{The Fixed Points of Rotations of \texorpdfstring{$K_{p^m}$}{K p m}}
\label{chap:fixed:point}

As described, we wish to understand the orbits \(E(K_{p^m})/\rotgr{p^m}\).
To that end, it is instructive to recall the definition of difference graphs and of the level
of a difference graph.

\defcircgraph*

For \(m = 1\), the graphs  \(\cgr{p}{A}\) are also called \emph{circulant graphs} in the
literature \cite{Automorphism_Groups_Ciculant_semigroup_theory, circulant_graph2}.

Next, let us quickly confirm the intuitive statement that different sets
\(A, B \subseteq \field{p^m}^+\) give rise to different difference graphs \(\cgr{p^m}{A}\) and
\(\cgr{p^m}{B}\).

\begin{lemma}\label{lem:10.15-1}
    For any sets \(A \neq B \subseteq \field{p^m}^+\), we have
    \(\cgr{p^m}{A}\neq\cgr{p^m}{B}\).
\end{lemma}
\begin{proof}
    Assume without loss of generality that \(B \setminus A \neq \emptyset\); otherwise
    swap \(A\) and \(B\).
    Now, consider an element \(b  \in B \setminus A\) and an \(x \in \field{p^m}\). Clearly, we
    have \(\{x, x + b\} \in E(\cgr{p^m}{B})\) and \(\{x, x + b\} \not\in
    E(\cgr{p^m}{A})\).
\end{proof}

Now, we obtain that the orbits \(E(K_{p^m})/\rotgr{p^m}\) are the difference graphs
of the singleton subsets of~\(\field{p^m}^+\).

\begin{lemma}\label{lem:orbits:Zpa}
    Let $p$ denote a prime and let $m > 0$ denote an integer.
    Further, let \(\rotgr{p^m}\) act on \(E(K_{p^m})\).

    Then, we have \(E(K_{p^m})/\rotgr{p^m} = \{ \cgr{p^m}{\{x\}}  \mid x \in \field{p^m}^+
    \}\).\footnote{Again, we abuse notation and identify subsets of edges with the
    corresponding edge-subgraphs.}
\end{lemma}
\begin{proof}
    First, we show that each edge set $E(\cgr{p^m}{\{x\}})$ defines a unique orbit.

    \begin{claim}\label{cl:10.14-1}
        For every \(x \in \field{p^m}^+\),
        we have \( \rotgr{p^m} \cdot\; \{0, x\} = E(\cgr{p^m}{\{x\}}) \).
    \end{claim}
    \begin{claimproof}
        For a \(b \in \field{p^m}\), write \(\varphi_b \in \rotgr{p^m}\) for the
        rotation \(x \mapsto (x + b)\).
        Fix an \(x \in \field{p^m}^+\).
        We compute the orbit of the edge $\{0, x\} \in E(K_{p^m})$ under
        \(\rotgr{p^m}\) as
        \begin{align*}
            \rotgr{p^m} \cdot\; \{0, x\}
            &= \{ \{\varphi_b(0), \varphi_b(x) \}  \mid \varphi_b \in \rotgr{p^m}\} \\
            &= \{\{b, x + b\}  \mid b \in \field{p^m} \} \\
            &= \{\{u, v\} : u, v \in \field{p^m}, u - v = x \}\\
            &= E(\cgr{p^m}{\{x\}}).
            \qedhere
        \end{align*}
    \end{claimproof}
    Now, \cref{cl:10.14-1,lem:10.15-1} indeed yield that the edge sets
    $E(\cgr{p^m}{\{x\}})$ form disjoint orbits.

    Lastly, we check that $\{\cgr{p^m}{\{x\}}  \mid x \in \field{p^m}^+\}$ are all orbits.

    \begin{claim}\label{cl:10.15-2}
        We have \(E(K_{p^m}) = \bigcup_{x \in \field{p^m}^+}  \rotgr{p^m} \cdot\; \{0, x\}\).
    \end{claim}
    \begin{claimproof}
        Consider an arbitrary edge \(\{u, v\} \in E(K_{p^m})\).
        Via the rotation \(\varphi_u\), we have
        \(\{u, v\} \in\; \rotgr{p^m} \cdot\; \{0, v - u\}\);
        via the rotation \(\varphi_v\), we have
        \(\{u, v\} \in\; \rotgr{p^m} \cdot\; \{0, u - v\}.\)
        Finally, we observe that \(\field{p^m}^+\) contains exactly one out of \(u-v\) and
        \(v-u\); completing the proof.
    \end{claimproof}

    Taken together, \cref{cl:10.14-1,cl:10.15-2} yield the desired
    \(E(K_{p^m})/\rotgr{p^m} = \{ \cgr{p^m}{\{x\}}  \mid x \in \field{p^m}^+ \} \).
\end{proof}

Combining \cref{remark:fixed:point:union,lem:orbits:Zpa}, we obtain a classification of
the fixed points \(\fpb{\rotgr{p^m}}{K_{p^m}}\).

\begin{lemma}\label{cor:fix:Hs}
    Let $p$ denote a prime and let $m > 0$ denote an integer.
    Further, let \(\rotgr{p^m}\) act on \(\edgesub(K_{p^m})\).

    Then, we have
    \[\fpb{\rotgr{p^m}}{K_{p^m}} = \{ \cgr{p^m}{A}  \mid A \subseteq \field{p^m}^+\}.\]
\end{lemma}
\begin{proof}
    From \cref{lem:orbits:Zpa}, we obtain
    \(E(K_{p^m})/\rotgr{p^m} = \{ \cgr{p^m}{\{x\}}  \mid x \in \field{p^m}^+ \}\).

    From \cref{remark:fixed:point:union}, we obtain that each fixed point in
    \(\fpb{\rotgr{p^m}}{K_{p^m}}\) is the (disjoint) union of orbits from
    \(E(K_{p^m})/\rotgr{p^m}\).

    Finally, we readily convince ourselves that for sets \(A, B \subseteq \field{p^m}^+\)
    we have \(\cgr{p^m}{A \cup B} = \cgr{p^m}{A} \cup \cgr{p^m}{B}\).
\end{proof}

\begin{remark}
    Recall that the level of $\cgr{p^m}A$ is the cardinality of $A$. From \cref{cor:fix:Hs} we
    see that, indeed, \(\cgr{p^m}{A}\) (as a fixed point) consists in \(|A|\) orbits and thus
    also has a level of \(|A|\) as a fixed point.

    Also consult \cref{fig:hasse:11} for a visualization of an example.
\end{remark}

Finally, we use that $\cgr{p^m}{A}$ is regular to show a lower bound for the treewidth.
\begin{corollary}\label{lm:10.25-3}
    Let $p$ denote a prime and let $m > 0$ denote an integer.

    Then, every fixed point $\cgr{p^m}{A} \in \fpb{\rotgr{p^m}}{K_{p^m}}$ has
    treewidth of at least \(\hl{A}\).
\end{corollary}
\begin{proof}
    First, for all $x \in V(\cgr{p^m}{A}) = \field{p^m}$, we obtain that $x$ is adjacent
    to all vertices of the form $x + c$ for $c \in A \cup (-A)$.

    If $p \neq 2$, then $A$ and $-A$ are disjoint by construction of $\field{p^m}^+$,
    hence the graph is $2\ell(A)$-regular. Otherwise, $p = 2$ and $A = -A$, thus the graph
    is $\ell(A)$-regular. According to \cite[Lemma 4]{treewidth} a $d$-regular graph has
    treewidth at least $d$, which proves the claim.
\end{proof}

\begin{figure}[p]
    \centering
    \includegraphics[width=\textwidth]{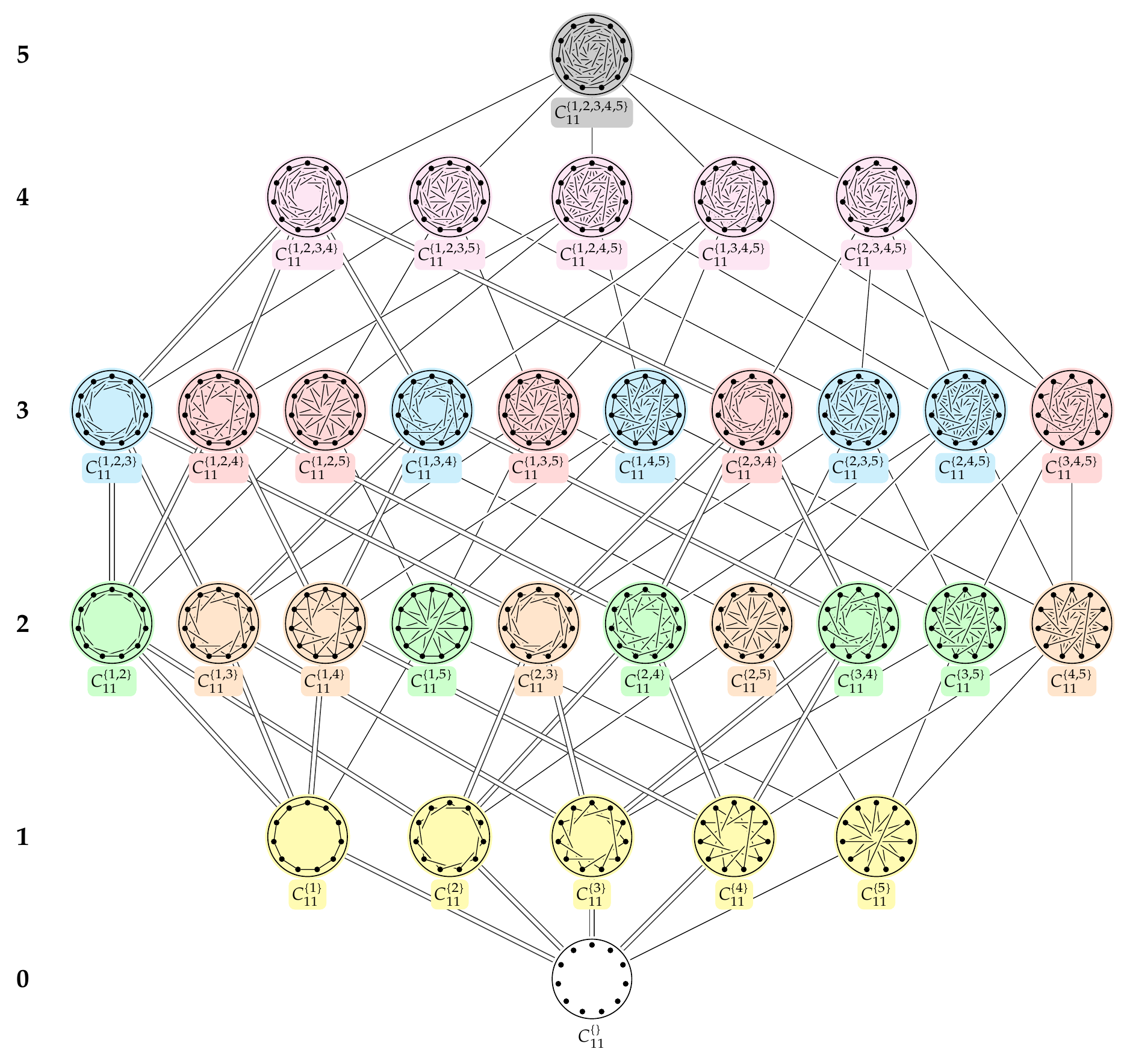}
    \caption{The fixed points \(\fpb{\rotgr{11}}{K_{11}}\) form the same lattice as
        subsets of a 5-element universe. We group the fixed points according to their
        level (which is denoted on the left side).
        Isomorphic graphs are colored with the same color.
        Two fixed points of adjacent levels are connected with an edge if one of them is a
        sub-point of the other. Double edges highlight the sub-points of
        \(\cgr{11}{\{1,2,3,4\}}\). Observe that if this fixed point satisfies an edge-monotone
        property $\Phi$, then every fixed point on the third level (every fixed point in
        red or blue) satisfies $\Phi$ as well.}
     \label{fig:hasse:11}
\end{figure}

\subsection{Nonvanishing Alternating Enumerators via Avalanches}
\label{sec:proof:prime:number}

Next, we wish to exploit our understanding of the fixed points \( \fpb{\rotgr{p^m}}{K_{p^m}}\).
In particular, we intend to show that for an edge-monotone property \(\Phi\), a single high-level
fixed point \(A \in \fpb{\rotgr{p^m}}{K_{p^m}}\) that satisfies \(\Phi\) already implies
that all fixed points of a small level also satisfy \(\Phi\).
We can exploit this ``avalanche'' effect to obtain another criterion for \w-hardness of \(\NUM{}\indsubsprob(\Phi)\).
To that end, we first need to understand when fixed points of
\(\fpb{\rotgr{p^m}}{K_{p^m}}\) are isomorphic.

\begin{definition}
    Let $p$ denote a prime and let $m > 0$ denote an integer.
    Two sets $A, B \subseteq \field{p^m}^+$ are \emph{isomorphic}, denoted by \(A
    \scalesim
    B\), if there is a \(\lambda \in \field{p^m}^\ast\) such that
    \(\lambda \cdot (A \cup (-A)) = B \cup (-B)\).
\end{definition}
\begin{lemma}\label{lem:circ_iso}
    Let $p$ denote a prime and let $m > 0$ denote an integer.
    For any two sets $A, B \subseteq \field{p^m}^+$ with \(A \scalesim B\), we have
    \(\cgr{p^m}{A} \cong \cgr{p^m}{B}\).
    \end{lemma}
\begin{proof}
    By definition, there is a \(\lambda \in \field{p^m}^+\) such that
    \(\lambda \cdot (A \cup (-A)) = B \cup (-B)\).

    Now, consider the function
    $\varphi_{\lambda} \colon V(K_{p^m}) \to V(K_{p^m})$ with $x \mapsto  \lambda x$.
    As \(\lambda\) is invertible,
    \(\varphi_{\lambda}\) is an automorphism of \(K_{p^m}\) and thus bijective.

    Next, for any \(u, v  \in \field{p^m}\) with \(u-v \in A \cup (-A)\),
    we have $\lambda u - \lambda v = \lambda(u - v) \in \lambda (A \cup (-A)) =  B \cup
    (-B)$.
    Hence, $\varphi_{\lambda}$ maps each edge $\{u, v\}$ of $\cgr{p^m}{A}$ to the
    edge $\{\lambda u, \lambda v\}$ of $\cgr{p^m}{B}$.
    Thus, \(\varphi_{\lambda}\) is also a homomorphism from \(\cgr{p^m}{A}\) to
    \(\cgr{p^m}{B}\);
    which is also surjective as \(|\lambda \cdot (A \cup (-A))| = |B \cup (-B)|\).

    In total, \(\varphi_{\lambda}\) is a homomorphism that is both surjective and
    injective; thus \(\varphi_{\lambda}\) is an isomorphism.
\end{proof}

Next, we show the aforementioned ``avalanche'' effect.
Our key statement shows that every sufficiently small $B$ set is isomorphic to a subset of $A$.
\begin{lemma}\label{lm:10.22-3}
    Let $p$ denote a prime, let $m > 0$ denote an integer, and
    let \(A \subsetneq \field{p^m}^+\) denote a set.

    For every \(B \subseteq \field{p^m}^+\) with
    \(|B| < |\field{p^m}^+|/ (|\field{p^m}^+| - |A|)\),
    there is an isomorphic set \(B' \scalesim B\) with \(B' \subseteq A\).
\end{lemma}
\begin{proof}
    Let \(B \subseteq \field{p^m}^+\) denote a set with
    \(|B| < |\field{p^m}^+| / (|\field{p^m}^+| - |A|)\).
    We order \(B\) arbitrarily and define a vector \(\vec{b}  \in
    (\field{p^m}^{\ast})^{|B|}\) via
    \[
    \vec{b} \coloneqq (B\position{0}, \dots, B\position{|B|-1}).
    \]
    It is easy but instructive to confirm that multiplying \(\vec{b}\) with different
    nonzero elements yield vectors that differ at every position.
    \begin{claim}\label{cl:10.22-1}
        Let \(\lambda, \mu \in \field{p^m}^{\ast}\) and
        let \(i \in \fragmentco{0}{|B|}\) denote a position.

        If we have \((\lambda \vec{b})_i = (\mu \vec{b})_i\), then \(\lambda = \mu\).
    \end{claim}
    \begin{claimproof}
        By construction, each element \(\vec{b}_i \in \field{p^m}^{\ast}\) is invertible;
        which yields the claim.
    \end{claimproof}

    Next, write \((\star)^+\) for the function that maps each element from
    \(\field{p^m}^{\ast}\) to its corresponding representative in~\(\field{p^m}^+\).
    Now, write
    \(\field{p^m}^{\ast} \vec{b} \coloneqq \{ \lambda \vec{b} \mid \lambda \in \field{p^m}^{\ast} \}\)
    for the set of multiples of \(\vec{b}\)
    and write
    \((\field{p^m}^{\ast} \vec{b})^+ \coloneqq \{ (\lambda \vec{b})^+ \mid \lambda \in \field{p^m}^{\ast} \}\)
    for the set of all multiples of \(\vec{b}\) when ``ignoring the signs'' of the entries
    of the multiples.
    In particular, each vector in \((\field{p^m}^{\ast} \vec{b})^+\) corresponds to some
    \(B' \subseteq \field{p^m}^+\) with \(B \scalesim B'\) (which need not be different
    for each vector).

    \begin{claim}\label{cl:10.22-2}
        We have
        \(|\field{p^m}^{\ast} \vec{b}| = |\field{p^m}^{\ast}|\) and \(|(\field{p^m}^{\ast}
        \vec{b})^+| \ge |\field{p^m}^{+}|\).
    \end{claim}
    \begin{claimproof}
        The first equality is immediate from \cref{cl:10.22-1} (applied to the first
        elements of the vectors in \(\field{p^m}^{\ast} \vec{b}\)).

        For the second equality, first consider the case \(p = 2\). Now, we have
        \(\field{2^m}^{\ast} = \field{2^m}^+\) and \((\star)^+\) is the identity; which
        yields the claim.

        Now, for \(p \neq 2\),  we have \(|\field{p^m}^{\ast}| = 2|\field{p^m}^+|\).
        Further, \((\star)^+\) identifies at most 2 different elements from
        \(\field{p^m}^{\ast}\) to the same element from \(\field{p^m}^{+}\); which in
        particular holds for the first elements of the vectors in \(\field{p^m}^{\ast}
        \vec{b}\).
        Taken together, we obtain the claim.
    \end{claimproof}

    Next, we say that an element \(s \in \field{p^m}^+ \setminus A \) \emph{sullies} a vector \(\lambda
    \vec{b} \in (\field{p^m}^{\ast} \vec{b})^+\) if
    there is a position \(i \in \fragmentco{0}{|B|}\) with
    \((\lambda \vec{b})_i = s\).
    Similarly, we say that an element \(s \in \field{p^m}^+ \setminus A\) \emph{sullies} a position
    \(i \in \fragmentco{0}{|B|}\) if there is a
    vector \(\lambda \vec{b} \in (\field{p^m}^{\ast} \vec{b})^+\) with \((\lambda \vec{b})_i = s\).

    We observe that we can prove the lemma by showing that there is a vector in \(
    (\field{p^m}^{\ast} \vec{b})^+\) that is not sullied by any  \(s \in \field{p^m}^+
    \setminus A \).

    To that end, we observe that each element \(s \in \field{p^m}^+\) sullies at most one position (by
    construction, all entries of \(\vec{b}\) are pairwise different) and at most one
    vector (by \cref{cl:10.22-1}).
    In particular, this means that each element \(s \in \field{p^m}^+\) sullies at most
    \(|B|\) vectors of \((\field{p^m}^{\ast} \vec{b})^+\).
    Hence, in total, at most  \(|B| \cdot |\field{p^m}^+ \setminus A|\) vectors of \(
    (\field{p^m}^{\ast} \vec{b})^+\) are sullied.
    Finally, we plug in
    \(|B| < |\field{p^m}^+| / (|\field{p^m}^+| - |A|)\) to obtain
    \[
        |B| \cdot |\field{p^m}^+ \setminus A| < |\field{p^m}^+| \le |(\field{p^m}^{\ast}
        \vec{b})^+|;
    \] from which we conclude that, indeed, there is a vector in \((\field{p^m}^{\ast}
    \vec{b})^+\) that is not sullied. This in turn completes the proof.
\end{proof}

We combine \cref{lem:circ_iso,lm:10.22-3}, to readily obtain that any (nontrivial) edge-monotone
graph property has to be true for all difference graphs with a small level.

\begin{corollary}\label{high_difference}
    Let $p$ denote a prime, let $m > 0$ denote an integer, and
    let \(A \subsetneq \field{p^m}^+\) denote a set.
    Further, let \(\Phi\) denote an edge-monotone graph property with \(\Phi(\cgr{p^m}{A})
    = 1\).

    For every \(B \subseteq \field{p^m}^+\) with
    \(|B| < |\field{p^m}^+|/ (|\field{p^m}^+| - |A|)\),
    we have \(\Phi(\cgr{p^m}{B}) = 1\).
\end{corollary}
\begin{proof}
    Let \(B \subseteq \field{p^m}^+\) denote a set with
    \(|B| < |\field{p^m}^+| / (|\field{p^m}^+| - |A|)\).

    From \cref{lm:10.22-3}, we obtain an isomorphic subset \(B' \scalesim B\) with \(B'
    \subseteq A\). From \cref{lem:circ_iso}, we obtain \(\cgr{p^m}{B'} \cong
    \cgr{p^m}{B}\). Combined, we thus obtain that \(\cgr{p^m}{B}\) is isomorphic to an edge-subgraph of
    \(\cgr{p^m}{A}\). Finally, as \(\Phi\) is edge-monotone, we obtain
    \(\Phi(\cgr{p^m}{B}) = 1\); which completes the proof.
\end{proof}

Next, we combine \cref{high_difference,remark:chi:all:children:true} to obtain another
criterion for \w-hardness of \(\NUM{}\indsubsprob(\Phi)\).

\lemlevelhighprime
\begin{proof}
    First, we recall that we have
    \(\hl{\cgr{p^m}{A}} = \hl{A} = |A|\).
    Further, we have
    \[
        c = \frac{cd}{d} \le \frac{ |\field{p^m}^+| }{|\field{p^m}^+| - (
        |\field{p^m}^+| - d )} \le
    \frac{ |\field{p^m}^+| }{ |\field{p^m}^+|  - |A|}.
\]

    Now, from \cref{high_difference}, we obtain that $\Phi(\cgr{p^m}{S}) = 1$ for all $S$ with
    \(|S| < c \le |\field{p^m}^+| / (|\field{p^m}^+| - |A|)\).

    As $\Phi$ is nontrivial on $p^m$, we have $\Phi(\cgr{p^m}{\field{p^m}^+})
    = \Phi(K_{p^m}) = 0$.
    Thus, there is a minimal level \(b \ge c\) on which not every fixed point satisfies \(\Phi\).
    Consider such a fixed point $\cgr{p^m}{B}$ with level \(\hl{B} = b\) that does not
    satisfy \(\Phi\).
    As \(b\) is minimal, all proper sub-points of \(\cgr{p^m}{B}\) do satisfy \(\Phi\).
    Hence, we may use \cref{remark:chi:all:children:true} to obtain
    $\ae{\Phi}{\cgr{p^m}{B}} \not \equiv_{p} 0$; thus completing the proof.
\end{proof}

On the one hand, we may use \cref{lem:level:high:prime_case} whenever we have a fixed point that
satisfies $\Phi$ and has a high level.
On the other hand, we may use \cref{lem:chi:basic} if there is no fixed point with a high
level that satisfies $\Phi$.
Hence, if we combine \cref{lem:level:high:prime_case,lem:chi:basic}, then we can show that
there is a fixed point with a nonvanishing alternating enumerator and a treewidth of
roughly ${p^{m/2}}$ if $\Phi$ is nontrivial on $p^m$.

\begin{corollary}\label{corollary:fp:sqrt:tw}
    Let $p$ denote a prime, let $m > 0$ denote an integer,
    and let \(\Phi\) denote an edge-monotone graph property that is nontrivial on \(p^m\).

    Then, there is a fixed point $\cgr{p^m}{A} \in \fpb{\rotgr{p^m}}{K_{p^m}}$ that has a
    nonvanishing alternating enumerator and a treewidth of at least $p^{m/2}/2 - 2$.
\end{corollary}
\begin{proof}
    We show that there is a fixed point $\cgr{p^m}{A}$ with $\ae{\Phi}{\cgr{p^m}{A}} \not
    \equiv_{p} 0$ and $\hl{A} \geq p^{m/2}/2 - 2$.
    Then, the claim follows from \cref{lm:10.25-3}.

    First, suppose that a fixed point $\cgr{p^m}{S}$ with $\hl{S} \geq |\field{p^m}^+| - d$ satisfies \(\Phi\).
    As the property \(\Phi\) is nontrivial on \(p^m\), this means that we can use
    \cref{lem:level:high:prime_case} to obtain a fixed point
    $\cgr{p^m}{A}$ with $\ae{\Phi}{\cgr{p^m}{A}} \not \equiv_{p} 0$ and $\hl{A} \geq c$
    for any \(c\) with \(cd \le |\field{p^m}^+|\).

    Next, suppose that no fixed point $\cgr{p^m}{S}$ with $\hl{S} \geq n - d$ satisfies \(\Phi\).
    As the property \(\Phi\) is nontrivial on \(p^m\), this means that we can use \cref{lem:chi:basic}
    to obtain a fixed point $\cgr{p^m}{A}$ with
    $\ae{\Phi}{\cgr{p^m}{A}} \not \equiv_p 0$ and $\Hasselevel(A) \geq d$.

    In both cases we obtain a nonvanishing fixed point $\cgr{p^m}{A}$ with level at
    least $c$.
    Choosing \(c \coloneqq d\),
    in both cases, we obtain a nonvanishing fixed point $\cgr{p^m}{A}$ with level at
    least $d$.

    Finally, to ensure that \(cd = d^2\) is at most \(|\field{p^m}^+|\),
    set \(d \coloneqq \lfloor{{|\field{p^m}^+|}^{1/2}}\rfloor\).
    Recalling \cref{ft:10.25-1}, we observe
    \begin{align*}\label{eq:10.25-2}
        d &= \lfloor{{|\field{p^m}^+|}^{1/2}}\rfloor
        \ge |\field{p^m}^+|^{1/2} - 1
        \ge (p^m/2 - 1)^{1/2} - 1
        \ge p^{m/2}/2 - 2.
    \end{align*}

    Now, the claim follows from \cref{lm:10.25-3}.
\end{proof}

\subsection{\NUM{}W[1]-hardness and Quantitative Lower Bounds for Prime Powers}

\Cref{corollary:fp:sqrt:tw} directly implies our first result.

\thmedgemonprimecase
\begin{proof}
    We start with \w-hardness.
    By assumption on \(\Phi\), for every $k \in \nat$, there is a prime number~$p_k$ and
    a positive integer $m_k$ such that \(\Phi\) is nontrivial on $p_k^{m_k} \geq k$.

    From \cref{corollary:fp:sqrt:tw}, we obtain a graph $H_k \coloneqq \cgr{p_k^{m_k}}{A}$
    with \(p_k^{m_k}\) vertices,
    nonvanishing alternating enumerator, and $\tw(H_k) \geq p_k^{m_k/2} / 2 - 2 \ge \sqrt{k}/2 - 2$.
    Now, using the constructed sequence of graphs,
    \cref{lem:alpha:treewidth} yields \w-hardness of $\NUM{}\indsubsprob(\Phi)$.

    We proceed to the ETH-based lower bound.
    To that end,
    write $\alpha_{\indsubsprob} > 0$ for the constant from \cref{lem:alpha:treewidth}
    and set $\alpha' \coloneqq \alpha_{\indsubsprob}/3$.

    Now, fix a prime power \(k \ge 12^2\) such that \(\Phi\) is nontrivial on \(k\).
    From \cref{corollary:fp:sqrt:tw}, we obtain a graph $H_k \coloneqq \cgr{k}{A}$ with
    \(k\) vertices,
    nonvanishing alternating enumerator, and
    \[
        \tw(H_k) \geq \sqrt{k}/2 - 2 \geq \sqrt{12^2} / 2 - 2 \geq 2.
    \]
    Thus, by \cref{lem:alpha:treewidth} and assuming ETH,
    there is no algorithm that for each graph \(G\) computes the number $\NUM{}\indsubs{(\Phi,
    k)}{G}$ in
    time $O(|V(G)|^{\alpha_{\indsubsprob} \tw(H_k) / \log\tw(H_k)})$.
    We complete the proof by showing the following inequality.

    \begin{claim}\label{10.27-2}
        For \(12^2 \le k\), we have \[
        \alpha' \frac{\sqrt{k}}{\log k} = \alpha_{\indsubsprob} \frac{\sqrt{k}/3}{\log k } \le \alpha_{\indsubsprob} \frac{\tw(H_k)}{\log\tw(H_k)}.
        \]
    \end{claim}
    \begin{claimproof}
        For \(k \ge 12^2\), we have $\sqrt{k}/3 \leq \sqrt{k}/2 - 2$. Further, we define
        $h_1 \colon \mathbb{R}_{> 1} \to \mathbb{R}, x \mapsto x / \log(x)$. The
        derivative of this function is $(\log(x) - 1)/(\log^2(x))$, thus the function is
        monotonically increasing for $x \geq \mathrm{e}$. Since $\mathrm{e} < \sqrt{k}/3
        \leq \sqrt{k}/2 - 2 \leq \tw(H_k) $, we obtain
        \[
            \alpha_{\indsubsprob} \frac{\sqrt{k}/3}{\log(k)}
            \le
            \alpha_{\indsubsprob} \frac{\sqrt{k}/3}{\log(\sqrt{k}/3)}
            \le
            \alpha_{\indsubsprob} \frac{\sqrt{k}/2 - 2}{\log( \sqrt{k}/2 - 2)}
            \le
            \alpha_{\indsubsprob} \frac{\tw(H_k)}{\log \tw(H_k)}
        ;\]
        which completes the proof.
    \end{claimproof}

    Now, from \cref{10.27-2}, we obtain
    \[
       O(|V(G)|^{\alpha' \sqrt{k} / \log k}) \subseteq O(|V(G)|^{\alpha_{\indsubsprob} \tw(H_k) / \log \tw(H_k)}). \]
    Finally,
    to obtain the claim also for all \(3 \le k < 12^2\), we chose \(\alpha \coloneqq
    \min( \alpha', 1 / 12)\).
    Observe that for \(k < 12^2\), we obtain
    \[
         O(|V(G)|^{\alpha \sqrt{k} / \log k}) = o( |V(G)| ).
     \]
    Now, such a running time is unconditionally unachievable
    for any algorithm that reads
    the whole input.
    This completes the proof.
\end{proof}

\section{Main Result 1: \NUM{}W[1]-hardness for Edge-monotone Properties}
\label{sec:mainresult1}
In this section, we prove \cref{theo:edge_mono}.

\thmedgemon*
\medskip

We prove \cref{theo:edge_mono} by raising the techniques of \cref{sec:difference:graphs}
from prime powers $p^m$ to multiples $d\cdot p^m$ of prime powers.
In \cref{sec:product,sec:fixedpoint:kdpm}, we introduce group-theoretic and
graph-theoretic concepts that allow us to do this generalization.
While the notions and the arguments are fairly natural, the statements and the proofs come
with some notational overhead.
One important point to keep in mind is that the resulting fixed points on $d\cdot p^m$
vertices are not just the disjoint unions of $d$ fixed points on $p^m$ vertices, but may
contain complete bipartite graphs between some of the copies.

In \cref{sec:scattered,sec:main1prof}, we exploit our understanding of these fixed points
to prove \cref{theo:edge_mono}.
Compared to the hardness proofs in \cref{sec:difference:graphs}, there is an extra
Inclusion-Exclusion step to obtain a reduction from $p^m$ vertices to $d\cdot p^m$
vertices.

\subsection{Product Groups, Graphs Unions, and Graph Joins}\label{sec:product}
We start by defining the disjoint union of sets and product groups.

\begin{definition}\label{def:10.31-1}
    For (not necessarily disjoint) sets $X_1, \dots, X_m$, the disjoint union $X_1 \unionSet \cdots \unionSet X_m$
    is the set $\{(i, x) \mid i \in \setn{m}, x \in X_i\}$.

    For graphs $G_1, \dots, G_m$ and a graph $C \in \graphs{m}$,
    we define the \emph{inhabited graph} $\metaGraph{C}{G_1, \dots, G_m}$ via
    \begin{align*}
        V(\metaGraph{C}{G_1, \dots, G_m}) &\coloneqq V(G_1) \unionSet \dots \unionSet V(G_m) \\
        E(\metaGraph{C}{G_1, \dots, G_m}) &\coloneqq \{\{(i, v_i), (j, u_j)\} \mid \{i,
        j\} \in E(C) \text{ or } (i = j \text{ and } \{v_i, u_i\} \in E(G_i)) \}.
    \end{align*}

    If \(C\) consists in a single edge \(\{i,j\}\), then we also write
    \[
        \metaGraph{\oneEdge{i}{j} }{G_1, \dots, G_m} \coloneqq
    \metaGraph{(\setn{m}, \{i,j\})}{G_1, \dots, G_m}.
    \]

    If \(C = IS_m\), then we call the graph \(\metaGraph{IS_m}{G_1, \dots, G_m}\) the
    \emph{disjoint union} of \(G_1, \dots, G_m\) and also write
    \[
    G_1 \UnionGraph \dots \UnionGraph G_m \coloneqq
    \metaGraph{IS_m}{G_1, \dots, G_m}.
    \]

    If \(C = K_m\), then we call the graph \(\metaGraph{K_m}{G_1, \dots, G_m}\) the
    \emph{join} of \(G_1, \dots, G_m\) and also write
    \[
    G_1 \JoinGraph \dots \JoinGraph G_m \coloneqq
    \metaGraph{K_m}{G_1, \dots, G_m}.
    \tag*{\qedhere}
    \]
\end{definition}

\begin{figure}[pt]
    \centering
    \begin{subfigure}{\textwidth}
    \centering
        \includegraphics[width=\textwidth]{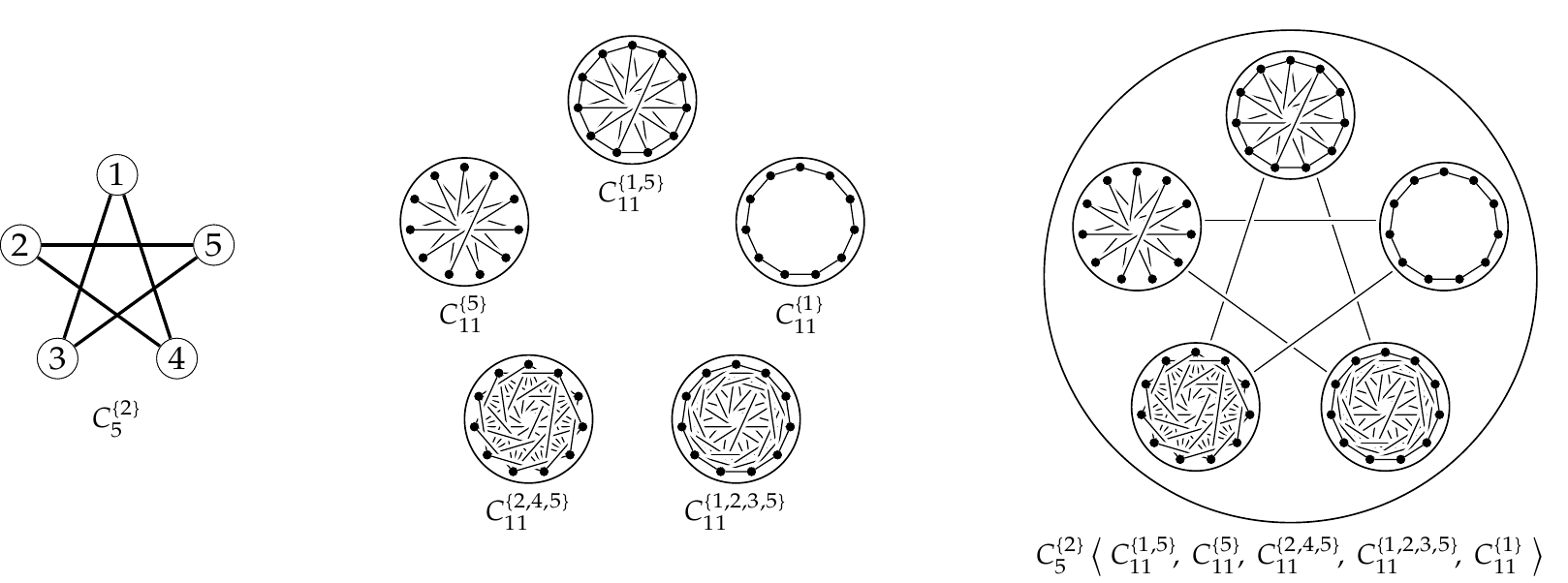}
        \caption{A 5-cycle that is inhabited by 5 difference graphs.}
    \end{subfigure}
    \begin{subfigure}{\textwidth}
    \centering
        \includegraphics[width=\textwidth]{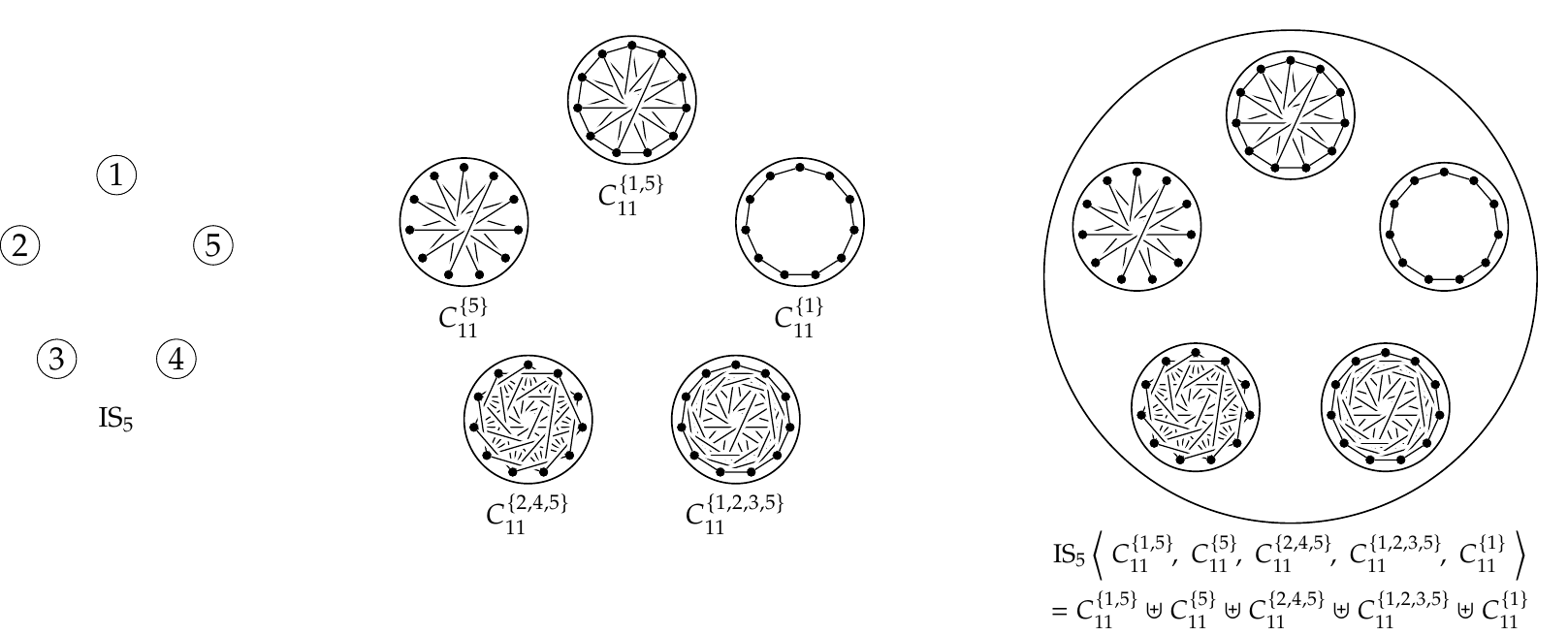}
        \caption{An independent set that is inhabited by 5 difference graphs.
        The resulting graph is also the disjoint union of the 5 graphs.}
    \end{subfigure}
    \begin{subfigure}{\textwidth}
    \centering
        \includegraphics[width=\textwidth]{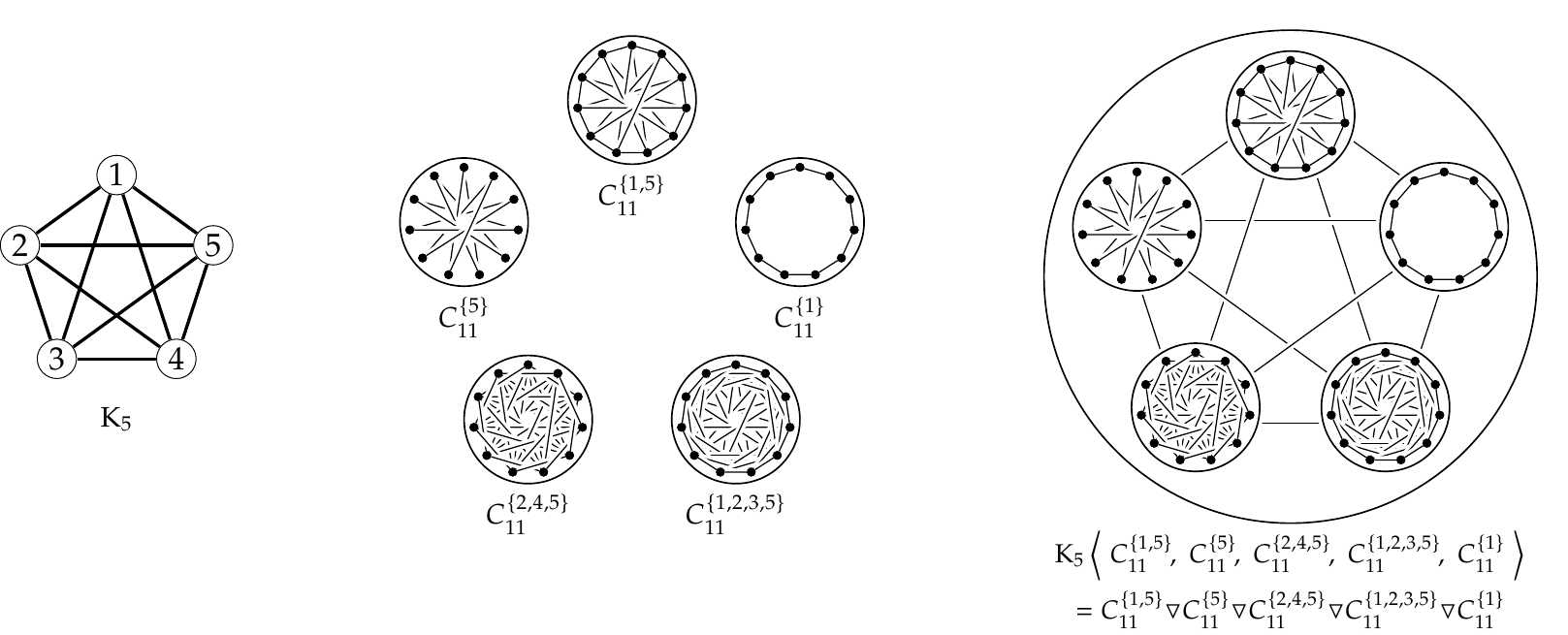}
        \caption{An clique that is inhabited by 5 difference graphs.
        The resulting graph is also the join of the 5 graphs.}
    \end{subfigure}
    \caption{Examples for fixed points in \(\fpb{\rotgr{5}^d}{K_{5}^d}\), which are also
    inhabited graphs.}
    \label{fig:example:fixed:points:prod}
\end{figure}

Consult \cref{fig:example:fixed:points:prod} for a visualization of an example for
\cref{def:10.31-1}.
It is instructive to briefly discuss unions of inhabited graphs.

\begin{lemma}\label{10.31-2}
    For two inhabited graphs \(F_1 = \metaGraph{C_1}{G_1, \dots, G_m}\) and
    \(F_2 = \metaGraph{C_2}{G'_1, \dots, G'_m}\) with a common vertex set, we
    have\footnote{Recall that we use \(G \cup H\) to denote the graph on \(V = V(G)
    =V(H)\) with edges \(E(G) \cup E(H)\).}
    \[
        F_1 \cup F_2 = \metaGraph{(C_1 \cup C_2)}{G_1 \cup G_1', \dots, G_m \cup G_m'}.
    \]
\end{lemma}
\begin{proof}
    The union of (edge) sets is an associative operation.
    Unfolding the definitions of inhabited graphs and their union yields the claim.
\end{proof}

Next, we define a standard operation on groups that mirrors the disjoint union of graphs.
Namely, the products of groups.

\begin{definition}
    For permutation groups $\Gamma_1 = (G_1, \circ), \dots, \Gamma_m = (G_m, \circ)$
    with $\Gamma_i \subseteq \sym{X_i}$,
    their product group $\Gamma \coloneqq \Gamma_1 \prodGroup \cdots \prodGroup \Gamma_m$ is
    the set $G \coloneqq G_1 \times \cdots \times G_m$ together with the component-wise
    function composition, that is, for \((g_1, \dots, g_m), (g_1', \dots,
    g_m') \in \Gamma\), we set
    \[
        (g_1, \dots, g_m) \circ (g_1', \dots, g_m')
        \coloneqq (g_1 \circ  g_1',\dots, g_m \circ g_m').
    \]

    We let \(\Gamma\) act on \(X \coloneqq X_1 \unionSet \cdots \unionSet X_m\) via
    \[
        g(j, x_j) \coloneqq (j, g_j(x_j)), \text{ for all } g = (g_1, \dots, g_m) \in G
        \text{ and } x = (j, x_j) \in X.
        \tag*{\qedhere}
    \]
\end{definition}

\begin{fact}
    The order of $G$ is $\prod_{i = 1}^m |G_i|$.
    In particular, the product group of $p$-groups is still a $p$-group.

    Further, observe that $G$ is a permutation group of $X_1 \unionSet \cdots \unionSet X_m$.
    Hence, we have $G \subseteq \sym{X_1 \unionSet \cdots \unionSet X_m}$.
\end{fact}

Finally, we let product groups act on joins of graphs; in particular, we are interested in
the resulting fixed points.
To that end, we first study the orbits that appear.

\begin{lemma}\label{10.31-3}
    For \(i  \in \setn{m}\), let \(G_i\) denote a graph with \(n(i) \coloneqq
    |V(G_i)|\) vertices and let \(\Gamma_i
    \subseteq \auts{G_i}\) denote a transitive permutation group.
    Further, let \(\Gamma_i\) act on \(E(G_i)\) and write
    \(E(G_i)/\Gamma_i =  \{ O^{i}_1, \dots, O^{i}_{s_i}\} \) for the resulting orbits.
    Finally, let \( \BigProdGroup_{i=1}^{m} \Gamma_i \) act on
    \(E(\BigJoinGraph_{i = 1}^m G_i)\).

    Then, we have
    \begin{align*}
        E\left(\BigJoinGraph_{i = 1}^m G_i\right) / \BigProdGroup_{i=1}^{m} \Gamma_i
        &= \{
            \metaGraph{\oneEdge{i}{j}}{\IS_{n(1)}, \dots, \IS_{n(m)}}  \mid
            i\neq j \in \setn{m} \}
      \\&\cup
            \{
                \metaGraph{\IS_m}{\IS_{n(1)}, \dots, \IS_{n({i-1})}, \smash{O^i_j},
            \IS_{n({i+1})}, \dots,\IS_{n(m)}}  \mid
            i\in \setn{m}, j \in \setn{s_i}
        \}.
    \end{align*}
\end{lemma}
\begin{proof}
    Recall that \(\Gamma \coloneqq \BigProdGroup_{i=1}^{m} \Gamma_i \) acts on
    \(E(\BigJoinGraph_{i = 1}^m G_i)\) by mapping
    the group element $(\alpha_1, \dots, \alpha_m)$ and the
    edge $\{(i, u), (j, v) \}$ to $\{(i, \alpha_i(u)), (j, \alpha_j(v)) \}$.

    We start by analyzing the orbits of the edges inside of the graphs \(G_i\).
    \begin{claim}\label{11.1-1}
        For every  \(i\in \setn{m}\) and \(O \in E(G_i)/\Gamma_i\), we have
        \[
                \metaGraph{\IS_m}{\IS_{n(1)}, \dots, \IS_{n({i-1})}, \smash{O},
            \IS_{n({i+1})}, \dots,\IS_{n(m)}}
            \in
            E\left(\BigJoinGraph_{i = 1}^m G_i\right) / \BigProdGroup_{i=1}^{m} \Gamma_i.
        \]
    \end{claim}
    \begin{claimproof}
        Fix an orbit \(O \in E(G_i)/\Gamma_i\) and write \(\{ u,v \} \in O\) for an edge
        of \(O\).
        We compute the orbit of the edge \(\{ (i,u),(i,v) \} \) under \(\Gamma\).

        Observe that whenever \(\Gamma\) acts on an edge \( \left\{ ( i, u ) ,( i, v
        )\right\}\) of
        \(\BigJoinGraph_{i = 1}^m G_i\), we recover the action of \(\Gamma_i\) on said
        edge: writing \(\star|_2\) for the restriction of pairs to their second component,
        we obtain
        \[
        ( g_1,\dots,g_i,\dots,g_m )\left\{ ( i, u ) ,( i, v )\right\}|_2
        =  \left\{ ( i, g_i(u) ), ( i, g_i(v) )\right\}|_2
        = \{ g_i(u), g_i(v) \}
        = g_i\{ u,v \}.
        \]
        In particular, we obtain \(\Gamma \{ ( i,u ) ,( i,v )  \}|_2 = \Gamma_i  \{ u,v\}
        = O\).

        Finally, we observe that adding isolated vertices to a graph has no effect on the
        orbits of the edges of the graph, which yields
        \[
            \Gamma \{ ( i,u ) ,( i,v )  \}
            =
                \metaGraph{\IS_m}{\IS_{n(1)}, \dots, \IS_{n({i-1})}, \smash{O},
            \IS_{n({i+1})}, \dots,\IS_{n(m)}};
        \] and thus the claim.
    \end{claimproof}

    Next, we analyze the orbits of edges between graphs \(G_i\).

    \begin{claim}\label{11.1-2}
        For every \(i \neq j \in \setn{m} \), we have
        \[
            \metaGraph{\oneEdge{i}{j}}{\IS_{n(1)}, \dots, \IS_{n(m)}}
            \in
            E\left(\BigJoinGraph_{i = 1}^m G_i\right) / \BigProdGroup_{i=1}^{m} \Gamma_i.
        \]
    \end{claim}
    \begin{claimproof}
        Fix $F \coloneqq
        \metaGraph{\oneEdge{i}{j}}{\IS_{n(1)}, \dots, \IS_{n(m)}}$ and observe that \(F\)
        consists in isolated vertices and a complete bipartite graph on \(n(i) + n(j)\)
        vertices.

        We consider vertices \(u \in V(G_i)\) and \(v \in V(G_j)\) and compute the orbit
        of the edge \(\{ (i,u), (j,v) \} \) under~\(\Gamma\).
        To that end, we observe that, as \(\Gamma_i\) and \(\Gamma_j\) are transitive, we
        have \(\Gamma_i u = V(G_i)\) and \(\Gamma_j v = V(G_j)\).
        In particular,
        we have \(\Gamma \{ (i,u), (j,v) \} = \{ ( i, a ) ,( j,b )  \mid a\in V(G_i), b\in
        V(G_j)\} \)---which are precisely the edges of \(F\), which in turn yields the claim.
    \end{claimproof}

    Finally, we readily see that \cref{11.1-1,11.1-2} cover all orbits in
    \(E\left(\BigJoinGraph_{i = 1}^m G_i\right) / \BigProdGroup_{i=1}^{m} \Gamma_i\).
    First, any edge  \(\{ (i,u),(i,v)\}\) is covered by an orbit that corresponds to an
    orbit in \(E(G_i)/\Gamma_i\).
    Second, any edge \(\{ (i,u),(j,v)\}\) is covered by the orbit
    \( \metaGraph{\oneEdge{i}{j}}{\IS_{n(1)}, \dots, \IS_{n(m)}}\).
    In total, this completes the proof.
\end{proof}

As before, we use orbits to build fixed points.

\begin{corollary}\label{theo:prod:fixed:points}
    For \(i  \in \setn{m}\), let \(G_i\) denote a graph
    and let \(\Gamma_i
    \subseteq \auts{G_i}\) denote a transitive permutation group.
    Further, let \(\Gamma_i\) act on \(E(G_i)\) and write
    \(\fpb{\Gamma_i}{G_i}\) for the resulting fixed points.
    Finally, let \( \BigProdGroup_{i=1}^{m} \Gamma_i \) act on
    \(\edgesub(\BigJoinGraph_{i = 1}^m G_i)\).

    Then, we have
    \begin{align*}
        \fpb{\BigProdGroup_{i = 1}^m \Gamma_i}{\BigJoinGraph_{i = 1}^m G_i} =
    \left\{ \metaGraph{C}{A^1, \dots, A^m} \mid C \in \graphs{m}  , A^i \in \fp(\Gamma_i,
    G_i)\right\},
    \end{align*}
\end{corollary}
\Cref{theo:prod:fixed:points} allows us to interpret the fixed points
$\fpb{\BigProdGroup_{i = 1}^m \Gamma_i}{\BigJoinGraph_{i = 1}^m G_i}$
as combinations of the fixed points of $\fpb{\Gamma_i}{G_i}$.
In particular, each fixed point $\metaGraph{C}{A^1, \dots, A^m}$ of
consists in $m$ blocks, where the $i$-th block is a fixed point $A^i \in
\fpb{\Gamma_i}{G_i}$ and we fully connect two blocks $i$ and $j$ with each other if
$\{i, j\} \in E(C)$.
\begin{proof}
    Recall that by \cref{remark:fixed:point:union} (the edge set of) each fixed point
    $F \in \fpb{\BigProdGroup_{i = 1}^m \Gamma_i}{\BigJoinGraph_{i = 1}^m G_i}$
    is a union of orbits from $E\left(\BigJoinGraph_{i = 1}^m G_i\right) /
    \BigProdGroup_{i=1}^{m} \Gamma_i$.

    From \cref{10.31-3},
    we understand the orbits $E\left(\BigJoinGraph_{i = 1}^m G_i\right) /
    \BigProdGroup_{i=1}^{m} \Gamma_i$ as inhabited graphs that are
    pairwise compatible.
    Further, from \cref{10.31-2}, we obtain that we may compute the union of compatible
    inhabited graphs in a block-wise fashion.

    Next, we apply \cref{remark:fixed:point:union} to each of the blocks, that is,
    each fixed point \(F_i \in \fpb{\Gamma_i}{G_i}\) is a union of orbits from
    \(E(G_i)/\Gamma_i\) (which are precisely the inner blocks of the orbits
    $E\left(\BigJoinGraph_{i = 1}^m G_i\right) / \BigProdGroup_{i=1}^{m} \Gamma_i$).
    Finally, we observe that we may obtain any \(m\)-vertex graph as the (edge-)union of
    \(m\)-vertex graphs that have a single edge.

    In total, this completes the proof.
\end{proof}

\subsection{The Fixed Points of Rotations of \texorpdfstring{$K^d_{p^m}$}{K p m d}}
\label{sec:fixedpoint:kdpm}

For a positive integer \(d\) and a prime power \(p^m\),
write $K^d_{p^m} \coloneqq K_{p^m} \JoinGraph \cdots \JoinGraph K_{p^m}$ for the join of
$d$ copies of the graph $K_{p^m}$---observe that $K^d_{p^m}$ is the complete graph on the
vertex set $\setn{d} \times \field{p^m}$.
Further, write $\rotgr{p^m}^d \coloneqq \rotgr{p^m} \prodGroup \dots \prodGroup
\rotgr{p^m}$ for the product of $d$ copies of the group $\rotgr{p^m}$.

Next, we use \cref{cor:fix:Hs,theo:prod:fixed:points} to understand
$\fpb{\rotgr{p^m}^d}{K^d_{p^m}}$.

\begin{lemma}\label{10.31-4}
    Let \(p\) denote a prime and let \(m\) and \(d\) denote positive integers.
    Further, let \(\rotgr{p^m}^d =  K_{p^m} \JoinGraph \cdots \JoinGraph K_{p^m}\)
    act on \(\edgesub(K^d_{p^m}) = \rotgr{p^m} \prodGroup \dots \prodGroup \rotgr{p^m}\).
    Then, we have
    \begin{align*}
        \fpb{\rotgr{p^m}^d}{K^d_{p^m}} &=
        \{\metaGraph{G}{\cgr{p^m}{A^1}, \dots, \cgr{p^m}{A^d}}
            \mid G \in \graphs{d}, A^i \subseteq \field{p^m}^+ \}.
    \end{align*}
\end{lemma}
\begin{proof}
    Recall that by \cref{cor:fix:Hs}, we have
    \[\fpb{\rotgr{p^m}}{K_{p^m}} = \{ \cgr{p^m}{A}  \mid A \subseteq \field{p^m}^+\}.\]
    Observe that $\rotgr{p^m}$ is transitive, since for all $x, y \in \field{p^m}$ we can
    find a element $\varphi_{y - x} \in \rotgr{p^m}$ with $\varphi_{y - x}(x) = x + (y -
    x) = y$.
    Hence, we use \cref{theo:prod:fixed:points} to obtain the claim.
\end{proof}

Consult again \cref{fig:example:fixed:points:prod} for visualizations of examples for
fixed points \(\fpb{\rotgr{p^m}^d}{K^d_{p^m}}\).

Next, for our ETH-based lower bounds, it is instructive to observe that every fixed point
\(\metaGraph{G}{\cdots} \in \fpb{\rotgr{p^m}^d}{K^d_{p^m}}\) contains a large biclique as
a subgraph as long as \(G\) contains at least one edge.

\begin{lemma}\label{rem:treewidth:biclique}
    Let \(p\) denote a prime and let \(m\) and \(d\) denote positive integers.

    Every fixed point $\metaGraph{G}{\cgr{p^m}{A^1}, \dots, \cgr{p^m}{A^d}} \in \fpb{\rotgr{p^m}^d}{K^d_{p^m}}$
    with $G \neq \IS_m$ contains the graph $K_{p^m, p^m}$ as a subgraph and has
    treewidth of at least \(p^m\).
\end{lemma}
\begin{proof}
    Consider an edge \(\{i , j\} \in E(G)\). Both graphs \(\cgr{p^m}{A^i}\) and
    \(\cgr{p^m}{A^j}\) have \(p^m\) vertices; by definition of an inhabited graph, said
    vertices are connected with a complete bipartite graph.

    Finally, a graph with \(K_{p^m,p^m}\) as a subgraph has a treewidth of at least \(p^m\).
\end{proof}

Finally, let us compute the level of a fixed point \(F \in
\fpb{\rotgr{p^m}^d}{K^d_{p^m}}\).

\begin{lemma}
    Let \(p\) denote a prime and let \(m\) and \(d\) denote positive integers.

    For every fixed point $F = \metaGraph{G}{\cgr{p^m}{A^1}, \dots, \cgr{p^m}{A^d}} \in
    \fpb{\rotgr{p^m}^d}{K^d_{p^m})}$, we have
    \[
    \hl{F} = \hl{\metaGraph{G}{\cgr{p^m}{A^1}, \dots, \cgr{p^m}{A^d}}} = |E(G)| + \sum_{i = 1}^d |A^i|.
    \]
\end{lemma}
\begin{proof}
    \Cref{theo:prod:fixed:points,10.31-3} allow us to compute the level of \(F\) as the
    sum of the level of \(G\) and the levels of the fixed points \(\cgr{p^m}{A}\).
    To that end, we observe that every edge of \(G\) contributes one orbit to the level of
    \(G\); each fixed point  \(\cgr{p^m}{A}\) has a level of \(|A|\).
    In total, this yields the claim.
\end{proof}

\subsection{The Property \texorpdfstring{$(\Phi-H)$}{(Φ-H)}}\label{sec:inclusion-exclusion-reduction}
We often use the following definition: given a property $\Phi$ and a graph $H$, we define
$(\Phi-H)$ to be the property that contains a graph only if it satisfies $\Phi$ when we
extend the graph with $H$ as a disjoint union. Using standard techniques, we show that
$\NUM{}\indsubsprob(\Phi-H)$ is not harder than $\NUM{}\indsubsprob(\Phi)$: we simply add a
copy of $H$ to the input graph and use the Inclusion-Exclusion Principle to count only
those induced subgraphs that fully contain this copy.

\begin{lemma}\label{lem:inc:exc}
    Let \(\Phi\) denote a graph property and suppose that there is an algorithm
    \(\mathbb{A}\) that
    computes for each graph \(G\) and integer \(k\) the value \(\NUM{}\indsubs{(\Phi,k)}{G}\) in time \(g(k,
    |V(G)|)\) for some computable function \(g\) that is monotonically increasing.
    Finally, for a graph \(H\), write
    \(
        (\Phi - {H}) \coloneqq \{ G \mid G \UnionGraph H  \in \Phi \}
    \)
    for the graph property of all graphs that is extended by $H$ to a graph in \(\Phi\).

    Then, there is an algorithm \(\mathbb{B}\) with oracle access to \(\mathbb{A}\) that computes for
    each graph \(G\) and positive integer \(k\) the value $\NUM{}\indsubs{((\Phi - H), k)}{G}$
    in time
    \[O(2^{|V(H)|} \cdot (|V(G)| + |V(H)|)^2 \cdot g(|V(H)| + k, |V(H)| + |V(G)|)).\]
    The algorithm \(\mathbb{B}\) queries \(\mathbb{A}\) on instances with a parameter of
    \(|V(H)| + k\).
\end{lemma}
\begin{proof}
    Fix a graph \(G\) and a positive integer \(k\).
    Observe that for any induced subgraph \(G\position{X} \in \Phi-H\),
    we have \(G\position{X} \UnionGraph H \cong (G \UnionGraph H)\position{X} \in \Phi\).
    In particular, once we extend \(G\) with
    \(H\), we may use the algorithm \(\mathbb{A}\) to compute the number
    \( \NUM{}\indsubs{(\Phi, k + |V(H)|)}{G \UnionGraph H}\).
    We then use the Inclusion-Exclusion Principle to recover the number
    \(\NUM{}\indsubs{((\Phi - H), k)}{G}\).

    For a formal proof, write \((\mathcal{G} \UnionGraph H)_{k + |V(H)|}\) for all size-\((k + |V(H)|)\) induced
    subgraphs of \(G \UnionGraph H\) that satisfy  \(\Phi\), that is,
    \begin{align*}
        (\mathcal{G} \UnionGraph H)_{k + |V(H)|}  &\coloneqq
        \{(G \UnionGraph H)\position{X} \mid |X| = |V(H)| + k \}
        \cap \Phi.
    \end{align*}
    Next, for each vertex \(x \in V(H)\), write \(\mathcal{G}_x\) for the set of all
    graphs in \((\mathcal{G} \UnionGraph H)_{k + |V(H)|}\) that contain the
    vertex~\(x\),\footnote{Technically, the operation \(\JoinGraph\) renames the vertex
    \(x\) to \((i, x)\) for some \(i\). We may safely ignore this detail, as it is not
    relevant for our proof.}
    that is,
    \begin{align*}
        \mathcal{G}_x &\coloneqq
        \{F \in (\mathcal{G} \UnionGraph H)_{k + |V(H)|} \mid x \in V(F)\}.
    \end{align*}

    Now, we first show that we can rewrite \(\NUM{}\indsubs{((\Phi - H)_k, k)}{G}\) as the number of
    graphs of \( (\mathcal{G} \UnionGraph H)_{k + |V(H)|}  \) that contain all vertices of \(H\).
    \begin{claim}\label{11.3-1}
        We have \[
            \NUM{}\indsubs{((\Phi - H), k)}{G} =
            \big|\!\bigcap_{x \in V(H)}\!\! \mathcal{G}_x\;\big|.
        \]
    \end{claim}
    \begin{claimproof}
        We prove both inclusions separately.

        First, fix a graph \(G\position{X} \in \indsubs{((\Phi - H), k)}{G}\).
        By definition, the graph \(G\position{X} \UnionGraph H\) has \((k + |V(H)|)\)
        vertices and the property~\(\Phi\).
        Further, as \(X\) and \(V(H)\) are in disjoint subgraphs of \(G\position{X} \UnionGraph H\),
        the graph $G\position{X} \UnionGraph H$ is isomorphic to the graph
        $(G \UnionGraph H)\position{X}$; in particular, we have \(G\position{X} \UnionGraph H \in
        (\mathcal{G} \UnionGraph H)_{k + |V(H)|} \).
        Finally, we observe that the graph \(G\position{X} \UnionGraph H\) contains every
        vertex \(v \in V(H)\). Hence, we have \[
            G\position{X} \UnionGraph H \in \bigcap_{x \in V(H)}\!\! \mathcal{G}_x.
        \]
        For the other direction, fix a graph $F \coloneqq (G \UnionGraph H)\position{X}$
        with $F \in \bigcap_{x \in V(H)} \mathcal{G}_x$.
        In particular, the graph \(F\) has \(k + |V(H)|\) vertices and we have \(F \in \Phi\).
        As \(F\) contains every vertex of \(H\), we also have that \(F\) contains \(H\) as
        an induced subgraph. This in turn means that we have
        \(F \cong G\position{X} \UnionGraph H\).
        Hence, we also have \(G\position{X} \in \indsubs{((\Phi - H), k)}{G}\), which
        completes the proof.
    \end{claimproof}

    Ultimately, we wish to remove the graph \(H\) from our oracle calls.
    Toward an application of the Inclusion-Exclusion Principle, we need to understand how
    to partially remove \(H\) from our oracle calls.
    To that end,  write \(\overline{\mathcal{G}_x}\) for the complement of
    \(\mathcal{G}_x\), that is, set \(\overline{\mathcal{G}_x} \coloneqq (\mathcal{G} \UnionGraph H)_{k + |V(H)|}
    \setminus \mathcal{G}_x\).
    Now, we show that for any \(X \subseteq V(H)\), we can rewrite \(  \NUM{}\indsubs{(\Phi,
    |V(H)| + k)}{G \UnionGraph (H \setminus  X)}\) as the number of
    graphs of \((\mathcal{G} \UnionGraph H)_{k + |V(H)|} \) that contain no vertex in
    \(X\).
    \begin{claim}\label{11.3-2}
    For every $X \subseteq V(H)$, we have
    \[
        \big|\bigcap_{x \in X} \overline{\mathcal{G}_x} \big|
        = \NUM{}\indsubs{(\Phi, |V(H)| + k)}{G \UnionGraph (H \setminus  X)}.
    \]
    \end{claim}
    \begin{claimproof}
        Unfolding the definition of \(\overline{\mathcal{G}_x}\), we obtain
        \[
            \bigcap_{x \in X} \overline{\mathcal{G}_x} =
            (\mathcal{G} \UnionGraph H)_{k + |V(H)|} \setminus \bigcup_{x \in X} \mathcal{G}_x.
        \] Now,
        $(\mathcal{G} \UnionGraph H)_{k + |V(H)|} \setminus \bigcup_{x \in X}
        \mathcal{G}_x$
        is the set of all induced
        subgraphs of $G \UnionGraph H$ of size $k + |V(H)|$ that satisfy $\Phi$
        and do not contain any vertex in $X$;
        which is precisely
        $\NUM{}\indsubs{(\Phi, k + |V(H)|)}{G \UnionGraph (H
        \setminus X)}$.
    \end{claimproof}

    Finally we express \(|\bigcap_{x \in V(H)} \mathcal{G}_x|\) in terms of
    \(|\bigcap_{x \in X} \overline{\mathcal{G}_x}|\).

    \begin{claim}\label{11.3-3}
        We have
        \[
        \big|\bigcap_{x \in V(H)} G_x\big|
            = \sum_{X \subseteq V(H)} (-1)^{|X| + 1} \big|\bigcap_{x \in X} \overline{G_x}
            \big|.
        \]
    \end{claim}
    \begin{claimproof}
        After unfolding the definition of \(\overline{\mathcal{G}_x}\), we see that the
        claim is equivalent to the classical Inclusion-Exclusion Principle for set
        intersection; which yields the claim.
    \end{claimproof}

    Now, the combination of \cref{11.3-1,11.3-2,11.3-3} yields
    \begin{align*}
            \NUM{}\indsubs{((\Phi - H), k)}{G}
            &= \sum_{X \subseteq V(H)} (-1)^{|X| + 1} \NUM{}\indsubs{(\Phi, k + |V(H)|)}{G
            \UnionGraph (H \setminus X)}.
    \end{align*}

    In particular, we can compute $\NUM{}\indsubs{((\Phi - H), k)}{G}$ by calling the algorithm
    \(\mathbb{A}\) with parameter \(k + |V(H)|\) on \(2^{|V(H)|}\) graphs \(G'\)
    with $|V(G')| \leq |V(G)| + |V(H)|$.
    Further, we can
    construct each graph $G'$ in time $O((|V(G)| + |V(H)|)^2)$.
    In total, we obtain a running time of
    $O(2^{|V(H)|} \cdot (|V(G)| + |V(H)|)^2 \cdot g(k + |V(H)|, |V(G)| + |V(H)|))$; which
    completes the proof.
\end{proof}

\subsection{Scattered Properties and Reducing to the Prime-Power Case}\label{sec:scattered}

For a positive integer \(n\), let us write $q(n)$ for the largest divisor of \(n\) that is a
prime power; we set $q(1) \coloneqq 1$.
It is instructive to discuss some useful properties of the function \(q(n)\).
\begin{lemma}\label{lem:upper:bound:q}
    \begin{enumerate}[(1)]
        \item For every positive integer $n$, we have $n \leq q(n)^{q(n)}$.
        \item Write $n \coloneqq p_1^{a_1} \cdots p_c^{a_c}$ for a positive
            integer and its corresponding prime factorization.
            Then, we have $\sqrt[c]{n} \le q(n)$.
        \item There is a universal constant $c > 0$ such that for every positive integer
            \(n\), we have $c \log(n) \le q(n)$.
    \end{enumerate}
\end{lemma}
\begin{proof}
    For (1), write $n = p_1^{a_1} \cdots p_c^{a_c}$ for the prime factorization of \(n\).
        We bound each prime power factor of \(n\) with \(q(n)\), which yields \(n \le
        q(n)^c\). Finally, we use $c \leq \max(p_1, \dots, p_c) \leq q(n)$ to obtain the claim.

    For (2), observe that we have $q(n) = p_i^{a_i}$ for some $i \in \setn{c}$.
        Since $p_i^{a_i}$ is the largest prime power factor of \(n\), we obtain $q(n)^c = (p_i^{a_i})^c
        \geq n$; which yields the claim.

    For (3), write $\omega(n)$ for the number of prime divisors of $n$.
        We show $q(n) \geq c \log(n)$ for all $n \geq 26$; sufficiently decreasing $c$
        then yields the claim.

        We use the following result due to Robin~\cite[Theorem 13]{Robin1983}. For every
        \(n \ge 26\), we have
        \begin{align}\label{eq:10.11-1}
            \omega(n) \leq \frac{\log n }{ \log(\log(n)) - 1.1714}.
        \end{align}
        Combining \cref{eq:10.11-1} with Claim~(2), for every $n \geq 26$,
        we obtain
        \[q(n) \geq n^{\omega(n)^{-1}} \geq n^{\frac{\log(\log(n)) - 1.1714}{\log(n)}} =
        \log(n)\;\mathrm{e}^{-1.1714}; \]
        which yields the claim.
\end{proof}

Let us also recall the definition of concentrated and scattered integers for \(\Phi\).

\defconsca

Let us observe that if $\Phi$ is computable, then we can decide which case holds for a given $k$.

\begin{lemma}\label{lem:case:decidability}
    Let $\Phi$ denote a computable edge-monotone property. For every integer $k$, we can
    decide if $\Phi$ is trivial, concentrated, or scattered on $k$.
\end{lemma}
\begin{proof}
    By evaluating $\Phi$ on every $k$-vertex graph, we can decide if $\Phi$ is trivial. If
    \(\Phi\) is not trivial on \(k\),
    then computing the alternating enumerator of every $k$-vertex graph that
    contains $K_{q(k),q(k)}$ as a subgraph tells us if $k$ is concentrated.
\end{proof}

In the next step, we show the key result toward our win-win approach: for an edge-monotone
property \(\Phi\) and any \(k\) on which \(\Phi\) is nontrivial, we either win by
obtaining a high-treewidth graph with nonvanishing alternating enumerator (if \(\Phi\) is
concentrated on \(k\)), or we win by
reducing to the prime power case (if \(\Phi\) is scattered on \(k\)).
Formally, we prove \cref{lem:reduction}.

\lemmareduction
\begin{proof}
    Fix an integer \(k \in \nt{\Phi}\).
    We prove that either \(k\) is concentrated or the claim holds.

    To that end, we let $\rotgr{q(k)}^{k/q(k)}$ act on $K^{k/q(k)}_{q(k)}$ and consider the resulting
    fixed points \(\fpb{\rotgr{q(k)}^{k/q(k)}}{K^{k/q(k)}_{q(k)}}\).
    As \(\Phi\) is nontrivial on \(k = q(k) \cdot k/q(k)\), there is a minimal integer \(i\)
    such that

    \begin{itemize}
        \item there is a fixed point with a level of $i$ that does not satisfy $\Phi$ but
        \item all fixed points with a level of less than $i$ satisfy $\Phi$.
    \end{itemize}

    As \(\Phi\) is nontrivial on \(k\), we have \(\Phi(\IS^{q(k)/d}_{q(k)}) = 1\) and
    \(\Phi(K^{q(k)/d}_{q(k)}) = 0\). In particular, as \(\IS^{q(k)/d}_{q(k)}\) is the (only)
    fixed point with the minimum level of \(0\) and as
    \(\smash{K^{q(k)/d}_{q(k)}}\) is the fixed point with the maximum level, we obtain that \(i\) satisfies
    \[
        0 = \hl{\IS^{q(k)/d}_{q(k)}} < i \le  \hl{K^{q(k)/d}_{q(k)}} = \binom{k/q(k)}{2} +
        \NUM{}\field{q(k)}^+ \cdot k/q(k).
    \]

    Write \(F \coloneqq \metaGraph{C}{\cgr{q(k)}{A^1}, \dots, \cgr{q(k)}{A^{k/q(k)}}}\) for a fixed
    point with \(\hl{F} = i\) and \(\Phi(F) = 0\).
    Next, we distinguish two cases based on whether \(C\) contains an edge.

    First, consider the case that \(C\) contains some edge.
    Now, \cref{remark:chi:all:children:true} yields $\ae{\Phi}{F} \neq 0$.
    Further, \cref{rem:treewidth:biclique} yields that \(F\) has $K_{q(k), q(k)}$ as a subgraph.
    Thus, \(\Phi\) is concentrated on \(k\).

    Second, consider the remaining case that $C = \IS_{k/q(k)}$.
    This in turn means that \(F\) is the disjoint union of \(k/q(k)\) fixed points of
    \(\fpb{\rotgr{q(k)}}{K_{q(k)}}\); we write
    \[
    F = \cgr{q(k)}{A^1} \UnionGraph \cdots \UnionGraph \cgr{q(k)}{A^{k/q(k)}}.
    \]
    As $\hl{F} = i \geq 1$, there is at least one block $x \in \setn{k/q(k)}$  with $A^x \neq \emptyset$.
    Without loss of generality, we may assume $x = 1$ (otherwise, pick the isomorphic
    graph with renamed blocks); and in particular \(\cgr{q(k)}{A^1} \neq \IS_{q(k)}\).
    Now, we set
    \[H \coloneqq \cgr{q(k)}{A^2} \UnionGraph \cdots \UnionGraph \cgr{q(k)}{A^{k/q(k)}}. \]
    Further, we define the  graph
    property \((\Phi - H)\) via
    \begin{align*}
        (\Phi - {H}) \coloneqq \{ G  \mid G \UnionGraph H \in \Phi  \}.
    \end{align*}
    As \(\Phi\) is edge-monotone and computable, so is \((\Phi - H)\).
    Finally, we show that $(\Phi - H)$ is nontrivial on $q(k)$.

    \begin{claim}\label{11.4-1}
        We have \(\IS_{q(k)} \in  (\Phi - H)\)
        and
        \(\cgr{q(k)}{A^1} \not\in  (\Phi - H)\).
    \end{claim}
    \begin{claimproof}
        First, the graph $\IS_{q(k)} \UnionGraph\; H$ is isomorphic to the fixed point $F'
        = \IS_{q(k)} \UnionGraph\; \cgr{q(k)}{A^2} \UnionGraph \cdots \UnionGraph
        \cgr{q(k)}{A^{k/q(k)}}$.
        Further, we have \(\hl{F'} < \hl{F} = i\) as \(\cgr{q(k)}{A^1}\) is not the empty
        graph.
        Thus, \(F'\) satisfies $\Phi$ as, by construction, all fixed points with a level
        less than $i$ satisfy \(\Phi\).

        Second, observe that we have that $\cgr{q(k)}{A^1} \UnionGraph H$ is isomorphic to $F$, hence
        \[
             (\Phi - H)(\cgr{q(k)}{A^1}) = \Phi(\cgr{q(k)}{A^1} \UnionGraph H) = \Phi(F) = 0.
             \tag*{\claimqedhere}
        \]
    \end{claimproof}
    \Cref{11.4-1} shows that, indeed, \(\Phi\) is scattered on \(k\).
    This completes the proof.
\end{proof}

Next, we show that scattered integers indeed yield a reduction to the prime power-power
case.
To that end, we first use the scattered part of a graph property define another graph
property on prime powers.

\begin{definition}\label{11.5-2}
    Let \(\Phi\) denote a computable, edge-monotone graph property
    and write \(\scat{\Phi}\) for the set of all integers on which \(\Phi\) is scattered.
    Further, write \(q(\scat{\Phi}) \coloneqq \{ q(k)  \mid k \in \scat{\Phi}\}\) for the set
    of all maximal prime powers corresponding to \(\scat{\Phi}\).
    Finally, for each \(m \in
    q(\scat{\Phi})\) write \(q^{-1}(m)\) for the minimum \(k \in \scat{\Phi}\) with \(q(k) =
    m\).

    Now, for each \(m \in q(\scat{\Phi})\), let \(H_m\) denote the lexicographically first graph on \(q^{-1}(m) - m\)
    vertices such that the graph property
    \((\Phi - H_{m}) \coloneqq \{ G  \mid G \UnionGraph H_{m} \in \Phi\}\)
    is nontrivial on $m$.

    We define the scattered property \(\Phi_{\scat{}}\) corresponding to $\Phi$ as
    \[
        \Phi_{\scat{}}(G)=1 \iff |V(G)|\in q(\scat{\Phi}) \text{ and } \Phi(G\UnionGraph H_{|V(G)|})=1
        \tag*{\qedhere}
    \]
\end{definition}

Observe that \cref{lem:reduction} implies that $H_m$ is well-defined in \cref{11.5-2}, as
there is at least one such graph.
Let us verify the main properties of the defined function $\Phi_{\scat{}}$.

\begin{lemma}\label{11.5-1}
    Let \(\Phi\) denote a computable, edge-monotone graph property
    and write \(\scat{\Phi}\) for the set of all integers on which \(\Phi\) is scattered.
    If \(\scat{\Phi}\) is infinite, then
     the scattered property
    \(\Phi_{\scat{}}\) corresponding to \(\Phi\)
    is computable, edge-monotone, and nontrivial on
    infinitely many prime powers.
\end{lemma}
\begin{proof}
    As in \cref{11.5-2}, write \(q(\scat{\Phi}) \coloneqq \{ q(k)  \mid k \in \scat{\Phi}\}\) for the set
    of all maximal prime powers corresponding to \(\scat{\Phi}\).
    Next, for each \(m \in
    q(\scat{\Phi})\) write \(q^{-1}(m)\) for the minimum \(k \in \scat{\Phi}\) with \(q(k) =
    m\).

    \begin{claim}
        The sets \(\scat{\Phi}\) and \(q(\scat{\Phi})\), the function \(q^{-1}\), and the
        graphs $H_m$ for every $m\in q(\scat{\Phi})$ are computable.
    \end{claim}
    \begin{claimproof}
        Fix an integer \(k\). As \(\Phi\) is computable and by using
        \cref{lem:case:decidability}, we can compute whether \(k \in \scat{\Phi}\).

        Next, we wish to decide \(k \in q(\scat{\Phi})\).
        To that end, we iterate through the integers $i$ starting from $k$ and ending
        with $k^k$.
        For each such \(j\), we first check if $j \in \scat{\Phi}$.
        If this is the case, we compute \(q(j)\) and check if \(k = q(j)\).
        If indeed \(k = q(j)\), we return \(k \in q(\scat{\Phi})\).
        Otherwise, if we find that for all integers \(k \le j \le k^k\), we have \(j
        \not\in \scat{\Phi}\), we return \(k \not\in q(\scat{\Phi})\).

        This algorithm is correct, as by \cref{lem:upper:bound:q}(1), any integer \(k\) may
        appear as the largest prime power only for integers that are at most \(k^k\).
        Finally, said algorithm also yields the smallest value \(j\) with \(q(j) = k\) ;
        which proves that the function \(q^{-1}\) is computable as well.

        Recall that, for $m\in q(\scat{\Phi})$, the graph $H_m$ has $q^{-1}(m)-m$ vertices (which
        is a number we can compute). Let us enumerate in lexicographic order the graphs $H_m$ on
        \(q^{-1}(m)-m\) vertices and check if $(\Phi-H_m)$ is a nontrivial property on $m$ vertices. As
        $\Phi$ is computable, this can be done by enumerating every graph on $m$ vertices.
    \end{claimproof}

    Finally, we observe that as \(\scat{\Phi}\) is infinite, so is \(q(\scat{\Phi})\)
    (as by \cref{lem:upper:bound:q}(3), we have \(q(k) \ge c \log(k)\) for all integer
    \(k\)). Thus \(\Phi_{\scat{}}\) is nontrivial on infinitely many prime powers. The edge-monotonicity of
    \(\Phi_{\scat{}}\) follows from its definition and from the edge-monotonicity of \(\Phi\).
\end{proof}

Next, we show that if
the scattered property \(\Phi_{\scat{}}\) corresponding to \(\Phi\)
is computable, edge-monotone, and nontrivial on infinitely many prime powers,
then we can reduce \(\NUM{}\indsubsprob(\Phi)\) to  \(\NUM{}\indsubsprob(\Phi_{\scat{}})\).

\begin{corollary}\label{11.5-20}
    Let \(\Phi\) denote a computable, edge-monotone graph property
    and write
    \(\Phi_{\scat{}}\) for the scattered property corresponding to \(\Phi\).
    Then, there is a parameterized Turing reduction from $\NUM{}\indsubsprob(\Phi_{\scat{}})$
    to $\NUM{}\indsubsprob(\Phi)$.
\end{corollary}
\begin{proof}
    Write \(\scat{\Phi}\) for the set of all integers on which \(\Phi\) is scattered.
    As in \cref{11.5-2}, write \(q(\scat{\Phi}) \coloneqq \{ q(k)  \mid k \in \scat{\Phi}\}\) for the set
    of all maximal prime powers corresponding to \(\scat{\Phi}\).
    Next, for each \(m \in
    q(\scat{\Phi})\) write \(q^{-1}(m)\) for the minimum \(k \in \scat{\Phi}\) with \(q(k) =
    m\). For $m\in q(\scat{\Phi})$, define $H_m$ as in \cref{11.5-2}.

    Fix a graph \(G\) and an integer \(k\).
    We wish to compute \(\NUM{}\indsubs{(\Phi_{\scat{}}, k)}{G}\).
    First, we check if \(k\) is a prime power that is contained in \(q(\scat{\Phi})\).
    If we observe \(k \not\in q(\scat{\Phi})\), we return \(0\), as no graph with $k$
    vertices is in $\Phi_{\scat{}}^{\mathcal{H}}$. Otherwise, a $k$-vertex graph is in
    $\Phi_{\scat{}}$ if and only if it is in $(\Phi - H_k)$. Thus, we can return
    \(\NUM{}\indsubs{((\Phi-H_k), k)}{G}\)
    using \cref{lem:inc:exc} (which we
    supply with our oracle for \(\NUM{}\indsubsprob(\Phi)\)).

    Finally, we analyze the running time of the reduction. To that end, we assume that we have
    access to a \(\NUM{}\indsubsprob(\Phi)\) oracle, that is, we assume that the running time $g(k,
    |V(G)|)$ of the oracle is constant.
    Now, first observe that all computations in the reduction (other than the call
    to \cref{lem:inc:exc}) have a running time that depends only on the parameter~\(k\).
    Next,
    we observe that \(|V(H_k)| = q^{-1}(k) -
    k\). Hence, the call to \cref{lem:inc:exc} takes time
    \[
        O(2^{q^{-1}(k) - k} \cdot (q^{-1}(k) - k + |V(G)|)^2 \cdot g( q^{-1}(k), q^{-1}(k) - k + |V(G)| ))
        = O( g'(k) \cdot |V(G)|^2 ),
    \]
    for some computable function \(g'\).
    Lastly, we observe that the reduction of \cref{lem:inc:exc} uses only oracle calls where the
    parameter has size $(q^{-1}(k) - k) + k  = q^{-1}(k)$.

    Hence, there is a parameterized
    Turing reduction from $\NUM{}\indsubsprob( \Phi_{\scat{}} )$ to
    $\NUM{}\indsubsprob(\Phi)$.
\end{proof}

\begin{remark}\label{11.5-12}
    If the graph property $\Phi_{\scat{}}$ is edge-monotone
    and nontrivial infinitely often (that is, nontrivial on infinitely many prime
    powers), then \cref{theo:edge_mono:prime} applies and the problem
    $\NUM{}\indsubsprob(\Phi_{\scat{}})$ is \w-hard (and we also obtain an
    ETH-based lower bound).
\end{remark}

\subsection{\NUM{}W[1]-hardness and Quantitative Lower Bounds for Edge-monotone Properties}
\label{sec:main1prof}

We are finally ready to prove \cref{theo:edge_mono}.

\thmedgemon
\begin{proof}
    Write \(\scat{\Phi}\) for the set of all integers on which \(\Phi\) is scattered
    and write \(\con{\Phi}\) for the set of all integers on which \(\Phi\) is
    concentrated.

    For the \w-hardness part, let us consider two cases: either $\scat{\Phi}$ or
    $\con{\Phi}$ is infinite. Suppose first that $\scat{\Phi}$ is infinite. Then, the
    scattered property $\Phi_{\scat{}}$ corresponding to $\Phi$ is computable and
    nontrivial on infinitely many prime powers (\cref{11.5-1}). Hence,
    $\NUM{}\indsubsprob(\Phi_{\scat{}}$ is \w-hard (\cref{theo:edge_mono:prime}) and can be
    reduced to $\NUM{}\indsubsprob(\Phi)$ (\cref{lem:inc:exc}), showing that
    $\NUM{}\indsubsprob(\Phi)$ is also \w-hard.

    Assume now that $\con{\Phi}$ is infinite. Let $\mathcal{H}$ contain every graph $H$ whose
    alternating enumerator is nonvanishing for $\Phi$ and contains $K_{q(|V(H)|),q(|V(H)|)}$
    as subgraph. The set $\mathcal{H}$ is clearly infinite and computable.
    Observe that every $H\in \mathcal{H}$ has treewidth at least \(q(k)\).\footnote{Consult for
        instance \cite[Corollary~9~and~Lemma~4]{treewidth} for a proof of this folklore
    fact.}
    In particular, by \cref{lem:upper:bound:q}(3), as the family \(\mathcal{H}\)
    is infinite, it has unbounded treewidth. Thus \cref{lem:alpha:treewidth} yields
    \w-hardness for \(\NUM{}\indsubsprob(\Phi)\).

    We turn to the ETH-based lower bound next.
    To that end, fix a \(k\) for which \(\Phi\) is nontrivial.
    Now, \(\Phi\) is either scattered or concentrated on \(k\). We provide a proof for both cases.

    First, we consider the case that \(\Phi\) is scattered on \(k\). Define $m \coloneqq q(k)\in
    q(\scat{\Phi})$. Let $H_m$ denote the graph defined in \cref{11.5-2}.
    Then, \((\Phi - {H_m}) \coloneqq \{ G \mid G \UnionGraph H_m \in \Phi\}
    \) is edge-monotone and nontrivial on \(m=q(k) \ge c \log(k)
    \ge 3\) vertices (for some constant \(c\) and sufficiently large \(k\)).

    Now, by \cref{theo:sqrt:lower:bound} and assuming ETH, there is an \(\alpha\) such
    that
    there is no algorithm that for each graph \(G\) computes the number
    $\NUM{}\indsubs{((\Phi - {H_m}), m)}{G}$ in
    time \(O(|V(G)|^{ \alpha \sqrt{q(k)} / \ \log(q(k)) } )\).
    Using the reduction of
    \cref{lem:inc:exc} (which has a quadratic overhead),
    we obtain that there is no algorithm that computes for every graph \(G\) the
    number \(\NUM{}\indsubs{(\Phi, k)}{G}\) in time
    \[
        O(|V(G)|^{ (\alpha \sqrt{q(k)} / \ \log(q(k))) - 2 } ),
    \]
    We conclude this case by proving the following inequalities.
    \begin{claim}\label{claim:ineq:lower:bound:proof}
        Choosing \(\gamma' \coloneqq \alpha \sqrt{c} / 2\),
        for any sufficiently large $k$, we have
        \begin{align*}
            \gamma' \frac{\sqrt{\log k }}{ \log \log k}
            &\leq \gamma' \frac{\sqrt{q(k)/c}}{ \log(q(k) / c)}
            \leq \gamma' \frac{\sqrt{q(k)/c}}{\log(q(k))}
            \leq \alpha \frac{ \sqrt{q(k)}}{ \log(q(k))} - 2.
        \end{align*}
    \end{claim}
    \begin{claimproof}
        \begin{enumerate}
            \item For the first inequality
                we consider the function $h_1 \colon \mathbb{R}_{> \mathrm{1}} \to
                \mathbb{R}; x \mapsto \sqrt{x} / \log(x)$.
                Observe that \(h_1\) is monotonically
                increasing for $x \geq \mathrm{e^2}$.\footnote{%
                    First, one may verify that the derivative of \(h_1\) is
                    $h_1'(x) = (\log(x) - 2)/(2 \sqrt{x} \log^2(x))$.
                    Next, one verifies that $h_1'(x)
                    \geq 0$ for all $x \geq \mathrm{e}^2$.
                    Thus, $h_1$ is monotonically increasing after \(\mathrm{e}^2\).}
                Finally, we choose \(k\) large enough to satisfy
                \(\mathrm{e}^2  \leq \log k  \leq q(k)/c.\)
            \item The second inequality is immediate.
            \item The last inequality is equivalent to $4 \leq \alpha
                \sqrt{q(k)} / \log(q(k))$,
                which holds for sufficiently large $q(k)$ since $h_1$ is unbounded and monotonically increasing.
                \claimqedhere
        \end{enumerate}
    \end{claimproof}

    Now, by \cref{claim:ineq:lower:bound:proof}, for a \(k \ge N_1\) (where \(N_1\) is a
    sufficiently large constant) and
    assuming ETH,
    there is no algorithm that for each graph \(G\) computes the number
    $\NUM{}\indsubs{(\Phi, k)}{G}$ in
    time
    \[
        O(|V(G)|^{ (\alpha \sqrt{q(k)} / \ \log(q(k))) - 2 } )
        \supseteq
        O(|V(G)|^{ (\gamma' \sqrt{q(k)} / \ \log(q(k))) } ).
    \]

    Next, we consider the case that \(\Phi\) is concentrated on \(k\).
    By definition, there exists a graph $H_k$ on $k$ vertices with a nonvanishing
    alternating enumerator that contains $K_{q(k), q(k)}$ as a subgraph.
    As before, this implies that $H_k$
    has a treewidth of at least $q(k)$.

    Now, from \cref{lem:upper:bound:q}(3), for each integer \(N_2\),
    we obtain that for some constant \(c\) and
    for \(k \ge \beta\) (for some constant
    \(\beta \coloneqq \beta(c, N_2)\) that depends on \(c\) and \(N_2\)), we have \[
        \tw(H_k) \ge q(k) \geq c \log(k) \ge N_2.
    \]

    Now, by \cref{lem:alpha:treewidth}, assuming ETH, and choosing \(N_2 \coloneqq 3\),
    there is no algorithm that for each graph \(G\) computes the number
    $\NUM{}\indsubs{(\Phi, k)}{G}$ in
    time \[
        O(|V(G)|^{\alpha_{\indsubsprob} \tw(H_k) / \log\tw(H_k)})
        \supseteq O(|V(G)|^{\alpha_{\indsubsprob} c \log(k) / \log (c \log(k))})
        \supseteq O(|V(G)|^{\alpha_{\indsubsprob} c  \log k / \log\log k })
    .\]
    For the first step, we use that the function $h_2 \colon \mathbb{R}_{>1} \to
    \mathbb{R}, x \mapsto x / \log(x)$ is monotonically increasing for $x > \mathrm{e}$.
    For the second step, we use $0 < c < 1$ (which we can assume without loss of
    generality).
    Finally, to obtain the claim also for all values \(k\) that are less than
    \(N_0 \coloneqq \max(N_1, \beta(c, N_2))\), we choose
    \(\gamma \coloneqq \min( 1/\sqrt{\log(N_0)}, \gamma', c \alpha_{\indsubsprob})\).
    Observe that now, for \(k < N_0\), we obtain that
    \[
        O(|V(G)|^{\gamma \sqrt{\log k} / \log\log k }) = o( |V(G)| ).
     \]
    Now, such a running time is unconditionally unachievable for any algorithm that reads
    the whole input.
    This completes the proof.
\end{proof}

\section{Main Result 2: Tight Bounds for Edge-monotone Properties on Prime Powers}\label{sec:lower:bounds}

In this section, we prove \cref{theo:tight:lower:bound}.
While we build on our structural understanding of fixed points from \cref{sec:alt:en:hasse},
the proof is independent of \cref{sec:difference:graphs,sec:mainresult1}.

\thetightlowerbound*

Recall that by \cref{sec:difference:graphs}, the rotation subgroup $\rotgr{p^m}\subseteq
\aut(K_{p^m})$ gives rise to the difference graphs over $\field{p^m}$ as fixed points.
In this section, we consider a much larger subgroup of $\aut(K_{p^m})$, namely the Sylow
$p$-subgroup.
Thereby, we obtain fixed points of a different type, the lexicographic product of
difference graphs over $\field{p}$.
To analyze such fixed points, for the rest of the section, we view the vertex set of $K_{p^m}$ as
$\fragmentco{0}{p}^m$.

\subsection{The Fixed Points of Sylow Groups on \texorpdfstring{$K_{p^m}$}{K p m}}\label{sec:fixed:point:p^k}

We intend to construct special subgroups of \(\auts{K_{p^m}}\).
To that end, we first define special types of bijections on \(\fragmentco{0}{p}^{m}\).

\begin{definition}\label{11.12-3}
    Consider a prime $p$ and a positive integer $m$.
    For each \(j \in \fragmentco{0}{m}\), write \(\varphi_i\) for a function
    \(\fragmentco{0}{p}^{j} \to \fragmentco{0}{p}\) and set \(\varphi \coloneqq
    (\varphi_0,\dots,\varphi_{m-1})\).
    We define the function $\sylelm : \fragmentco{0}{p}^{m} \to \fragmentco{0}{p}^{m}$ via\footnote{We write $\varphi_0$ for
    $\varphi_0(())$ since $\varphi_0$ is a function that is defined on a single element}
    \[\sylelm(x_1, \dots, x_m) \coloneqq (x_1 + \varphi_0, x_2 + \varphi_1(x_1), x_3 + \varphi_2(x_1,
    x_2), \dots, x_m + \varphi_{m-1}(x_1, \dots, x_{m-1})), \] where all computations are done
    modulo $p$.

    We write \(\overline{p}^m\) for the set of all functions $\sylelm :
    \fragmentco{0}{p}^{m} \to \fragmentco{0}{p}^{m}$ that are obtained in this fashion,
    that is,
    \[
        \overline{p}^m \coloneqq \{ \sylelm \mid \varphi_i \in \fragmentco{0}{p}^{j} \to
        \fragmentco{0}{p} \}.
        \qedhere
    \]
\end{definition}

\begin{lemma}\label{11.12-2}
    For each \(j \in \fragmentco{0}{m}\), write \(\varphi_i\) for a function
    \(\fragmentco{0}{p}^{j} \to \fragmentco{0}{p}\) and set \(\sylelm \coloneqq
    (\varphi_0,\dots,\varphi_{m-1})\).
    Then, the function $\sylelm$ is a bijection on $\fragmentco{0}{p}^{m}$.
\end{lemma}
\begin{proof}
    Consider two different tuples \(x \coloneqq (x_1,\dots, x_m)\) and \(x' \coloneqq (x'_1,\dots, x'_m)\) and
    write \(i\) for the minimal position with \(x_i \neq x'_i\).
    We claim that \( \sylelm(x)_i \neq  \sylelm(x')_i\).
    To that end, observe that \(\varphi_{i - 1}\) depends only on the values \(x_1 = x'_1,
    \dots, x_{i-1} = x'_{i-1}\). Thus, we have \(\varphi_i(x_1,\dots,x_{i-1}) =
    \varphi_i(x'_1,\dots,x'_{i-1})\).
    Hence, we have
    \begin{align*}
        \sylelm(x)_i &= x_i + \varphi_i(x_1,\dots,x_{i-1}) \neq x'_i + \varphi_i(x_1,\dots,x_{i-1})
        = x'_i + \varphi_i(x'_1,\dots,x'_{i-1})
        = \sylelm(x')_i.
    \end{align*}
    Hence, \(\sylelm\) is injective, and as a mapping between equal-sized sets thus also
    bijective.
\end{proof}

Recall that $\fragmentco{0}{p}^{m}$ is also the vertex set of $K_{p^m}$.
Thus, \cref{11.12-2} implies $\sylelm \in \aut(K_{p^m}) \cong \sym{p^m}$.

\begin{remark}\label{11.13-4}
    Let us discuss and interpret bijections $\sylelm \in \overline{p}^m$.
    To that end, consider an ordered, complete \(p\)-tree \(T_{p^m}\) of height \(m\).
    Label the root with the empty tuple \(()\). Further, for each vertex with label
    \((v_1,\dots,v_j)\), label its \(i\)-th child with \((v_1,\dots,v_j, i)\) (where we
    number children 0-indexed).
    Observe that \(T_{p^m}\) has \(p^m\) leaves that correspond to the elements of
    $\fragmentco{0}{p}^{m}$.
    Consult  \cref{fig:sylow:a} for a visualization of an example.

    Now, $\sylelm$ acts on $T_{p^m}$ in the following way.
    Starting with the root, each level \(i\) of $T_{p^m}$ is rotated
    by $\varphi_i(v_1, \dots, v_i)$ nodes,
    that is, the $j$-th child of a vertex \((v_1,\dots,v_i)\)
    becomes the $((i + \varphi_i(v_1, \dots, v_i)) \bmod p)$-th child of $(v_1,\dots,v_i)$.
    Consult \cref{fig:sylow:b} for a visualization of an example.
\end{remark}

\begin{figure}[tp]
    \begin{subfigure}[t]{.48\textwidth}
    \centering
    \includegraphics[scale = 1.5]{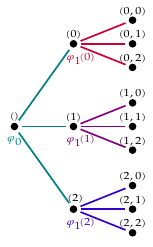}
    \caption{The tree \(T_{3^2}\) from \cref{11.13-4}.
        Each vertex is label with a label that is depicted above the corresponding vertex.
        Further, we depict below each vertex \(v\) the function (value) in \(\sylelm\)
        that is responsible for rotating the subtree rooted at \(v\).
    }
    \label{fig:sylow:a}
    \end{subfigure}\quad
    \begin{subfigure}[t]{.48\textwidth}
    \centering
    \includegraphics[scale = 1.5]{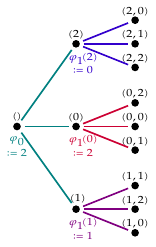}
    \caption{The image \(\sylelm T_{3^2}\), where \(\varphi_0 \coloneqq 2\) and
        \(\varphi_1 \coloneqq \{ 0 \mapsto 2, 1 \mapsto 1, 2 \mapsto 0 \} \).
        Reading of the labels of the leaves from top to bottom,
        we obtain
        the corresponding permutation of $\fragmentco{0}{p}^{m}$.}
    \label{fig:sylow:b}
    \end{subfigure}
    \caption{\Cref{11.12-3,11.13-4} visualized. }
    \label{fig:sylow}
\end{figure}

Next, we show that \(\overline{p}^m\) forms a group with the function
composition as the group operation.

\begin{lemma}\label{11.13-1}
    The pair $\sylow_{p^m} \coloneqq (\overline{p}^m, \circ)$ forms a $p$-group.
\end{lemma}
\begin{proof}
    We first show that $\sylow_{p^m}$ defines a group.

    \begin{claim}\label{11.12-5}
        The composition of elements of \(\overline{p}^m\) is an element of
        \(\overline{p}^m\).
    \end{claim}
    \begin{claimproof}
        Fix $\sylelm, \overline{\psi} \in \overline{p}^m$.
        Expanding the definition yields
        \begin{align*}
            \overline{\psi}(\sylelm(x_1, \dots, x_m))
            &= (x_1 + \varphi_0 + \psi_0,
                x_2 + \varphi_1(x_1) + \psi_1((\varphi_0)(x_1)),
                \dots, \\
            &\quad\;\; x_m + \varphi_{m-1}(x_1, \dots, x_{m-1})
            + \psi_m((\varphi_0, \dots, \varphi_{m-1})(x_1, \dots, x_{m-1}))) \\
            &= \overline{\lambda}(x_1, \dots, x_m),
        \end{align*}
        where
        $\lambda_0 \coloneqq  \varphi_0 + \psi_0$ and
        $\lambda_j(x_1, \dots x_{j}) \coloneqq
        \varphi_{j}(x_1, \dots, x_{j})
        + \psi_j((\varphi_0, \dots, \varphi_{j-1})(x_1, \dots, x_{j-1}))$.
        Thus, indeed $\overline{\lambda} \in \sylow_{p^m}$, which yields the claim.
    \end{claimproof}
    \begin{claim}\label{11.12-6}
        Every element of \(\sylelm \in \overline{p}^m\) has an inverse \(\sylelm^{\,-1} \in \overline{p}^m\).
    \end{claim}
    \begin{claimproof}
        We define $\sylelm^{\,-1}$ via $\varphi^{-1}_0 \coloneqq - \varphi_0$ and \[
            \varphi^{-1}_j(x_1, \dots, x_{j})
            \coloneqq - \varphi_j((\varphi^{-1}_0,
            \dots, \varphi^{-1}_{j-1})(x_1, \dots, x_{j-1}))
            \text{ for } j \in \fragmentco{1}{m}.
    \]
        An induction on $j \in \fragmentco{0}{m}$ readily yields
        $\sylelm^{\,-1} \circ \sylelm = \sylelm \circ \sylelm^{\,-1} = \id$; which
        completes the proof.
    \end{claimproof}

    By \cref{11.12-5,11.12-6}, $\sylow_{p^m}$ is indeed a group.

    Finally, we show that $\sylow_{p^m}$ is a $p$-group.
    To that end, observe that each $m$-tuple $(\varphi_0, \dots, \varphi_{m-1})$
    with $\varphi_j \colon \fragmentco{0}{p}^{j} \to \fragmentco{0}{p}$ defines a
    different group element \(\sylelm\).
    The number of functions $\varphi_j$ from $\fragmentco{0}{p}^{j}$ to $\fragmentco{0}{p}$ is
    equal to $p^{p^j}$, thus we obtain $|\sylow_{p^m}| = |\overline{p}^m|
    = p^1 \cdot p^{p^1} \cdot \dots \cdot p^{p^{m-1}}$; which completes the proof.
\end{proof}

\begin{remark}
    Let us briefly discuss how $\sylow_{p^m}$ relates to Sylow $p$-subgroups that appear
    in the literature.
    Classically, a {\em Sylow $p$-subgroup} of $\Gamma$ is a $p$-subgroup that is
    not a proper subgroup of any other $p$-subgroup of $\Gamma$.
    Due to the Sylow~Theorems (consult for instance \cite[Theorem
    4.12]{group_theory_rotman}), there is only one Sylow
    $p$-subgroup in $\sym{p^m}$ (up to isomorphism).
    Usually, the Sylow \(p\)-subgroup is
    constructed by taking $m$ copies of the cyclic group $\Z_p$ and using the wreath
    product \(\Z_p \wr \dots \wr \Z_p\)~\cite[Theorem 7.27]{group_theory_rotman}.
    Indeed, one may prove that the group $\sylow_{p^m}$ defined via \cref{11.13-1} is
     isomorphic to the Sylow $p$-subgroup of $\sym{p^m}$, that is, we have
    $\sylow_{p^m} \cong \Z_p \wr \dots \wr \Z_p$.
\end{remark}

Next, we analyze the fixed point structure of $\fp(\sylow_{p^m}, K_{p^m})$.
To that end, we introduce the lexicographic product of graphs.

\deflexgraph
As an easy example, observe that we have  $G_1 \wrGraph G_2 \cong \metaGraph{G_1}{G_2,
\dots, G_2}$ and $G_1 \wrGraph G_2    \wrGraph G_3 \cong \metaGraph{G_1}{G_2 \wrGraph G_3,
\dots, G_2 \wrGraph G_3}$.

\begin{remark}
    Another name for the lexicographic product is \emph{wreath product of
    graphs}~\cite{Automorphism_Groups_Ciculant_semigroup_theory,Automorphism_Groups_Wreath_Product}.
    The name wreath product emphasizes the close connection
    between the group theoretical wreath products of automorphism subgroups $A
    \subseteq \aut(G)$ and $B \subseteq \aut(H)$, and the wreath products of \(G\) and~\(H\).
    Further, there is a close connection between the group-theoretical wreath product and
    the wreath product of groups when considering fixed points.
    However, as our proofs do not use the terminology of wreath product of groups,
    we chose the more common name \emph{lexicographic product}.
\end{remark}

Let us take a closer look at the lexicographic product of difference graphs.

\begin{lemma}\label{rem:sylow:graphs}
    Write $p$ for a prime and \(m\) for a positive integer.
    Further, for each \(i \in \setn{m}\), write \(A_i \subseteq \field{p}^+\) for a
    subset.
    Then, $E(\cgr{p}{A_1} \wrGraph \cdots \wrGraph \cgr{p}{A_m})$ is equal to
    \begin{align*}
        \{ \{(a_1, \dots a_m), (b_1, \dots b_m)\}
             \mid \exists i \in \setn{m} \colon (\forall j < i
        \colon a_j = b_j) \,\text{ and }\, a_i - b_i \in A_i\cup (-A_i)  \}
    \end{align*}
\end{lemma}
\begin{proof}
    Unfolding the definition of the lexicographic product, we observe that
    $\{(a_1,
    \dots, a_m), (b_1, \dots, b_m)\}$ is an edge of $\cgr{p}{A_1} \wrGraph \dots \wrGraph
    \cgr{p}{A_m}$ if we can find a $i$ such that $a_{j} = b_{j}$ for $j < i$ and
    $\{a_i, b_i\} \in \cgr{p}{A_i}$.
    Finally, we recall that $\{a_i, b_i\} \in \cgr{p}{A_i}$ is equivalent to $a_i - b_i \in A_i \cup (-A_i)$.
\end{proof}
Next, we use the lexicographic product of difference graphs
to understand the fixed points $\fpb{\sylow_{p^m}}{K_{p^m}}$.

\lemfixedpointslexprod
\begin{proof}
    Write $C \coloneqq  \{\cgr{p}{A_1} \wrGraph \cdots \wrGraph
    \cgr{p}{A_m} : A_i \subseteq \field{p}^+\}$.

    \subparagraph*{\normalfont $C \subseteq \fpb{\sylow_{p^m}}{K_{p^m}}$.}
    Fix sets \(A_i\) and the corresponding graph
    $H \coloneqq \cgr{p}{A_1} \wrGraph \cdots \wrGraph \cgr{p}{A_m}$.
    We prove that \(H\) is a fixed point of $\sylow_{p^m}$
    by verifying that $\sylelm(H) = H$ holds for
    every function $\sylelm \in \sylow_{p^m}$.

    To that end, fix a $\sylelm \in \sylow_{p^m}$.
    Further, fix an edge  $\{(a_1, \dots, a_m), (b_1, \dots, b_m)\} \in E(H)$. We show
    that $\{\sylelm(a_1, \dots, a_m), \sylelm(b_1, \dots, b_m)\} \in E(H)$.
    From \cref{rem:sylow:graphs} we obtain an $i \in \setn{m}$ such that $a_j =
    b_j$ for all $j < i$ and $a_i - b_i \in A_i \cup (-A_i)$.
    Without loss of generality, we assume that $a_i - b_i \in A_i$ (otherwise switch the
    roles of $a$ and $b$). Observe that we need only the values
    $a_1, \dots, a_j$ to compute the $j$-th coordinate of $\sylelm(a)$ (the same is true
    for the $j$-th coordinate of $\sylelm(b)$).
    This implies $\sylelm(a)_j = \sylelm(b)_j$
    for all $j < i$. Further, we obtain
    \begin{align*}
        \sylelm(a)_i -  \sylelm(b)_i &= a_i + (\varphi_0, \dots, \varphi_{i-1})(a_1,
        \dots, a_{i-1}) - (b_i + (\varphi_0, \dots, \varphi_{i-1})(b_1, \dots, b_{i-1}))
        \\
        &= a_i + (\varphi_0, \dots, \varphi_{i-1})(a_1, \dots, a_{i-1}) - (b_i +
        (\varphi_0, \dots, \varphi_{i-1})(a_1, \dots, a_{i-1})) \\
        &= a_i - b_i \in A_i,
    \end{align*}
    which implies that $\{\sylelm(a_1, \dots, a_m), \sylelm(b_1, \dots, b_m)\}$ is an edge in $H$.
    Thus, $\sylelm$ maps edges of $H$ to edges of $H$ and is therefore an automorphism.

    \subparagraph*{\normalfont $\fpb{\sylow_{p^m}}{K_{p^m}} \subseteq C$.}
    We prove the contrapositive, that is, we prove that each graph $H$ with vertex
    set $V(K_{p^m})$ that is not in $C$ is also not in $\fp(\sylow_{p^m}, K_{p^m})$.

    To that end, we construct a $\sylelm \in \sylow_{p^m}$ with $\sylelm(H) \neq H$.
    If $H \notin C$, then, due to \cref{rem:sylow:graphs}, there are an
    $i \in \setn{m}$ and an $x \neq 0$ such that
    \begin{itemize}
        \item $u \coloneqq \{(a_1, \dots, a_{i-1}, a_i, b_{i+1}, \dots, b_m), (a_1, \dots,
            a_{i-1}, a_i + x, c_{i+1}, \dots, c_{m})\}$ is an edge in $H$ but
        \item $v \coloneqq \{(a_1, \dots, a_{i-1}, d_i, \beta_{i+1}, \dots, \beta_m),
                (a_1, \dots, a_{i-1}, d_i + x, \gamma_{i+1}, \dots, \gamma_{m})\}$ is not an edge in $H$.
    \end{itemize}
    Next, we consider the group element $\sylelm \coloneqq (\varphi_0, \dots \varphi_{m-1})$ with
    \[
        \varphi_{m-1}(x_1, \dots x_{m-1}) \coloneqq
        \begin{cases}
            \beta_m - b_m  &\text{if } (x_1, \dots x_{m-1}) =  (a_1, \dots, a_i, b_{i+1}, \dots, b_{m-1}) \\
            \gamma_m - c_m &\text{if }  (x_1, \dots x_  {m-1}) = (a_1, \dots, a_i + x, c_{i+1}, \dots, c_{m-1})  \\
            0 &\text{otherwise}
        \end{cases}
    \]
    and $\varphi_j \coloneqq \id$ for $j \neq m-1$.
    We obtain
    \[
        \sylelm(u) = \{(a_1, \dots, a_i, b_{i+1}, \dots, b_{m-1}, \beta_m),
            (a_1, \dots, a_i + x, c_{i+1}, \dots, c_{m-1}, \gamma_m).
    \]
    In particular, the last coordinate of the
    edge $\sylelm(u)$ is equal to the last coordinate of the edge $v$.
    By iterating this construction, we obtain an $\overline{\psi} \in \sylow_{p^m}$
    with
    \[
        \overline{\psi}(u) = \{(a_1, \dots, a_i, \beta_{i+1}, \dots, \beta_m),
        (a_1, \dots, a_i + x, \gamma_{i+1}, \dots, \gamma_{m})\}.
    \] Lastly, we define
    $\Tilde{\psi} \coloneqq (\Tilde{\psi}_0, \dots, \Tilde{\psi}_{m-1})$ with
    \[
        \Tilde{\psi}_{i}(x_1, \dots x_{i-1}) \coloneqq
        \begin{cases}
            d_i - a_i &\text{if } (x_1, \dots x_{i-1}) =  (a_1, \dots, a_{i-1}) \\
            0 &\text{otherwise}
        \end{cases}
    \]
    and $\Tilde{\psi}_j \coloneqq \id$ for $j \neq i$.
    Observe that $\Tilde{\psi} (\overline{\psi}(u)) = v$.
    Thus there is an element in $\sylow_{p^m}$ that maps the edge $u$ to the non-edge $v$
    which shows that $H$ is not in $\fpb{\sylow_{p^m}}{K_{p^m}}$.
\end{proof}

Lastly, we compute the level of fixed points in \(\fpb{\sylow_{p^m}}{K_{p^m}}\).
\begin{lemma}\label{lem:level:sylow:graphs}
    For any prime \(p\) and any positive integer \(m\),
    the level of $\cgr{p}{A_1} \wrGraph \dots
    \wrGraph \cgr{p}{A_m} \in \fpb{\sylow_{p^m}}{K_{p^m}}$ is
    \[\Hasselevel(\cgr{p}{A_1} \wrGraph \dots
    \wrGraph \cgr{p}{A_m}) = \sum_{i = 1}^m |A_i|.\]
\end{lemma}
\begin{proof}
    For an \(x  \in \field{p}^+\), consider the fixed point
    $F \coloneqq \cgr{p} {A_1} \wrGraph \cdots \wrGraph \cgr{p}{A_j}
    \wrGraph \cdots \wrGraph \cgr{p}{A_m}$,
    where $A_j \coloneqq \{x\}$ and $A_i \coloneqq \emptyset$ otherwise.
    \begin{claim}
        We have \(\hl{F} = 1\).
    \end{claim}
    \begin{claimproof}
        Recall from \cref{def:10.4-1} that the level of $F$ is defined as the size of
        the orbit factorization of $F$.
        As $F$ has edges, the orbit factorization of \(F\) is not empty.
        Thus, the level of $F$ is at least $1$.

        Next, suppose that $\hl{F} > 1$.
        Then, the orbit factorization of \(F\) contains at least two orbits.
        This implies that there is a proper sub-point $F' \subsetneq F$
        that contains at least one orbit and is therefore not the empty graph.
        However, if $F'$ is a proper sub-point of $F$
        then
        \[
            F' = \cgr{p} {B_1} \wrGraph \cdots \wrGraph \cgr{p}{B_j} \wrGraph \cdots
            \wrGraph \cgr{p}{B_m},
        \] with $\emptyset \subsetneq  B_j \subsetneq A_j = \{x\}$ and $B_i \subseteq A_i
        = \emptyset$; which is a contradiction.
        Hence, we have \(\hl{F} \le 1\) which yields the claim.
    \end{claimproof}

    Next, observe that for all $x \neq y \in \field{p}^+$,
    the conditions
    \begin{center}
        $a_j - b_j \in \{x, -x\} \qquad a_j - b_j \in \{y, - y\}$
    \end{center}
    are mutually exclusive.
    Hence,
    for a fixed point $F' \coloneqq \cgr{p} {\emptyset} \wrGraph \cdots \wrGraph \cgr{p}{\emptyset} \wrGraph
    \cgr{p}{A_j} \wrGraph \cgr{p}{\emptyset}
    \wrGraph \cdots \wrGraph \cgr{p}{\emptyset}$ with $A_j \subseteq \field{p^m}^+$, we
    obtain the following disjoint union
    \begin{align*}
        E(F) &= \{\{(a_1, \dots a_m), (b_1, \dots b_m)\}  \mid  (\forall i < j
            \colon a_i = b_i) \text{ and } a_j - b_j \in A_j \cup (-A_j)  \} \\
             &= \bigcup_{x \in A_j} \{\{(a_1, \dots a_m),
             (b_1, \dots b_m)\} :  \forall i < j
             \colon a_i = b_i) \text{ and }  a_j - b_j \in \{x, -x\} \} \\
             &=  \bigcup_{x \in A_j}  E(\cgr{p}
             {\emptyset} \wrGraph \cdots \wrGraph \cgr{p}{\emptyset} \wrGraph
             \cgr{p}{\{x\}} \wrGraph \cgr{p}{\emptyset}
             \wrGraph \cdots \wrGraph \cgr{p}{\emptyset}).
    \end{align*}
    Hence, we have $\Hasselevel(F') = |A_j|$.

    Lastly, for a general fixed point we obtain
    \[
        E(\cgr{p}{A_1} \wrGraph \cdots \wrGraph \cgr{p}{A_j}
        \wrGraph \cdots \wrGraph \cgr{p}{A_m}) = \bigcup_{j \in \setn{m}} E(\cgr{p}
            {\emptyset} \wrGraph \cdots \wrGraph \cgr{p}{\emptyset} \wrGraph
            \cgr{p}{A_j} \wrGraph \cgr{p}{\emptyset}
        \wrGraph \cdots \wrGraph \cgr{p}{\emptyset}).
    \]
    Thus, we have $\Hasselevel(\cgr{p}{A_1} \wrGraph \cdots \wrGraph \cgr{p}{A_j}
    \wrGraph \cdots \wrGraph \cgr{p}{A_m}) = \sum_{i = 1}^m |A_i|$; which completes
    the proof.
\end{proof}

Let us briefly discuss why the fixed points \(\fpb{\sylow_{p^m}}{K_{p^m}}\) are useful to
us.
Recall that our goal is the find nonvanishing fixed points with large treewidth.
Achieving this goal is easier if many fixed points have a large treewidth.
Fixed points have a large treewidth if they consist of orbits that have a large treewidth.
Typically, orbits (and graphs in general) have a large treewidth if they have many edges.
Since each orbit has the form
$\sylow_{p^m} \cdot \{i, j\} \coloneqq \{\{\sylelm(i), \sylelm(j)\} \mid \sylelm \in \sylow_{p^m}\}$, for
some edge $\{i, j\} \in E(K_{p^m})$, we observe that larger groups lead to larger orbits.
Now, we additionally need that our group is a $p$-group.
As the $p$-Sylow group $\sylow_{p^m}$ is by definition the largest $p$ group in
$\sym{p^m}$, the choice of the group $\sylow_{p^m}$ is in some sense optimal.

More concretely, we may compare the number of orbits $E(K_{p^m}) / \rotgr{p^m}$ and the
number of orbits $E(K_{p^m}) / \sylow_{p^m}$.
As orbits form a partition of the edge set of $K_{p^m}$, having fewer orbits
leads to orbits that have more edges on average.
Now, recall that by \cref{lem:orbits:Zpa}, $E(K_{p^m}) / \rotgr{p^m}$ contains
$|\field{p^m}^+|$ orbits, which is equal to $(p^m - 1) / 2$ if $p$ is odd.
Next, \cref{lem:level:sylow:graphs} shows that $E(K_{p^m}) / \sylow_{p^m}$ contains
$m \cdot |\field{p}^+|$ orbits,
which is equal to $m \cdot (p -
1)/2$ if $p$ is odd. This is much smaller than  $(p^m - 1) / 2$.

\subsection{Prime Powers Contain Large Bicliques}\label{sec:edge-monotone:prime:powers}

Next, we prove \cref{theo:edge_mono:prime:power:bicliques} by analyzing the
fixed points \(\fpb{\sylow_{p^m}}{K_{p^m}}\).
To that end, one important observation is that a fixed point
$\cgr{p}{A_1} \wrGraph \dots \wrGraph \cgr{p}{A_m}$ has a large biclique if there is a
small position $i$ with $A_i \neq \emptyset$.
This implies that there is an edge between all vertices
$(a_1, \dots, a_{i-1}, a_i, a_{i+1}, \dots a_m)$ and $(a_1, \dots, a_{i-1}, b_i,
b_{i+1}, \dots, b_m)$ as long as $a_i - b_i \in A_i \cup (-A_i)$.
As we are free in our choice of
$a_{i+1}, \dots, a_m, b_{i+1}, \dots, b_m \in \fragmentco{0}{p}$, this leads to a large biclique.
Thus, it is useful to keep track of this value $i$.

\begin{definition}
    Write $H \coloneqq \cgr{p}{A_1} \wrGraph \dots \wrGraph \cgr{p}{A_m}$ for a fixed
    point of \(\fpb{\sylow_{p^m}}{K_{p^m}}\).
    The \emph{empty-prefix} of \(H\)
    is the smallest index \(i\) with $A_i \neq \emptyset$, minus one;
    we write $\wrLevel(A_1, \dots, A_m) \coloneqq i-1$.
\end{definition}

As discussed earlier, we use the empty-prefix of a fixed point to find large a biclique as a
subgraph.

\lemwreathlevel
\begin{proof}
    Set $t \coloneqq \wrLevel(A) + 1$. We have $A_{t} \neq \emptyset$.
    Hence, there are $x \in A_{t}$ and $\alpha, \beta \in
    \fragmentco{0}{p}$ with $\alpha - \beta = x$.
    Now, we obtain that all vertices of
    the form $(a_1, \dots a_{{t}-1}, \alpha, b_{t+1}, \dots, b_{m})$ are
    connected to the vertices $(a_1, \dots a_{t-1}, \beta, c_{t+1}, \dots,
    c_{m})$ since they coincide on the first $t-1$ elements and the $t$-th
    element differs by $x$.
    In particular, as we can freely choose
    $b_{t+1}, \dots, b_{m}$ and $c_{t+1}, \dots, c_{m}$,
    we obtain a complete bipartite subgraph where one side
    contains all vertices of the from $(a_1, \dots a_{t-1}, \alpha, b_{t+1},
    \dots, b_{m})$ and the other side contains all vertices of the from $(a_1,
    \dots a_{t-1}, \beta, c_{t+1}, \dots, c_{m})$.
    Both sides contain
    $p^{{m} - (t + 1)  + 1} = p^{m - t} = p^{m - 1 - \wrLevel(A)}$ vertices.
    Thus,
    $K_{p^{m - 1 - \wrLevel(A)}, p^{m - 1 - \wrLevel(A)}}$
    is a subgraph of $\cgr{p}{A_1} \wrGraph \cdots \wrGraph \cgr{p}{A_m}$.
\end{proof}

Imagine that our graph property is nontrivial on $p^m$ for $m \geq 2$.
If we find a fixed point $H = \cgr{p}{A_1} \wrGraph \cdots \wrGraph
\cgr{p}{A_m}$ with a minimal empty-prefix of $\wrLevel(A_1, \dots, A_m) =
0$, then we know that the treewidth of $H$ is at least $p^{m-1}$.
Thus, our goal is to find a fixed point $H$ with $\wrLevel(A_1, \dots, A_m) = 0$
and $\ae{\Phi}{H} \not \equiv_p 0$.

To this end, we prove that a fixed point with a high empty-prefix is
always isomorphic to an edge-subgraph of a fixed point with a low
empty-prefix.
This in turn allows us to consider only fixed points with low
empty-prefix.
To be more precise, we show that each fixed point
$\cgr{p}{\emptyset} \wrGraph \cdots \cgr{p}{\emptyset} \wrGraph
\cgr{p}{A_1} \wrGraph \cdots \wrGraph \cgr{p}{A_{m-j}}$ is isomorphic to an
edge-subgraph of $\cgr{p}{A_1} \wrGraph \cdots \wrGraph \cgr{p}{A_m}$.

For a formal proof, we define the {forward revolution}.
\begin{definition}\label{pushforward}
    For a prime \(p\) and a positive integer \(m\), we define the \emph{forward revolution} of
    \(K_{p^m}\) as
    \begin{align}
        \push_{p^m} \colon V(K_{p^m}) \to V(K_{p^m});
        \quad (a_1 \dots, a_m) \mapsto (a_m, a_1, \dots, a_{m-1}).
        \tag*{\qedhere}
    \end{align}
\end{definition}

Using \cref{pushforward} allows us to prove the following lemma.

\lemwreathsubiso
\begin{proof}
    We only show that $\widetilde{H} \coloneqq \cgr{p}{\emptyset} \wrGraph \cgr{p}{A_1}
    \wrGraph \cdots \wrGraph \cgr{p}{A_{m-1}}$ is isomorphic to an edge-subgraph of
    $H \coloneqq \cgr{p}{A_1} \wrGraph \cdots \wrGraph \cgr{p}{A_{m}}$.
    Iterating then yields the general result.

    We show that $\widetilde{H}$ is an edge-subgraph of $\push_{p^m}(H)$ by showing
    that $\{a, b\} \in E(\widetilde{H})$ implies $\{a, b\} \in E(\push_{p^m}(H))$.
    This proves the claim
    as $\push_{p^m}$ is a bijective function and hence,
    $\push_{p^m}(H)$ is isomorphic to $H$.

    For all $\{a, b\} \in E(\widetilde{H})$, we obtain that
    there is an $i \in \{2, \dots, m\}$ with $a_j = b_j$ for all $j < i$ and
    $a_{i} - b_{i} \in A_{i-1} \cup (-A_{i-1})$.
    Without loss of generality, we assume that $a_{i} - b_{i} \in A_{i-1}$ (otherwise we
    switch the roles of $a$ and $b$).

    Next, we show that $\{\widetilde{a}, \widetilde{b}\} \in E(H)$, where
    $\widetilde{a} = (a_2, \dots, a_m, a_1)$ and $\widetilde{b} = (b_2, \dots, b_m,
    b_1)$.
    Observe that for $s \coloneqq i-1 \in \setn{m}$, we obtain $\widetilde{a}_j = a_{j+1} =  b_{j +
    1} = \widetilde{b}_j$ for all $j +1 < i$ which is equivalent to $j < s$.
    Further, we obtain $\widetilde{a}_s - \widetilde{b}_s \equiv a_{i} - b_{i} \in A_s$ which
    proves that $\{\widetilde{a}, \widetilde{b}\} \in E(H)$.
    This implies
    $\{\push_{p^m}(\widetilde{a}), \push_{p^m}(\widetilde{b})\} \in E(\push_{p^m}(H))$.
    Finally,
    we observe that $\push_{p^m}(\widetilde{a}) = a$ and $\push_{p^m}(\widetilde{b}) =
    b$ which shows that $\{a, b\}$ is an edge of $\push_{p^m}(H)$. This completes the proof.
\end{proof}

We are now ready to prove \cref{theo:edge_mono:prime:power:bicliques}.

\thmedgemonprimepowerbicliques
\begin{proof}
    Our goal is to show that there is a fixed point $H$ with the following
    properties

    \begin{itemize}
        \item $H = \cgr{p}{A_1} \wrGraph \cdots \wrGraph  \cgr{p}{A_{m}}$ for
            $A_i \subseteq \field{p}^+$,
        \item $\Phi(H) = 0$,
        \item $\wrLevel(A_1, \dots, A_{m}) = 0$, and
        \item $\Phi(\Tilde{H}) = 1$ for all proper sub-points $\Tilde{H}$ of $H$.
    \end{itemize}

    As \(\Phi\) is nontrivial on \(p^m\), we have $\Phi(K_{p^m}) = 0$
    and \(\Phi(\IS_{p^m}) = \Phi(\cgr{p}{\emptyset} \wrGraph \cdots \wrGraph \cgr{p}{\emptyset}) = 1\).
    Now, let $i$ denote the smallest value such that
    there is a fixed point $H$ of level $i$ that does not satisfy \(\Phi\),
    but all fixed points of level smaller than $i$ satisfy \(\Phi\).
    As \(\Phi(\IS_{p^m}) = 1\), we have $i > 0$.

    \begin{claim}\label{11.13-10}
        There is a fixed point $H$ of
        level $i$ with $\Phi(H) = 0$ and whose empty-prefix is zero.
    \end{claim}
    \begin{claimproof}
        Toward an indirect proof,
        assume that all fixed points $\cgr{p}{A_1} \wrGraph \cdots \wrGraph \cgr{p}{A_{m}}$
        of level $i$ with $\wrLevel(A_1, \dots, A_{m}) = 0$ satisfy $\Phi$.
        We show that now, all fixed
        points of level $i$ satisfy $\Phi$, which is a contradiction to our choice of $i$.

        To that end, fix an $F \coloneqq
        \cgr{p}{A_1} \wrGraph \dots \wrGraph \cgr{p}{A_{m}}$.
        If $\wrLevel(A_1, \dots, A_m) = 0$, then $\Phi(F) = 1$ due to our assumption.
        Otherwise, $\wrLevel(A_1, \dots, A_m) > 0$, which means that $F$ has the
        form $F = \cgr{p}{\emptyset} \wrGraph \cdots \wrGraph \cgr{p}{\emptyset} \wrGraph
        \cgr{p}{A_1} \wrGraph \dots \wrGraph \cgr{p}{A_{m - j}}$ and is
        thus, by \cref{lem:wreath:sub:iso},
        isomorphic to an edge-subgraph of
        \[
            F' \coloneqq \cgr{p}{A_1} \wrGraph \cdots \wrGraph F_p^{A_{m - j}} \wrGraph \cgr{p}{\emptyset}
            \wrGraph \cdots \wrGraph \cgr{p}{\emptyset}.
        \]
        By \cref{lem:level:sylow:graphs},
        the level of $F$ is equal to the level of $F'$.
        Thus by assumption, $F'$ satisfies $\Phi$ and hence $F$ also satisfies $\Phi$ as $\Phi$ is edge-monotone.
    \end{claimproof}

    By \cref{11.13-10},
    we may assume that there is a fixed point $H \coloneqq \cgr{p}{A_1} \wrGraph \cdots \wrGraph \cgr{p}{A_{m}}$ of
    level $i$ with $\wrLevel(A_1, \dots, A_{m}) = 0$ that does not satisfy $\Phi$.
    Now, as the empty-prefix of \(H\) is zero, \cref{lem:treewidth:wreath:level} yields that
    $H$ contains $K_{p^{m-1}, p^{m-1}}$ as a subgraph.
    Finally, we apply
    \cref{remark:chi:all:children:true} to show that $\ae{\Phi}{H}$ does not vanish; which
    completes the proof.
\end{proof}

\subsection{Quantitative Lower Bounds}

Next, we prove \cref{theo:tight:lower:bound}.
Observe that  \cref{theo:tight:lower:bound} differs from \cref{theo:sqrt:lower:bound}
since the constant $\alpha$ of \cref{theo:sqrt:lower:bound} does not depend on a specific prime.

\thetightlowerbound
\begin{proof}
    First, we show that assuming ETH, there are a constant $\gamma'_p
    > 0$ and a constant $N_p$ such that for all fixed $k = p^m \geq N_p$ and each
    edge-monotone graph property $\Phi$ that is nontrivial on $k$,  no
    algorithm computes for each graph \(G\) the number $\NUM{}\indsubs{(\Phi, k)}{G}$ in time
    $O(|V(G)|^{\gamma_p k})$.

    We intend to use \cref{theo:lower:bound:bicliques}, which shows that nonvanishing
    graphs with large bicliques are sufficient to prove the result.
    Write $\beta > 0$ and $N'$ for the constants from \cref{theo:lower:bound:bicliques}.
    We define $h(k) \coloneqq k / p$ and $\gamma'_p \coloneqq
    \beta / p$ and $N_p \coloneqq N' \cdot p$.

    Consider a $k = p^m \geq N_p$. Clearly, $h(k) \geq N'$.
    Further, write $\Phi$ for an edge-monotone graph property that is nontrivial on
    $k$.
    \Cref{theo:edge_mono:prime:power:bicliques} yields a graph $F$ with
    $k$ vertices, $\ae{\Phi}{F} \neq 0$, and that contains $K_{h(k), h(k)}$ as a subgraph.
    Now, if for every graph \(G\), we could compute the number $\NUM{}\indsubs{(\Phi,
    k)}{\star}$ in time $O(|V(G)|^{\gamma_p k}) = O(|V(G)|^{\beta
    h(k)})$, then \cref{theo:lower:bound:bicliques} would show that ETH fails.

    To obtain our lower bound also for \(k < N_p\), we set $\gamma_p \coloneqq \min(\gamma'_p, 1/(N_p+1))$.
    Observe that for \(k = p^m < N_p\), we obtain
    \[
         O(|V(G)|^{\gamma_p k}) = o( |V(G)| ).
    \]
    Now, such a running time is unconditionally unachievable
    for any algorithm that reads
    the whole input.
    This completes the proof.
\end{proof}

Finally, we extend \cref{theo:tight:lower:bound}
from prime powers to products of $c$ prime powers (times a constant~$d$).
However, this comes at the cost of a weaker lower bound in the exponent.

\begin{theorem}
    Let $p_1, \dots, p_c$ denote primes and let $d$ denote a positive integer.
    Then, there is a constant $\gamma > 0$ (that depends on $p_1,
    \dots, p_c$ and $d$) such that for all fixed $k = d \cdot p_1^{m_1} \cdots
    p_c^{m_c} \geq 3$ and all edge-monotone graph properties $\Phi$ that are nontrivial on $k$, no
    algorithm (that reads the whole input) computes for every graph $G$ the number
    $\NUM{}\indsubs{(\Phi, k)}{G}$ in time $O(|V(G)|^{\gamma \sqrt[c]{k}})$, unless ETH fails.
\end{theorem}
\begin{proof}
    First, we show that assuming ETH,
    there are a constant $\gamma' > 0$ and a constant $N$ such that for all fixed $k = d
    \cdot p_1^{m_1} \cdot \dots \cdot p_c^{m_c} \geq N$ and each edge-monotone
    graph property $\Phi$ that is nontrivial on $k$, no algorithm computes for each graph
    \(G\) the number $\NUM{}\indsubs{(\Phi, k)}{G}$ in time
    $|V(G)|^{\gamma \sqrt[c]{k}}$.

    To that end, we set $h(k) \coloneqq \lceil \sqrt[c]{k/d} \rceil$.
    Observe that from \cref{lem:upper:bound:q}, we obtain
    $q(k) \geq h(k)$ for all $k$ with $k = d \cdot p_1^{m_1} \cdot \dots \cdot p_c^{m_c}$.
    As \(q(k)\) is a positive integer,
    we may safely round up $\sqrt[c]{k/d}$. Further, write $\hat{c} > 0$ for the constant of \cref{lem:upper:bound:q} (3).

    Next, write $\gamma_0 > 0$ and $N'$ for the universal constants from
    \cref{theo:lower:bound:bicliques}.
    Further, for each prime number $p_i$ we use
    \cref{theo:tight:lower:bound} to obtain a constant $\gamma_i$.
    Now, set $\gamma' \coloneqq
    \min(\gamma_0, \dots, \gamma_p) / (2\sqrt[c]{d})$
    and $N \coloneqq \max(d \cdot N'^c, 4^c / (d \gamma^c), \exp(d/{\hat{c}}) + 1)$.

    Fix a $k = d \cdot p_1^{m_1} \cdot \dots \cdot p_c^{m_c} \geq N$ and
    an edge-monotone graph property $\Phi$ that is nontrivial on $k$.
    Now, \(\Phi\) is either scattered or concentrated on \(k\); we provide a proof for both cases.

    First, we consider the case that \(\Phi\) is scattered on \(k\).
    Define $m=q(k)$. Let $H_m$ be the graph defined in \cref{11.5-2} such that \(
        (\Phi - {H}) \coloneqq \{ G \mid G \UnionGraph H  \in \Phi\}
    \) is edge-monotone and nontrivial on \(q(k)\).
    Now, we use \cref{lem:upper:bound:q} to obtain $q(k) \geq \hat{c} \log(k) > \hat{c}
    \log(\exp(d/{\hat{c}})) = d$, which implies  $q(k) = p_i^{\hat{m}_i}$ for some $i \in
    \setn{c}$ and $\hat{m}_i \geq m_i$. Further, we obtain $q(k) \geq \sqrt[c]{k/d} \geq
    N'$ due to \cref{lem:upper:bound:q}.
    Similarly to the proof of \cref{theo:edge_mono}, assuming ETH,
    \cref{theo:tight:lower:bound} shows that there is no algorithm that
    for each graph \(G\) computes the number
    $\NUM{}\indsubs{((\Phi - {H}), q(k))}{G}$ in time $O(|V(G)|^{\gamma_i q(k)})$.

    Now, assume that for each graph \(G\), we could compute $\NUM{}\indsubs{(\Phi, k)}{G}$ in
    time $O(|V(G)|^{\gamma'
    \sqrt[c]{k}})$. Then, using \cref{lem:inc:exc} to compute
    $\NUM{}\indsubs{((\Phi - {H}), q(k))}{G}$ in time
    \begin{align*}
        O(2^{k - q(k)} (|V(G)| + k - q(k))^2 \cdot (|V(G)| +
        k - q(k))^{\gamma' \sqrt[c]{k}})
        &= O(|V(G)|^2 \cdot |V(G)|^{\gamma' \sqrt[c]{k}}) \\
        &\stackrel{\ast}{\subseteq} O(|V(G)|^{\gamma_i \sqrt[c]{k/d} })
        \subseteq  O(|V(G)|^{\gamma_i
    q(k)})
    \end{align*}
    would show that ETH fails.
    The last step is justified by $\sqrt[c]{k/d}
    \leq q(k)$. The step $(\ast)$ is justified by \[2 + \gamma \sqrt[c]{k} \leq  2 +
        \frac{\gamma_i}{2 \sqrt[c]{d}} \sqrt[c]{k} \leq \frac{\gamma_i}{\sqrt[c]{d}} \sqrt[c]{k} =
    \gamma_i \sqrt[c]{k/d}.\] The second inequality is equivalent to $4^c / (d \gamma_i^c) \leq
    k$, which is true since $N \leq k$.

    Next, we consider the case that \(\Phi\) is concentrated on \(k\).
    Then by definition there is a \(k\)-vertex graph $H$ with a nonvanishing
    alternating enumerator that contains $K_{q(k), q(k)}$ as a subgraph.
    Now, since
    $k \geq d \cdot N'^c$ and since $h$ is monotonically increasing, we obtain $h(k) \geq h(d \cdot
    N'^c) \geq N'$.
    Further, we have $q(k) \geq h(k) \geq \sqrt[c]{k/d}$, thus
    \cref{theo:lower:bound:bicliques} yields that there is no algorithm that
    computes for each graph \(G\) the number
    $\NUM{}\indsubs{(\Phi, k)}{G}$ in time $O(|V(G)|^{\gamma_0 \sqrt[c]{k/d}}) \supseteq
    O(|V(G)|^{\gamma' \sqrt[c]{k}})$, where this is justified by $\gamma' \leq
    \gamma_0 / \sqrt[c]{d}$.

    To obtain our lower bound also for \(k < N\), we set $\gamma \coloneqq \min(\gamma', 1/(N+1))$.
    Now, for \(k < N\) we obtain
    \[
        O(|V(G)|^{\gamma \sqrt[c]{k}}) = o( |V(G)| ).
    \]
    Now, such a running time is unconditionally unachievable
    for any algorithm that reads
    the whole input.
    This completes the proof.
\end{proof}

{
    \bibliographystyle{alphaurl}
    \bibliography{refs}
}

\appendix
\clearpage
\normalsize
\section{Graph Properties and the Alternating Enumerator} \label{sec:appendix:A}

As described in the Technical Overview, using a \(p\)-subgroup of the automorphism group,
we can simplify the computation of the alternating enumerator (modulo \(p\)).

\lemchicomp
\begin{proof}
    We follow the proof of \cite[Lemma~1]{alge}.

    First, we rewrite the definition of the alternating enumerator to use edge-subgraphs
    of \(H\) instead of subsets of the edges of \(H\).
    To that end, we readily see that
    for each edge-subgraph $A\in\edgesub(H)$, we can find a subset $S \subseteq E(G)$ with
    $A = \ess{H}{S}$ and vice versa. Hence, we obtain \[
    \ae{\Phi}{H}
    = \sum_{S \subseteq E(H)} \Phi(\ess{H}{S}) (-1)^{\NUM{}E(S)}
    = \sum_{A \in \edgesub(H)} \Phi(A) (-1)^{\NUM{}E(A)}
    .\]

    Now, recall that $\Gamma$ acts on $\edgesub(G)$
    and consider the orbits of this group action $\cdot : \Gamma \times \edgesub(G) \to
    \edgesub(G)$. We choose a representative for each orbit, which means that each orbit has the form
    $\Gamma A_0 \coloneqq \{g A_0 \mid g \in \Gamma \}$ for an edge-subgraph $A_0 \in \edgesub(G)$.
    We write \(\mathcal{A}\) to denote the set of all representatives.

    The orbits of $\cdot$ partition the set $\edgesub(G)$, which
    allows us the rewrite the alternating enumerator as
    \[\ae{\Phi}{H}
        = \sum_{A \in \edgesub(H)} \Phi(A) (-1)^{\NUM{}E(A)}
        = \sum_{A_0 \in \mathcal{A}} \sum_{A \in \Gamma A_0} \Phi(A) (-1)^{\NUM{}E(A)}.  \]
    Now, fix an $A \in \Gamma A_0$.
    By construction, there is a graph automorphism $g$ with $g(A) = A_0$.
    Thus, we have $\Phi(A) = \Phi(A_0)$ and $\NUM{}E(A)= \NUM{}E(A_0)$ and hence
    \[\ae{\Phi}{H}
        = \sum_{A_0 \in \mathcal{A}} \sum_{A \in \Gamma A_0} \Phi(A_0) (-1)^{\NUM{}E(A_0)}
        = \sum_{A_0 \in \mathcal{A}} (\NUM{} \Gamma A_0) \Phi(A_0) (-1)^{\NUM{}E(A_0)}.
    \]
    Now, we use the Orbit Stabilizer Theorem to see that the size $\NUM{} \Gamma A_0$
    of any orbit of $\cdot$ is a divider of the group order of $\Gamma$.
    As \(\Gamma\) is a \(p\)-group, its order is equal to $p^k$ for some $k \in \nat$.
    Hence, $\NUM{} \Gamma A_0 \mod p$ is either equal to \(0\) (if $\NUM{} \Gamma A_0 > 1$);
    or equal to $1$ (if $\NUM{} \Gamma A_0 = 1$).
    Finally, observe that $\NUM{} \Gamma A_0 = 1$ if and only if $A_0$ is a
    fixed point of $\Gamma$. This in turn means that only the
    fixed points remain when computing $\ae{\Phi}{H} \mod p$. Hence, we obtain the claimed equation
    \[
        \ae{\Phi}{H}
        \equiv_p \sum_{A \in \fp(\Gamma, H)} \Phi(A) (-1)^{\NUM{}E(A)}.
        \qedhere
    \]
\end{proof}

Using \cref{remark:subgroup:fixpoint}, we obtain a useful strengthening of \cref{lem:chi:comp}.

\begin{corollary}\label{cor:chi:compute:fixed:points}
    For a graph $G$, a \(p\)-group $\Gamma \subseteq \auts{G}$, and a fixed point $H \in
    \fps{G}$, we have
    \begin{align*}
        \ae{\Phi}{H} \equiv_p
        \sum_{\substack{A \in \fps{G} \\ E(A) \subseteq E(H) }} \Phi(A) (-1)^{\NUM{}E(A)}.
    \end{align*}
\end{corollary}
\begin{proof}
    We use \cref{remark:subgroup:fixpoint}(2), and in particular the characterization of
    \(\fp{H}\) in terms of \(\fp{G}\). Now, \cref{lem:chi:comp} yields
    the claim.
\end{proof}

\subsection{Lower Bounds for Counting Induced Subgraphs via the Alternating Enumerator}

In this section, we prove \cref{lem:alpha:treewidth}; which was implicitly proved in \cite{alge} (for the
\w-hardness).
As advertised, we augment \cref{lem:alpha:treewidth} with a slightly stronger quantitative
lower bound by lifting a similar result for \(\NUM{}\homsprob(\mathcal{H})\) due to \cite{cohenaddad2021tight}.

\lemchitw*
\medskip
To that end, we need to work with colored versions of $\NUM{}\homsprob(\mathcal{H})$ and
$\NUM{}\indsubsprob(\Phi)$, which we define next. Informally, in a colored problem, the
vertices of the input graph $G$ are partitioned into classes and we have to select exactly
one vertex from each class. In the colored version of the homomorphism problem, we are
given two graphs $G$ and $H$, with a coloring $c:V(G)\to V(H)$ (that is, a partitioning
$V(G)$ into $|V(H)|$ classes). A~color-prescribed homomorphism $h$ from $H$ to \(G\)
is a homomorphism from $H$ to $G$ such that $c(h(v)) = v$ for
all $v \in V(H)$. Observe that if $u$ and $v$ are not adjacent in $H$, then the existence of
the edges in $G$ between $c^{-1}(u)$ and $c^{-1}(v)$ does not play any role whatsoever in
the problem. Hence we might as well assume that there are no such edges in $G$ at all,
which formally means that $c$ is a homomorphism from $G$ to $H$. Therefore, we assume
that $c$ is indeed such a homomorphism, or in other words, $G$ is {\em $H$-colored via
$c$.}

We write \(\cphoms{H}{G}\) for the set of all color-prescribed homomorphism from $H$ to
a graph \(G\) that is \(H\)-colored via \(c\).
For a recursively enumerable class of graphs $\mathcal{H}$,
in the problem $\NUM{}\cphomsprob(\mathcal{H})$ we are given a
graph $H \in \mathcal{H}$ and a graph $G$ that is \(H\)-colored via \(c\), the task is
to compute the value $\NUM{}\cphoms{H}{G}$.
We parameterize $\NUM{}\cphomsprob(\mathcal{H})$ by $\kappa(H, G) \coloneqq |V(H)|$.

In the colored variant of \(\NUM{}\indsubsprob(\Phi)\), the vertices of the input graph $G$
are partitioned into $k$ classes and we are counting the number of $k$-vertex induced
subgraphs satisfying $\Phi$ that contains exactly one vertex from each class. However, we
need to define the problem in a way that allows a closer connection to
$\NUM{}\cphomsprob(\mathcal{H})$.
For a recursively enumerable class of graphs $\mathcal{H}$, the input of
\(\NUM{}\cpindsubsprob(\Phi,\mathcal{H})\) consists   of a graph $G$, a graph
$H\in\mathcal{H}$, and a $H$-coloring $c$ of $G$. The task is to compute the number of
$|V(H)|$-vertex induced subgraphs of $G$ that satisfies $\Phi$ and has exactly one vertex
with each of the $|V(H)|$ colors. We denote this number by
$\NUM{}\cpindsubs{(\Phi, H)}{G}$.
We parameterize \(\NUM{}\cpindsubsprob(\Phi,\mathcal{H})\) by $\kappa(H, G) \coloneqq |V(H)|$.

Write $\Phi$ for a property and let $\mathcal{H}$ contain every graph with nonvanishing
alternating enumerator.
The proof of \cref{lem:alpha:treewidth} use the hardness of   $\NUM{}\homsprob(\mathcal{H})$
to
obtain hardness for \(\NUM{}\indsubsprob(\Phi)\). It relies on the following chain of reductions of~\cite{alge}.

\begin{align}\label{reduction:chain}
    \NUM{}\homsprob(\mathcal{H})
    &\overset{\text{\cite[Lemma 4]{alge}}}{\fpt}
    \NUM{}\cphomsprob(\mathcal{H})\\\nonumber
    &\overset{\text{\cite[Lemmas 7 and 8]{alge}}}{\fpt}
    \NUM{}\cpindsubsprob(\Phi,\mathcal{H})
    \overset{\text{\cite[Lemma 10]{alge}}}{\fpt}
    \NUM{}\indsubsprob(\Phi);
\end{align}

The first reduction is very simple (essentially, making $|V(H)|$ copies of the vertex set
of $G$) and the last reduction is a standard application of the Inclusion-Exclusion
principle. Thus, let us focus on  the reduction \(\NUM{}\cphomsprob(\mathcal{H}) \fpt
\NUM{}\cpindsubsprob(\Phi)\) and the corresponding key lemma from \cite{alge}.

\begin{lemma}[{\cite[Lemma 8]{alge}}]\label{lem:chi:cpindub}
    Let $H$ denote a graph, let $\Phi$ denote a graph property,
    and let $G$ denote an $H$-colored graph.
    Then, we have
    \[
        \NUM{}\cpindsubs{(\Phi, H)}{G}
        = \sum_{S \subseteq E(H)} \Phi(\ess{H}{S}) \sum_{J \subseteq E(H) \setminus S}
        (-1)^{\NUM{}J} \cdot \NUM{}\cphoms{\ess{H}{S \cup J}}{G}.
    \]
    Further, the absolute values of $\ae{\Phi}{H}$ and of the coefficient of
    $\NUM{}\cphoms{H}{G}$ are equal.
    \lipicsEnd
\end{lemma}

In particular, the second part of \cref{lem:chi:cpindub} yields
that a term $\NUM{}\cphoms{H}{G}$ appears in the sum if and only if $\ae{\Phi}{H} \neq 0$,
which is part of our assumption in \cref{lem:alpha:treewidth}.

Observe that \cref{lem:chi:cpindub} in itself does not suffice to obtain the claimed reduction
\(\NUM{}\cphomsprob(\mathcal{H}) \fpt \NUM{}\cpindsubsprob(\Phi)\):
the oracle for $\NUM{}\cpindsubsprob(\Phi)$ computes only a sum in which $\NUM{}\cphoms{H}{G}$ occurs
as some term---we still need to extract the value $\NUM{}\cphoms{H}{G}$ out of the result of
the oracle.
Fortunately for us, \cite[Lemma~7]{alge} does exactly that by showing a generalization of
the \emph{Complexity Monotonicity} of~\cite{hom:basis}.
The other reductions of (\ref{reduction:chain}) can be used without modifications.
In total, we obtain the \w-hardness part of \cref{lem:alpha:treewidth}.

Next, we turn to ETH-based lower bounds.
We start from the binary constraint satisfaction problem (CSP).

\begin{definition} A (binary) CSP instance is a triple $(V, D, C)$ where
    \begin{itemize}
        \item $V$ is a set of variables
        \item $D$ is a domain of values,
        \item $C$ is a set of constraints. Each constraint is a triple $(u, v, R)$ where
            $(u, v) \in V^2$, and $R \subseteq D^2$.
    \end{itemize}
    A \emph{solution} to $(V, D, C)$ is a function $f \colon V \to D$ such that for all constraints
    $(u, v, R)$, the pair $(f(u), f(v))$ is~in~$R$.
    The \emph{primal graph} of a
    CSP insurance $(V, D, C)$ is a graph $H$ with vertex set $V$ such that $u, v \in V(H)$ are
    adjacent if and only if there is a constraint in $C$ of the form $(u, v, R)$.
\end{definition}

In particular, we use the following result of Cohen-Addad, Colin de Verdière, Marx, and de
Mesmay, which we can easily modify for our purposes.

\begin{theoremq}[{\cite[Theorem 2.7]{cohenaddad2021tight}}]
    Assuming ETH, there is a universal constant $\alpha_{\text{CSP}} > 0$ such that for any
    fixed graph $H$ with $\tw(H) \geq 2$, there is no algorithm that decides the binary CSP
    instances whose primal graph is \(H\) in time
    $O(|D|^{\alpha_{\text{CSP}} \cdot \tw(H) / \log\tw(H)})$.
\end{theoremq}

\begin{corollary}\label{lem:hom:lower:bound}
    Assuming ETH, there is a universal constant $\alpha_{\homsprob} > 0$ such that for any
    fixed graph $H$ with $\tw(H) \geq 2$, there is no algorithm that computes
    $\NUM{}\homs{H}{\star}$ on input $G$ in time $O(|V(G)|^{\alpha_{\homsprob} \cdot \tw(H) / \log \tw(H)})$.
\end{corollary}
\begin{proof}
     We model the decision problem $\homsprob(\{H\})$ as a CSP.
     For a given graph $G$, we define the CSP instance $I = (V(H), V(G), C)$ with $C \coloneqq
     \{( u, v, E(G)) : \{u, v\} \in E(H)\}$.
     Each solution $h \colon V(H) \to V(G)$ to $I$ is also a homomorphism from $H$ to $G$
     and each  homomorphism from $H$ to $G$ is also a solution to \(I\) (observe that
     whenever the sole relation in \(I\) is the edge relation, the definition of a
     solution to \(I\) coincides with the definition of a graph homomorphism).
     Thus, $\homs{H}{G}$ is non-empty if and only if $I$ has a solution.

     As $\homsprob(\{H\})$ can be solved via its counting version
     $\NUM{}\homsprob(\{H\})$, we obtain that $\NUM{}\homs{H}{G}$ cannot be computed in time
     $O(|V(G)|^{\alpha_{\text{CSP}} \cdot \tw(H) / \log \tw(H)})$.
     Choosing \(\alpha_{\homsprob} \coloneqq \alpha_{\text{CSP}}\) yields the claim.
\end{proof}

Next, we again use the reductions of \eqref{reduction:chain}.
In particular, we observe that said reductions preserve the exponent in the running time
(up to some constant additive term).
It is useful to have a separate (sub-)claim for a reduction that starts from
color-prescribed homomorphisms.

\begin{lemma}\label{lem:cphom:to:indsub}\label{claim:hom:tw}
    Let $\Phi$ denote a graph property and
    let $H$ denote a graph such that $\ae{\Phi}{H} \neq 0$.

    Any algorithm that computes for each graph \(G'\) the number $\NUM{}\indsubs{(\Phi,
    |V(H)|)}{G'}$ in time $O(|V(G')|^{\beta})$ implies
    \begin{itemize}
        \item an algorithm that for each graph
            \(G\) computes the number $\NUM{}\cphoms{H}{G}$ in time $O(|V(G)|^{\beta + 2})$.
        \item an algorithm that for each graph
            \(G\) computes the number $\NUM{}\homs{H}{G}$ in time $O(|V(G)|^{\beta + 3})$.
    \end{itemize}
\end{lemma}
\begin{proof}
    Write \(k \coloneqq |V(H)|\) and observe that for our purposes, \(k\) is a
    constant.
    We first construct an algorithm that solves $\NUM{}\cphomsprob(\{H\})$ using an oracle for
    $\NUM{}\indsubs{(\Phi, k)}{\star}$.
    To that end, suppose that we are given a graph \(G\) that is $H$-colored via \(c\).
    We wish to compute the number \(\NUM{}\homs{H}{G}\).
    \begin{enumerate}[(1)]
        \item First, we use the reduction from \(\NUM{}\cphomsprob(\mathcal{H})\)
            to \(\NUM{}\cpindsubsprob(\Phi)\)~\cite[Lemmas~7~and~8]{alge} to compute for
            the graph $G$ the number
            $\NUM{}\cphoms{H}{G}$  in time $O(f(k) \cdot |V(G)|)$
            using an oracle for
            $\NUM{}\cpindsubs{(\Phi, H)}{\star}$
            for a computable function $f$.
            In said reduction, each graph $G'$ that is used inside an oracle call satisfies
            $|V(G')| \leq f(k) \cdot |V(G)|$.

        \item Next, we use the reduction from \(\NUM{}\cpindsubsprob(\Phi)\)
            to \(\NUM{}\indsubsprob(\Phi)\)~\cite[Lemma~10]{alge} to compute for
            each \(G'\) the number
            $\NUM{}\cpindsubs{(\Phi, H)}{G'}$  in time $O(g(k) \cdot |V(G')|)$ using an oracle for
            $\NUM{}\indsubs{(\Phi, |V(H)|)}{\star}$ for a computable function $g$.
            In said reduction, each graph $G''$ that is used in oracle calls satisfy
            $|V(G'')| \leq |V(G')|$.

        \item Lastly, we compute for each \(G''\) the number $\NUM{}\indsubs{(\Phi, |V(H)|)}{G''}$
            in time  $O(|V(G'')|^{\beta})$ using the algorithm which we
            assumed to exist.
    \end{enumerate}
    We combine the above steps to obtain an algorithm that computes $\NUM{}\cphoms{H}{G}$
    in time
    \[
        O\big(\underbrace{f(k) \cdot |V(G)|}_{(1)}
            \quad\cdot\quad \underbrace{g(k) \cdot (f(k) \cdot |V(G)|)}_{(2)}
        \quad\cdot\quad \underbrace{(f(k)\cdot |V(G)|)^{\beta}}_{(3)}  \;\big),
    \]
    which can be rewritten into $O(|V(G)|^2 \cdot |V(G)|^{\beta})$, as \(k = |V(H)|\) is a constant.

    Next, we construct an algorithm that solves $\NUM{}\homsprob(\{H\})$ using an oracle for
    $\NUM{}\indsubs{(\Phi, k)}{\star}$.
    To that end, suppose that we are given a graph \(G\).
    We wish to compute the number \(\NUM{}\homs{H}{G}\).

    First, we use the reduction from \(\NUM{}\homsprob(\mathcal{H})\)
    to \(\NUM{}\cphomsprob(\mathcal{H})\)~\cite[Lemma~4]{alge} to compute the value
    $\NUM{}\homs{H}{G}$  in time $O(h(k) \cdot |V(G)|)$ using an oracle for
    $\NUM{}\cphoms{H}{\star}$ for a computable function $h$.
    In said reduction, each graph $G'$ that is used in oracle calls satisfy $|V(G')|
    \leq h(k) \cdot |V(G)|$.

    Next, our algorithm from the first part of the proof allows us to compute each value
    $\NUM{}\cphoms{H}{G'}$ in time $O((h(k) |V(G')| )^{\beta + 2})$ using a oracle for $\NUM{}\indsubs{(\Phi,
    |V(H)|)}{\star}$.
    Hence, we obtain a total running time of $O(h(k) \cdot |V(G)| \cdot (h(k) \cdot |V(G)|)^{\beta + 2} )$
    which can be rewritten into $O(|V(G)|^3 \cdot |V(G)|^{\beta})$ as $k$ is a constant.
\end{proof}

Putting everything together, we obtain \cref{lem:alpha:treewidth}.

\lemchitw
\begin{proof}
    We write $\mathcal{H}$ for the set $\{H \in \mathcal{G} : \ae{\Phi}{H} \neq 0 \}$.
    The set \(\mathcal{H}\) is recursively enumerable since the set of all graphs is recursively
    enumerable and $\Phi$ is computable. This in turn implies that the alternating enumerator is
    computable.

    Observe that the treewidth of the elements of $\mathcal{H}$ is
    unbounded by the assumption of the lemma.
    This means that $\NUM{}\homsprob(\mathcal{H})$ is \w-hard when parameterized by
    the pattern size $|V(H)|$.
    Finally, we use the parameterized reductions
    \eqref{reduction:chain} from $\NUM{}\homsprob(\mathcal{H})$ to $\NUM{}\indsubsprob(\Phi)$, which proves
    that $\NUM{}\indsubsprob(\Phi)$ is also \w-hard.

    Next, we turn to the ETH-based lower bounds.
    Set $\alpha'_{\indsubsprob} \coloneqq \alpha_{\homsprob}/2$, $N \coloneqq  \max(2, (6
    / \alpha_{\homsprob})^2)$. Let \(k\) denote a fixed integer such that there is a
    nonvanishing $k$-vertex graph $H_k$ with $\tw(H_k) \geq N$.
    We show that any algorithm \(\mathbb{A}\) that computes for every graph \(G\) the
    value  $\NUM{}\indsubs{(\Phi, k)}{G}$ in time
    $O(|V(G)|^{\alpha'_{\indsubsprob} \tw(H_k) / \log\tw(H_k) })$ implies that we can compute for
    any graph \(G\) the value
    $\NUM{}\homs{H_k}{G}$ in time $O(|V(G)|^{\alpha_{\homsprob} \tw(H_k) / \log\tw(H_k)})$ (which
    would violate ETH due to \cref{lem:hom:lower:bound}).

    To that end, first observe that by \cref{claim:hom:tw}, the algorithm \(\mathbb{A}\)
    yields an algorithm \(\mathbb{B}\) to compute for every graph \(G\) the value \(\NUM{}\homs{H_k}{G}\) in
    time \(O(|V(G)|^{\alpha'_{\indsubsprob} \tw(H_k) / \log\tw(H_k) + 3})\).

    \begin{claim}\label{cl:10.27-1}
        For $\tw(H_k) \geq \max( 2, (6 / \alpha_{\homsprob} )^2 )$, we have
        \[\alpha'_{\indsubsprob}  \tw(H) / \log\tw(H_k) + 3 \le \alpha_{\homsprob}
        \tw(H_k) / \log\tw(H_k).\]
    \end{claim}
    \begin{claimproof}
        For \(\tw(H_k) > 1\), we have $\tw(H_k) / \log\tw(H_k) > \sqrt{\tw(H_k)}$.
        Hence, for $\tw(H_k) \geq (6 / \alpha_{\homsprob} )^2$, we have
        \[
            \tw(H) / \log\tw(H_k) \geq 6 / \alpha_{\homsprob} =
            3 / ( \alpha_{\homsprob} - \alpha_{\homsprob} / 2 ) =
            3 / ( \alpha_{\homsprob} - \alpha'_{\indsubsprob} ).
        \]
        Now, rearranging yields the claim.
    \end{claimproof}

    From  \cref{cl:10.27-1}, we conclude that \(\mathbb{B}\) computes
    $\NUM{}\homs{H_k}{G}$ in time $O(|V(G)|^{\alpha_{\homsprob} \tw(H_k) / \log\tw(H_k)})$.
    From \cref{lem:hom:lower:bound} (and assuming \(\tw(H_k) \ge 2\)), we conclude that the algorithm \(\mathbb{B}\)
    violates ETH.

    Define $\alpha_{\indsubsprob} \coloneqq \min(\alpha'_{\indsubsprob}, 1/N)$. For $2 \leq
    \tw(H_k) < N$, we obtain an algorithm that computes $\NUM{}\indsubs{(\Phi,
    |V(H_k)|)}{G}$ in time $|V(G)|^{\alpha_{\indsubsprob} \tw(H_k) / \log \tw(H_k)}
    = o(|V(G)|)$.
    Now, such a running time is unconditionally unachievable
    for any algorithm that reads
    the whole input.
    This completes the proof.
\end{proof}

\section{ETH-based Lower Bounds for \texorpdfstring{\(k\)}{k}-Clique} \label{sec:tight:bounds:bicliques}

We discuss the following useful result on finding \(k\) cliques.

\begin{theoremq}[Theorem 14.21 in \cite{param_algo}]\label{theo:k-clique:ETH:original}
    Assuming ETH, there is no $f(k)|V(G)|^{o(k)}$-time algorithm for ${\clique}$ or
    $\textsc{Independent Set}$ for any computable function $f$.
\end{theoremq}

Next, we modify \cref{theo:k-clique:ETH:original}.
The decision problem $k$-$\clique$ gets as input a graph $G$ and checks if $G$
contains a $k$-clique as a subgraph.

\begin{lemma}[{Modification of 14.21 in \cite{param_algo}}]\label{theo:k-clique:ETH}
    Assuming ETH, there is a constant $\alpha > 0$ such that for $k \geq 3$, no
    algorithm (that reads the whole input) solves $k$-$\clique$ on graph $G$ in time $O(|V(G)|^{\alpha k})$.
\end{lemma}
\begin{proof}
    Write $n \coloneqq |V(G)|$ for the number of vertices of the input graph $G$.
    Assuming ETH, there is a $\delta > 0$ such that no algorithm
    solves 3-\textsc{SAT} in time $O(2^{\delta v})$, where $v$ is the number of variables.
    We show that a similar statement also holds for the 3-\textsc{Coloring} problem.
    \begin{claim}\label{claim:3-coloring}
        Assuming ETH, there is a $\delta' > 0$ such that no algorithm
        solves 3-\textsc{Coloring} in time $O(2^{\delta' n})$, where $n$ is the number of
        vertices.
    \end{claim}
    \begin{claimproof}
       For a 3-\textsc{SAT} formula $\phi$ with $v$ variables and $m$ clauses, the standard
       reduction \cite[Theorem 2.1]{GAREY1976237} from 3-\textsc{SAT} to 3-\textsc{Coloring} constructs a graph with
       $3 + 2v + 6m$ vertices that contains a 3-coloring if and only if $\phi$ is satisfiable. An
       algorithm that solves 3-\textsc{Coloring} in
       time $O(2^{\beta n})$ could be used to solve 3-\textsc{SAT} in time $O(2^{6\beta(v + m)})$.
       The {Sparsification Lemma} \cite[Theorem 14.4]{param_algo} yields the existence of
       a $\delta' > 0$ such that there is no algorithm that solves 3-\textsc{Coloring} in
       time $O(2^{\delta' n})$.
    \end{claimproof}
    Next, we show that we can solve 3-\textsc{Coloring} by using an algorithm for
    $k$-\textsc{Clique} that runs in $O(n^{\alpha k})$.
    \begin{claim}\label{claim:clique:to:coloring}
        If we can solve $k$-\textsc{Clique} in time $O(n^{\alpha
        k})$, then we can solve 3-\textsc{Coloring} in time $O(3^{\alpha n} + n^2 \cdot
        3^{2n/k})$.
    \end{claim}
    \begin{claimproof}
        First, we split the $n$ vertices of $G$
        into $k$ blocks $V_1, \dots, V_k$ of size at most $\lceil n / k \rceil$ each.
        Next, we construct a graph $H$ in the following way.
        For each proper 3-coloring of a block
        $V_i$, we create a vertex that represents this coloring and add said vertex to $H$.
        The number of vertices in $H$ is at most
        \[|V(H)| \leq k \cdot 3^{\lceil n / k \rceil} \leq k \cdot 3^{  n / k + 1}. \]
        Write $u$ for a vertex in $H$ that represents a coloring of $G\position{V_i}$ and
        write $v$ for a vertex in $H$ that represents a coloring of $G\position{V_j}$.
        We add an edge between $u$ and
        $v$ if and only if $i \neq j$ and the coloring of $u$ and $v$ is a proper 3-coloring of
        $G\position{V_i \cup V_j}$.
        Observe that we can construct this graph in time
        $O((k \cdot 3^{  n / k + 1})^2 \cdot n^2 ) = O(n^2 \cdot 3^{2n/k})$.

        It is easy to verify that $H$ contains a $k$-clique if and only if there is a
        proper 3-coloring of $G$. Thus, we can use our $O(n^{\alpha k})$ time algorithm
        for $k$-\textsc{Clique} to solve 3-\textsc{Coloring} in time
        $O((k \cdot 3^{  n / k + 1} )^{\alpha k} + n^2 \cdot 3^{2n/k}) = O(3^{\alpha n} + n^2 \cdot 3^{2n/k})$.
    \end{claimproof}
    Set $\alpha \coloneqq \min(\log_3(2) \delta', \alpha/3)$.
    If there is a $k > 2/\alpha$ such
    that we can solve $k$-\textsc{Clique} in time $O(n^{\alpha k})$, then we can use
    \cref{claim:clique:to:coloring} to solve 3-\textsc{Coloring} in time
    $O(3^{\alpha n} + n^2 \cdot 3^{2n/k}) \subseteq O(2^{\delta' n})$. According to
    \cref{claim:3-coloring}, this is only possible if ETH fails. Otherwise, $\alpha k <
    1$, and there is no sublinear algorithm that reads the whole input and solves
    $k$-\textsc{Clique} in sublinear time.
\end{proof}

\section{Tight Lower Bounds for Counting Induced Subgraphs}
\label{app:use:biclique}

In this section, we show how to obtain tight lower bounds under ETH for $\NUM{}\indsubs{(\Phi,
k)}{\star}$ by modifying a reduction from \cite{alge}. We start from the non-parameterized decision
problem $k$-\textsc{Clique} and solve it with an oracle for $\NUM{}\indsubs{(\Phi,k)}{\star}$.  Observe that
$k$-${\clique}$ cannot be solved in time $O(n^{\alpha k})$ for a fixed $\alpha > 0$ unless ETH
fails (see \cref{theo:k-clique:ETH}), which yields a lower bound for computing
$\NUM{}\indsubs{(\Phi,k)}{\star}$. The reduction is similar to the reduction shown in
\cref{sec:appendix:A} with the difference that we start from $k$-${\clique}$ instead of
$\homsprob(\{H\})$.  The main idea is that we can solve $k$-${\clique}$ using an oracle
for $\NUM{}\cphoms{F}{\star}$ as long as $F$ contains $K_{k, k}$ as a subgraph.

\begin{lemma}[Modification of {\cite[Lemma 11]{alge}}]\label{lem:tight:biclique:construction}
    There is an algorithm that
    given a positive integer $\ell > 1$, a graph $F$ (that contains $K_{\ell, \ell}$ as a
    subgraph), and a graph $G$; computes a $F$-colored graph $G'$ with $2 \ell |V(G)|
    + (|V(F)| - 2 \ell)$ vertices. Further, the number of cliques of size $\ell$ in $G$
    equals $\NUM{}\cphoms{F}{G'}$.
    The running time of the algorithm is $O(|V(F)|^{|V(F)| + 2} + |V(F)|^2 |V(G)|^2 )$
\end{lemma}
\begin{proof}
    Write $\Tilde{\ell}$ for $|V(F)| -
    2\ell$. For the proof,  we modify the construction of \cite[Lemma 11]{alge}. First,
    observe that we can locate the vertices of the subgraph $K_{\ell, \ell}$ in $F$ by
    looping over all vertex sets $A$ and $B$ of size $\ell$ and checking if the induced
    graph has the complete bipartite graph $K_{{\ell, \ell}}$ as a subgraph. This can be
    done in $O(|V(F)|^2 \cdot |V(F)|^{2 \ell}) \subseteq O(|V(F)|^{|V(F)| + 2})$ by using
    a brute force implementation.

    Next, we relabel the vertices of the graph $F$ by splitting them up into a left side,
    a right side and the remaining vertices. Formally, we use $V(F) = \{a_i, b_i \mid i \in
    \setn{\ell}\}  \cup \{x_i \mid i \in \setn{\Tilde{\ell}}\}$, and ensure that the induced subgraph
    of $\{a_i, b_i \mid i \in \setn{\ell}\}$ contains $K_{{\ell, \ell}}$, where $\{a_i \mid i \in
    \setn{\ell}\}$ is the left side and $\{b_i \mid i \in \setn{\ell}\}$ is the right side of the
    complete bipartite graph.\footnote{Note that they may be edges of the form $a_i$ to
    $a_j$ or form $b_i$ to $b_j$.}

    Now, let $G$ denote a graph with vertex set $\{v_i \mid i \in \setn{n}\}$. We construct the graph
    $G'$ on the vertex set $\{u_{i, j}, w_{i, j}, y_{k} : i \in \setn{\ell}, j \in \setn{n}, k \in
    \setn{\Tilde{\ell}}\}$ with the $F$-coloring given by $c(u_{i,j}) = a_i$, $c(w_{i, j}) =
    b_i$ and $c(y_k) = x_k$. For each $k$, we add an edge between $y_k$ and all other
    vertices in $G'$. Additionally, we add an edge between $u_{i, j}$ and $w_{i',
    j'}$ if and only if
    \begin{itemize}
        \item either $(i, j) = (i', j')$,
        \item or $i < i', j < j'$ and the vertices $v_j$ and $v_{j'}$ are adjacent,
        \item or $i > i', j > j'$ and the vertices $v_j$ and $v_{j'}$ are adjacent.
    \end{itemize}
    Further, we add the edges $\{u_{i, j}, u_{i', j'}\}$ and $\{w_{i, j}, w_{i', j'}\}$ to
    $G'$.

    Let $C = (v_{j_1}, \dots, v_{j_\ell})$ denote an ordered tuple (that is $j_{k} < j_{k'}$ for $k <
    k'$) such that $\{v_{j_1}, \dots, v_{j_\ell}\}$ is a $\ell$-clique in $G$. We construct
    the homomorphism $h_C : V(F) \to V(G')$  with $h_C(a_i) = u_{i, j_i}$, $h_C(b_i) = v_{i,
    j_i}$ and $h_C(x_k) = y_k$. Observe that this defines a color-prescribed homomorphism $h_C
    \in \cphoms{F}{G'}$.

    Next, consider an $h' \in \cphoms{F}{G'}$. Then, $h'(x_k) = y_k$ since $y_k$ is the
    only vertex in $G'$ with $c(y_k) = x_k$. Further, we obtain $h'(a_i) = u_{i, \alpha_i}$
    and $h'(b_j) = w_{j, \beta_j}$. Observe that for all $i \in \setn{\ell}$ the edge $\{a_i,
    b_i\}$ is in $F$. Thus $\{h'(a_i), h'(b_i)\} = \{u_{i, \alpha_i}, w_{i, \beta_i}\}$ is an
    edge in $G'$ which is only possible if $\alpha_i = \beta_i$.

    Next, we define the tuple $C
    \coloneqq (v_{\alpha_1}, \dots, v_{\alpha_\ell})$. For all $i < i'$, we know that the edge $\{a_i,
    b_{i'}\}$ is in $F$, hence the edge $\{u_{i, \alpha_i}, w_{i', \beta_{i'}}\}$ is
    also in $G'$ which implies that $\alpha_i < \beta_{i'} = \alpha_{i'}$. Observe that
    this also implies that there is an edge between $v_{\alpha_i}$ and $v_{\alpha_{i'}}$.
    So, the
    indices of the tuple $C$ are ordered and $\{v_{\alpha_1}, \dots,
    v_{\alpha_\ell}\}$ is a $\ell$-clique in $G$. Also, observe $h' = h_C$, where $h_C$ is
    the color-prescribed homomorphism from above. Thus each ordered $\ell$-cliques $C$
    yields a color-prescribed homomorphism $h_C$ and each color-prescribed homomorphism
    $h'$ yields an ordered $\ell$-clique $C$. This shows an one-to-one correspondence
    between color-prescribed homomorphisms in $\cphoms{F}{G'}$ and $\ell$-cliques in $G$.
\end{proof}

This means that we can solve the decision problem $k$-$\textsc{Clique}$ for a graph $G$
and a parameter $k$ by computing $\NUM{}\cphoms{F}{\star}$. Further, if we assume that $F$
is non-vanishing then we can use the reduction shown in \cref{lem:cphom:to:indsub} to
compute $\NUM{}\cphoms{F}{\star}$ using $\NUM{}\indsubs{(\Phi, |V(F)|)}{\star}$. Thus
we can use $\NUM{}\indsubs{(\Phi, |V(F)|)}{\star}$ to solve $k$-$\textsc{Clique}$.

\thmlowerboundbicliques
\begin{proof}
    We show how to
    use $\NUM{}\indsubs{(\Phi, k)}{\star}$ to solve $h(k)$-\textsc{Clique}.

    \begin{claim}\label{claim:clique:to:indsub}
        For a fixed $k$ such that there is a
        graph $F$ with $k$ vertices, $\ae{\Phi}{F} \neq 0$, and $F$ contains $K_{h(k),
        h(k)}$ as a subgraph; if there is an algorithm that computes for each graph $G'$
        the value $\NUM{}\indsubs{(\Phi,
        k)}{G'}$ in time $O(|V(G')|^{\gamma})$, then $h(k)$-\textsc{Clique} can be
        computed for each graph $G$ in time
        $O(|V(G)|^{\gamma + 2})$.
    \end{claim}
    \begin{claimproof}
        We construct an algorithm
        that solves $h(k)$-\textsc{Clique} using an oracle for $\NUM{}\indsubs{(\Phi,
        j)}{\star}$. Since $k$ and $F$ are fixed, we can assume that our algorithm knows
        these elements. Fix a graph $G$. Then we use the algorithm from
        \cref{lem:tight:biclique:construction} to construct a graph $G'$ such that $G$
        contains an $h(k)$-clique if and only if $\NUM{}\cphoms{F}{G'}$ is not zero. The
        running time of
        this construction is in $O(|V(G)|^2)$ since $|V(F)| = k$ is constant. Further,
        we obtain $|V(G')| \leq 2 h(k) |V(G)| + k = O(|V(G)|)$.

        If we can compute $\NUM{}\indsubs{(\Phi, k)}{|V(G')|}$ in time $O(|V(G')|^{\gamma})$, then we
        can use \cref{lem:cphom:to:indsub} to compute $\NUM{}\cphoms{F}{G'}$ in time $O((2
        h(k) |V(G)| + k)^{\gamma + 2})$. Thus, we can solve $h(k)$-\textsc{Clique} in time
        $O(|V(G)|^2 + (2 h(k) |V(G)| + k)^{\gamma + 2})$ which is in $O(|V(G)|^{\gamma + 2})$ since $k$
        is fixed.
    \end{claimproof}
    According to \cref{theo:k-clique:ETH}, there is a constant $\alpha >
    0$ such that no algorithm solves $h(k)$-\textsc{Clique} in time
    $O(|V(G)|^{\alpha h(k)})$ for a fixed $h(k) \geq 3$ unless ETH fails. Set $\beta \coloneqq
    \alpha/2$ and $N \coloneqq \max(3, 4/\alpha)$.
    If there are a $k$ with $h(k) \geq N \geq 3$,
    a graph $F$ with $k$ vertices such that $\ae{\Phi}{F} \neq 0$,
    and $F$ contains $K_{h(k), h(k)}$ as a subgraph, and an algorithm that solves $\NUM{}\indsubs{(\Phi,
    k)}{G'}$ in time $O(|V(G')|^{\beta h(k)})$, then we can use \cref{claim:clique:to:indsub}
    to solve $h(k)$-\textsc{Clique} in time $O(|V(G)|^{\beta h(k) + 2})$. Observe that $\beta
    h(k) + 2 \leq \alpha h(k)$ for $h(k) \geq 4/\alpha$. Thus, we can use solve
    $h(k)$-\textsc{Clique} in time $O(n^{\alpha h(k)})$. Hence, ETH fails.
\end{proof}

\end{document}